\documentclass[11pt]{article}
\usepackage{amsmath}
\usepackage{graphicx}
\usepackage{psfrag}
\usepackage{epsf}
\usepackage{enumerate}
\usepackage[round]{natbib}
\usepackage{url} % not crucial - just used below for the URL
\usepackage{amsthm}
\usepackage{xpatch}
\usepackage{subcaption}

\usepackage{longtable}
\usepackage{mathrsfs}
\usepackage{tikz}
\usepackage{color}
\usepackage{amssymb}
\usepackage{latexsym}
\usepackage{xr}
\usepackage{epstopdf}
\usepackage{float}
\usepackage{multirow}
\usepackage{longtable}

\usepackage{bbm}

\numberwithin{equation}{section}

\def\RR{\mathbb R}
\def\ZZ{\mathbb Z}

\def\EE{\mathbb E}
\def\PP{\mathbb P}

\def\PP{\mathbb P}

\newtheorem{theorem}{Theorem}[section]
\newtheorem{proposition}{Proposition}[section]
\newtheorem{lemma}{Lemma}[section]
\newtheorem{corollary}{Corollary}[section]

\theoremstyle{definition}
\newtheorem{remark}{Remark}[section]
\newtheorem{example}{Example}[section]
\newtheorem{definition}{Definition}[section]

\makeatletter
\xpatchcmd{\proof}{\@addpunct{.}}{\@addpunct{:}}{}{}
\makeatother

\begin{document}

\def\spacingset#1{\renewcommand{\baselinestretch}%
{#1}\small\normalsize} \spacingset{1}

%%%%%%%%%%%%%%%%%%%%%%%%%%%%%%%%%%%%%%%%%%%%%%%%%%%%%%%%%%%%%%%%%%%%%%%%%%%%%%

\title{Stationary subspace analysis of nonstationary covariance processes: eigenstructure description and testing
\footnote{The second author was supported in part by the National Science Foundation grant DMS-1712966.   The work of
the third author was supported by the National Science Foundation grant DMS-1612984.}}

\author{
Raanju R.\ Sundararajan\\ Southern Methodist University \\ King Abdullah University of Science and Technology \and Vladas Pipiras\\ University of North Carolina
 \and  Mohsen Pourahmadi\\ Texas A\&M University }
\date{\today}

\maketitle

\bigskip

\begin{abstract}

\noindent Stationary subspace analysis (SSA) searches for linear combinations of the components of nonstationary vector time series that are stationary. These linear combinations and their number define an associated stationary subspace and its dimension. SSA is studied here for  zero mean nonstationary covariance processes. We characterize stationary subspaces and their dimensions in terms of eigenvalues and eigenvectors of certain symmetric matrices. This characterization is then used to derive formal statistical tests for estimating dimensions of stationary subspaces. Eigenstructure-based techniques are also proposed to estimate stationary subspaces, without relying on previously used computationally intensive optimization-based methods. Finally, the introduced methodologies are examined on simulated and real data.

\end{abstract}

\section{Introduction}
\label{s:intro}

The goal of this work is to provide new fundamental insights into the so-called stationary subspace analysis (SSA), a technique for finding linear combinations of  components of a multivariate time series that are stationary. More precisely, consider a setup where  
 the observed $p$-vector nonstationary  time series $ X_{t} $ is  a linear transformation of a $d$-vector stationary series
 $U_{t}^{s}$  and a $(p-d)$-vector nonstationary series  $U_t^{n}$ through
\begin{equation}\label{e:intro-ssamodel}
X_t = M U_t = \left[ \begin{array}{c|c}
M_s & M_n
\end{array} \right] \; \begin{bmatrix}
U_t^s \\
U_t^{n}
\end{bmatrix},
\end{equation}
and $M$ is an unknown $p \times p$ (invertible) mixing matrix, $M_s$ and $M_n$ are
 $p \times d $ and $ p \times (p-d)$ matrices, respectively. It is further assumed that no linear
  transformation of $U_t^n$ is stationary. Given the data  $X_1,X_2,\hdots, X_T$,
   SSA seeks to find the demixing matrix $B = (M^{-1})'$ so that $B' X_t = U_t$
   is naturally partitioned into its
    stationary and nonstationary sources. The space spanned by the first $d$ columns of $B$ is referred to as a stationary subspace and $d$ as its dimension.

SSA was introduced and studied by \cite{ssa09}, with applications to analyzing EEG data from neuroscience experiments. In that work, the observed vector time series is assumed to be independent across time and the notion of stationarity is with respect to the first two moments, that is, the mean and lag-0 covariance are required to be time invariant.  The demixing matrix in SSA
   is found in the spirit of ANOVA by  dividing
     the observed time series data into $N$ segments and minimizing a Kullback-Leibler (KL) divergence between
     Gaussian distributions measuring differences in the means and covariances across these segments.
        A sequential
          likelihood ratio test is used  in \citet{ssa09} and  \citet{SSAchangepoint} to determine
          the dimension of the stationary subspace $d$
          under the additional assumption of  normality of the data.
 The frequency domain or dependent SSA (DSSA)  in \cite{sundararajan:2017} avoids dividing the data into segments and relies on the approximate
   uncorrelatedness of the discrete Fourier transform (DFT) of a second-order
    stationary time series at Fourier frequencies. The sum of the Frobenius norms of the estimated
        covariances of the DFTs at the first few  lags is used as a discrepancy measure  and the demixing matrix is obtained by optimizing this measure. Then, a sequential test of second-order stationarity is used
           to determine $d$  and the
            consistency of the estimated $d$ is discussed using the asymptotic
           distribution of the test statistic under the alternative hypothesis of local stationarity of the time series  (\citet{dahlhaus97,dahlhaus12}).

Overall, the research thus far suggests the need for a better mathematical formulation and understanding of the problem. A reformulation and, in particular, an optimization-free solution seems to be the key in finding a transparent and interpretable solution of the SSA problem. In this work, we shall tackle these issues for a special but general case of \eqref{e:intro-ssamodel}, namely, that of zero mean vectors $U_t$, assuming $U_t = C(t)Y_t$ with i.i.d.\ zero mean vectors $Y_t$. We shall further write this model formulation as

\begin{equation}\label{e:intro-vcmodel}
X_t = A \Big( \frac{t}{T} \Big) Y_t, \quad t=1,2, \hdots ,T,
\end{equation}
where $T$ is the sample size and the time dependence is brought into $A(\frac{t}{T}) = M \cdot C(t)$. Additional assumptions on the matrix-valued function $A:(0,1 ) \rightarrow \mathbb{R}$ and the i.i.d.\ vectors $Y_t$ can be found below in Section \ref{s:models}. The modification \eqref{e:intro-vcmodel} of \eqref{e:intro-ssamodel} takes the heterogeneity out of $U_t$ and places it into the deterministic matrix-valued function $A(\cdot)$. The nonstationary covariance process \eqref{e:intro-vcmodel} will be said to follow a \textit{varying covariance} (\textit{VC}) model. The assumption of zero mean in \eqref{e:intro-vcmodel} is made for several reasons. For one, all previous works that identify EEG data analysis as the motivating application involve this assumption. We are currently working in parallel on analogous approaches to the SSA problem for time-varying means (\citet{ssa-means}), and will possibly look at the combined model in the future. In the latter regard, we should also note that dealing with SSA for varying covariances is seemingly much more involved than for varying means. Indeed, as seen from this work, the SSA for the VC model has a surprisingly rich structure.

Our contributions to the SSA for the VC model \eqref{e:intro-vcmodel} are as follows. First, by using basic ideas from linear algebra, we provide an interpretation of a stationary subspace and its dimension $d$ in terms of eigenvalues and eigenvectors of certain symmetric matrices. This interpretation, in fact, is given at two levels: ``local'' or for fixed $u \in (0,1)$, and ``global'' or for $(0,1) = \{ u :\; u \in (0,1) \}$, where $u$ is thought here as a variable of $A(u)$ replaced by $t/T$ in \eqref{e:intro-vcmodel}. Second, in the context of the obtained interpretation, we develop formal statistical tests for both ``local'' and ``global'' dimensions of stationary subspaces. Together with the algebraic interpretation, these tests are the key theoretical  contributions of this work. Third, by leveraging the new interpretation of stationary subspaces, we provide more direct and algebraic ways to construct them. These are shown to outperform the computationally more expensive optimization-based solutions of the previous SSA approaches in a number of simulation settings. We should also note that the proposed dimension tests in Section \ref{s:inf-dimens} assume the existence of a common stationary subspace and dimension across the ``local'' levels (see Section \ref{s:char-implic} for more details); testing for the latter remains an open problem. Furthermore, this work concerns the asymptotics under $T \rightarrow \infty$ with a fixed $p$. Fourth, we revisit an SSA application to EEG data from a Brain-Computer Interface (BCI) experiment and provide additional insights by using the proposed methods.

The outline of the paper is as follows. Section \ref{s:models} reintroduces more formally the VC model and its stationary subspace and dimension. Section \ref{s:char} gives an eigenstructure-based characterization of a stationary subspace and its dimension. Section \ref{s:inf-dimens} considers statistical tests for the dimension of a stationary subspace at both the ``local'' and ``global'' levels. Section \ref{s:inf-subspace} discusses estimation  methods for stationary subspaces that are based on algebraic constructs and do not involve iterative and computationally heavy optimization methods. Sections \ref{s:simulation} and \ref{s:application} illustrate the proposed methodologies using simulated and real data. Section 8 concludes. 

%Finally, Section \ref{s:conclusion} provides the concluding remarks along with some open questions.

\section{Model of interest and its stationary subspace}
\label{s:models}

We focus throughout this work on the varying covariance (VC) model \eqref{e:intro-vcmodel}, where $ A:(0,1) \rightarrow \mathbb{R}^{p \times p}$ is a smoothly varying matrix-valued function and  $Y_t$ are i.i.d.\ random vectors with i.i.d.\ entries, $\EE(Y_t)=0$ and $\EE(Y_t Y_t') = I_p$.
% {\bf Taking the first few rows of $b()$ to be constant leads to the setup in (1).}
Further technical assumptions can be found below.

\begin{definition}\label{d:model-vc-ss}
If $d$ is the largest integer in $\{0,1,2,\hdots,p\}$ for which there is a $p \times d$ matrix $B_1$ such that
\begin{equation}\label{e:model-vc-ss}
	B_1' A^2(u) B_1 = \Sigma, \; \; \forall \; u\in (0,1),
\end{equation}
where  $A^2(u) = A(u) A(u)'$ , $\Sigma$ does not depend on $u$ and $B_1' B_1 = I_d$,  then the space ${\cal B}_1$ spanned by the columns of $B_1$ will be called a (second-order) \textit{stationary subspace} of dimension $d$ of the model \eqref{e:intro-vcmodel}.
\end{definition}

Note that (\ref{e:model-vc-ss}) states effectively that the covariance matrix of $B_1' X_t$ does not depend on $t$. It
 %We exclude the case $ u = 1 $ from (\ref{e:model-vc-ss}) for simplicity, in order not to have to deal with boundary effects below. In what follows, the subspace $\mathcal{B}_1$ and its dimension $d$ defined above in Definition \ref{d:model-vc-ss}  will be sometimes be referred to as \textit{global stationary subspace} and \textit{global dimension}.
%
%\begin{remark}\label{r:A-interprete}
%Note that $A^2(u)$ can be thought as a ``local'' covariance matrix of $X_t$ when $t/T$ is close to $u$. From this perspective, we shall think of $A(u)$ in the model \eqref{e:intro-vcmodel} as the unique positive-definite square root of $A^2(u)$ (assuming positive definiteness of the latter). Being positive-definite, $A(u)$ is also symmetric. By writing \eqref{e:intro-vcmodel}, we thus mean that $A(t/T)^{-1}X_t=Y_t$ with the specified assumptions on $Y_t$, when $A(u)$ is exactly the unique positive-definite square root of $A^2(u)$ (and not some other $A(u)$).
%\end{remark}
%
can be reformulated as follows: Let $\overline{A}^2 = \int_0^1 A^2(u) du$ and define a $p \times p$ symmetric matrix $M(u)$ as
\begin{equation} \label{e:models-M}
M(u) = A^2(u) - \overline{A}^2.
\end{equation}
Then, the condition \eqref{e:model-vc-ss} is equivalent to
\begin{equation}\label{e:model-vc-ss-2}
	B_1' M(u)  B_1  = 0, \; \; \forall \; u \in (0,1).
\end{equation}

\noindent Indeed, \eqref{e:model-vc-ss} implies \eqref{e:model-vc-ss-2} after integrating \eqref{e:model-vc-ss} over $u \in (0,1)$, noting that $\int_0^1 \Sigma du = \Sigma$ and subtracting the two sides of the resulting relation $B_1' \overline{A}^2 B_1 = \Sigma$ from those of \eqref{e:model-vc-ss}. Similarly \eqref{e:model-vc-ss-2} implies \eqref{e:model-vc-ss} with $\Sigma =B_1' \overline{A}^2 B_1 $. The matrix $M(u)$ will play a central role henceforth.

Our approach to finding a matrix $B_1$ and the corresponding stationary subspace $\mathcal{B}_1$ and dimension $d$ is based on the relation \eqref{e:model-vc-ss-2} for a fixed $u$, that is, $B_1' M(u) B_1 = 0$ for a fixed $u$ and $B_1 = B_1(u)$ of dimension $d = d(u)$. As shown in the next section, for a fixed $u$, the matrix $B_1(u)$ and its dimension $d(u)$ can be characterized using the eigenstructure of the matrix $M(u)$. When it comes to $M(u)$, we shall be using the terminology of the following definition.

\begin{definition} \label{d:model-vc-ss-local}

Let $B_1 = B_1(u)$ be a matrix with $d=d(u)$ columns that satisfies \eqref{e:model-vc-ss-2} for a fixed $u \in (0,1)$. The space ${\cal B}_1(u)$ spanned by the columns of $B_1(u)$ will be called a \textit{local stationary subspace} of \textit{local dimension} $d(u)$.  The respective quantities in Definition 2.1 will be referred to as a \textit{global stationary subspace} and a \textit{global dimension}.
\end{definition}

Relationships between local and global stationary subspaces and their dimensions are discussed in Section \ref{s:char-implic}.

% Analogous to \eqref{e:model-vc-ss-2}, we have the following condition equivalent to Definition \ref{d:model-vc-ss},
% \begin{equation*}
% B_1 (A(u) - \bar{A}) = 0, \; \; \forall u \in (0,1),
% \end{equation*}
% where  $\bar{A} = \int_{0}^1 A(u)du$. 

\section{Matrix pseudo nullity and pseudo null space}
\label{s:char}

In this section we describe an eigenstructure-based characterization of a stationary subspace and its dimension. We first define the notions of \textit{pseudo null space} and \textit{psuedo nullity} for symmetric matrices and then later, in Section \ref{s:char-implic}, associate them to stationary subspaces and their dimensions. In view of the relation \eqref{e:model-vc-ss-2}  and Definition \ref{d:model-vc-ss-local}, we start with the following definition.

\begin{definition} \label{d:descr-psudo-nullity}
Let $M$	be any $p\times p$ symmetric matrix. A {\it pseudo nullity} of $M$, denoted by $d(M)$, is defined as the largest positive number $d_1$ such that
\begin{equation}\label{e:M-pdim}
	C_1' M C_1 = 0,
\end{equation}
for a $p\times d_1$ matrix $C_1$ with $C_1' C_1 = I_{d_1}$. A {\it pseudo null space} of $M$, denoted as ${\cal P}(M)$, is defined as
the linear span of the $d_1$ columns of the matrix $C_1$ in (\ref{e:M-pdim}). A column of $C_1$, that is, a column vector $s$ such that $s'Ms = 0$ will be
called a {\it pseudo eigenvector}.
\end{definition}

If $M$ is positive semi-definite, note that its pseudo nullity is its nullity (i.e.\ the number of zero eigenvalues of $M$) and its pseudo null space is its null space; thus, the psuedo- terminology is used to draw attention to the contrast between these two cases. Otherwise, we should caution the reader against drawing other parallels between the two contexts. For example, if $s_1$ and $s_2$ are two pseudo eigenvectors (which can be either orthogonal or non-orthogonal), note that there is a priori no reason to have $s_1'Ms_2=0$ and hence $C_1' M C_1 = 0$ with $C_1= (s_1\, s_2)$. In particular, if for e.g.\ $d(M)=2$, and $s_1$ and $s_2$ are orthogonal, the linear space spanned by $s_1$ and $s_2$ is not necessarily a pseudo null space.

Another word of caution is that ${\cal P}(M)$ is not unique in general. This, perhaps surprising, fact will be explained below. By writing ${\cal P}(M)$, we mean one of the pseudo null spaces.

\subsection{Characterization of pseudo nullity}
\label{s:char-pdim}

We characterize here the pseudo nullity $d(M)$ of a symmetric matrix $M$ in terms of its inertia.  More precisely, let
\begin{equation}\label{e:M-0+-}
	d_0 = d_0(M), \quad d_{+} = d_+(M), \quad d_{-} = d_-(M)
\end{equation}
be the number of zero, positive and negative eigenvalues of $M$, respectively. Part of the proof of the following result is constructive and will be used to construct pseudo null spaces (stationary subspaces) in Section \ref{s:inf-subspace}.

\begin{proposition}\label{p:rank_claim_1}
Let $M$ be a symmetric matrix. Then, $d(M) = d_0 +  \min( d_+, d_-)$.
% \begin{equation} \label{e:rank_claim_1}
% \mbox{pdim}\{M\} \leq \mbox{zero}\{ M \} +  \min( \mbox{pos}\{ M \}, \mbox{neg} \{ M \}  ).
% \mbox{pdim}\{M\} \leq n_0 +  \min( n_+, n_-).
% \end{equation}
\end{proposition}

\begin{proof}
By the Poincare Separation Theorem (e.g.\ \citet{magnus:neudecker:1999}),
\begin{equation*}
\lambda_i \leq \mu_i \leq \lambda_{n-d_1+i}, \quad i=1,\ldots,d_1,
\end{equation*}
where $\mu_1 \leq \ldots \leq \mu_{d_1}$ are the ordered eigenvalues of $C_1'MC_1$ for any $p\times d_1$ matrix $C_1$ such that $C_1' C_1=I_{d_1}$. Taking $C_1$ as in (\ref{e:M-pdim}) with $d_1=d(M)$, we have $C_1' M C_1 = 0$, and hence
\begin{equation*}
\lambda_i \leq 0 \leq \lambda_{n-d_1+i}, \quad i=1,\ldots,d_1.
\end{equation*}
This shows that $d(M) \leq d_0 + \min(n_+ , n_-)$.

To prove the reverse inequality $d_0 + \min(n_+ , n_-) \leq d(M)$, let $d_* =  d_0 +  \min( d_+, d_-)$. The rest of the proof constructs a $p \times d_*$ matrix $C_1$ such that $C_1' M C_1 = 0$ and $C_1' C_1=I_{d_*}$, which yields the desired inequality. By the Schur decomposition (e.g.\ \citet{magnus:neudecker:1999}),
\begin{equation}\label{e:M-schur}
S'MS = \mbox{diag}(\lambda_1,\lambda_2,\ldots,\lambda_p),	
\end{equation}
where the columns $s_1, \ldots,s_p$ of $S$ are orthonormal and represent the eigenvectors  associated with the negative, positive and zero eigenvalues of $M$. We now need to separate the eigenvectors into those associated with the eigenvalues $\lambda_1,\ldots,\lambda_p$ of $M$. Let $s_{0,i}$ be the eigenvectors associated with the zero eigenvalues $\lambda_{0,i}$,  $i=1,\ldots,d_0$, $s_{+,i}$ be the eigenvectors associated with the positive eigenvalues $\lambda_{+,i}$, $i=1,\ldots,d_{+}$, and $s_{-,i}$ be the eigenvectors associated with the negative eigenvalues $\lambda_{-,i}$, $i=1,\ldots,d_-$. Note by (\ref{e:M-schur}) that
\begin{equation}\label{e:M-schur-2}
s_i's_j = \delta_{ij}, \quad  s_i'Ms_j = \delta_{ij} \lambda_{i},
\end{equation}
where $\delta_{ij}=1$ if $i=j$ and $0$ otherwise.

For the zero eigenvalues and the corresponding eigenvectors, we have $M s_{0,i} = 0 \cdot s_{0,i}$ and hence
\begin{equation}\label{e:M-zero-eig}
	s_{0,i}' M s_{0,i} = 0,\quad  i=1,\ldots,d_0.
\end{equation}
Similarly, for the positive and negative eigenvalues, we have
\begin{eqnarray}
	s_{+,i}' M s_{+,i} & = &  \lambda_{+,i}>0,\quad i=1,\ldots,d_{+}, \notag \\
	s_{-,i}' M s_{-,i} & = &  \lambda_{-,i}<0,\quad i=1,\ldots,d_{-}. \notag
\end{eqnarray}

Let $d_\pm = \min(d_{+},d_{-})$. Then, for some $\alpha_i \in (0,1)$ and $\beta_i = 1-\alpha_i$, we have
\begin{equation}\label{e:M-posneg-eig=0}
(\alpha_i^{1/2}s_{+,i} + \beta_i^{1/2}s_{-,i})' M (\alpha_i^{1/2}s_{+,i} + \beta_i^{1/2}s_{-,i})
= \alpha_i(s_{+,i}' M s_{+,i}) + \beta_i (s_{-,i}' M s_{-,i}) = 0,
\end{equation}
for $i=1,\ldots,d_\pm$. Set

\begin{equation} \label{e:M-posneg-eig}
s_{\pm,i} = \frac{\alpha_i^{1/2}s_{+,i} +  \beta_i^{1/2}s_{-,i}}{|| \alpha_i^{1/2}s_{+,i} + \beta_i^{1/2}s_{-,i} ||_2},\quad i=1,\ldots,d_\pm,
\end{equation}
and define a $(d_0+d_\pm) \times p = d_* \times p$ matrix $C_1$ as

\begin{equation*}
C_1 = (s_{0,1}\ \ldots\ s_{0,d_0}\ s_{\pm,1}\ \ldots \ s_{\pm,d_\pm}).
\end{equation*}

By using (\ref{e:M-schur-2}), (\ref{e:M-zero-eig}) and (\ref{e:M-posneg-eig=0}), we have $C_1' M C_1 = 0$. Since $s_i$ are orthonormal and in view of (\ref{e:M-posneg-eig}), we also have
$C_1' C_1 = I_{d_*}$. This shows that $ d_0 + d_{\pm}  \leq d(M)$ and concludes the proof.
\end{proof}

\subsection{Characterization of pseudo null space}
\label{s:char-pnull}

A closer examination of the proof of Proposition \ref{p:rank_claim_1} also suggests further insights into the structure of a pseudo null space ${\cal P}(M)$. The next auxiliary result describes an element of ${\cal P}(M)$. As in the proof of Proposition \ref{p:rank_claim_1}, we let $s_i$ ($s_{0,i}$, $s_{+,i}$ and $s_{-,i}$, resp.) be the orthonormal eigenvectors (associated with the zero, positive and negative eigenvalues, resp.) of $M$. The corresponding eigenvalues are denoted $\lambda_i$ ($\lambda_{0,i}=0$, $\lambda_{+,i}$, $\lambda_{-,i}$ resp.). We also let ${\cal N}_0 (M)$ denote the null space of the matrix $M$, that is, the linear space spanned by the eigenvectors $s_{0,i}$, $i=1,\ldots,d_0$.

\begin{lemma}\label{l:M-null}
Let $M$ be a symmetric matrix and $w$ be a pseudo eigenvector. Then, 	
\begin{equation}\label{e:M-null-decomp}
	w = w_0 + w_\pm,
\end{equation}
where $w_0\in {\cal N}_0 (M)$ and
\begin{equation}\label{e:M-null-decomp-1}
	w_\pm = \sum_{i=1}^{d_+} \alpha_{+,i} s_{+,i} + \sum_{i=1}^{d_-} \alpha_{-,i} s_{-,i},
\end{equation}
where $\alpha_{-,i},\alpha_{+,i}\in\RR$ are such that
\begin{equation}\label{e:M-null-decomp-2}
	\sum_{i=1}^{d_+} \alpha_{+,i}^2 \lambda_{+,i}  + \sum_{i=1}^{d_-} \alpha_{-,i}^2 \lambda_{-,i} = 0.
\end{equation}
Conversely, a vector $w$ expressed through (\ref{e:M-null-decomp})--(\ref{e:M-null-decomp-2}) is a pseudo eigenvector.
\end{lemma}

\begin{proof}
Any vector $w$ can be expressed as a linear combination of the basis vectors $s_i$ as in (\ref{e:M-null-decomp})--(\ref{e:M-null-decomp-1}). The relation (\ref{e:M-null-decomp-2}) follows since
$$
0 = w' M w = \sum_{i=1}^{d_0} \alpha_{0,i}^2 s_{0,i}'Ms_{0,i}  + \sum_{i=1}^{d_+} \alpha_{+,i}^2 s_{+,i}'Ms_{+,i} + \sum_{i=1}^{d_-} \alpha_{-,i}^2 s_{-,i}'Ms_{-,i}
= \sum_{i=1}^{d_+} \alpha_{+,i}^2 \lambda_{+,i} + \sum_{i=1}^{d_-} \alpha_{-,i}^2 \lambda_{-,i}.
$$
The converse statement follows similarly.
\end{proof}

The next lemma clarifies when two pseudo eigenvectors belong to a pseudo null space. See also the discussion following Definition \ref{d:descr-psudo-nullity}.

\begin{lemma}\label{l:M-null-2}
Let $M$ be a symmetric matrix and $w_k$, $k=1,2$, be its two pseudo eigenvectors expressed through (\ref{e:M-null-decomp-1})--(\ref{e:M-null-decomp-2}) with the coefficients $\alpha_{k,+,i}$ and $\alpha_{k,-,i}$, $k=1,2$. Then, $w_k$, $k=1,2$, belong to the same pseudo null space if and only if
\begin{equation}\label{e:M-null-decomp-3}
	\sum_{i=1}^{d_+} \alpha_{1,+,i}\alpha_{2,+,i} \lambda_{+,i}  + \sum_{i=1}^{d_-} \alpha_{1,-,i} \alpha_{2,-,i} \lambda_{-,i} = 0.
\end{equation}
\end{lemma}

\begin{proof}
The result follows as in the proof of Lemma \ref{l:M-null} above from requiring that $w_1'Mw_2=0$.
\end{proof}

The next result characterizes a pseudo null space ${\cal P}(M)$.

\begin{proposition}\label{p:M-null}
Let $M$ be a symmetric matrix and ${\cal P}(M)$ be its pseudo null space. Then,
\begin{equation}\label{e:M-null}
	{\cal P}(M) = {\cal N}_0(M) \oplus {\cal N}_\pm (M),
\end{equation}
where ${\cal N}_\pm (M)$ is a linear space spanned by $\min(d_-,d_+)$ orthogonal pseudo eigenvectors expressed as (\ref{e:M-null-decomp-1})--(\ref{e:M-null-decomp-2}) and satisfying (\ref{e:M-null-decomp-3}) pairwise. Conversely, the right-hand side of (\ref{e:M-null}) defines a pseudo null space. Moreover, a pseudo null space is not unique in general.
\end{proposition}

\begin{proof}
The first two statements follow from Lemmas \ref{l:M-null}--\ref{l:M-null-2} and Proposition  \ref{p:rank_claim_1}. The last statement is illustrated in Example \ref{ex:Mu-analysis} given below.
\end{proof}

\subsection{Implications for stationary subspace and its dimension}
\label{s:char-implic}

In view of Definitions \ref{d:model-vc-ss} and \ref{d:model-vc-ss-local} and their notation, the global stationary subspace $\mathcal{B}_1$ and its dimension $d$ are given by:
\begin{equation}\label{e:model-vc-dB1}
	d = d(u) = d(M(u)),\quad {\cal B}_1 = {\cal B}_1(u) =  {\cal P}(M(u)),\quad  \forall \; u\in (0,1).
\end{equation}

\noindent In view of Proposition \ref{p:rank_claim_1}, the first relation in (\ref{e:model-vc-dB1}) can now be reformulated as
\begin{equation}\label{e:model-vc-d-dim}
	d = d(u) = d(M(u)) = d_0(u) +  \min( d_+(u), d_-(u)),\quad  \forall \; u\in (0,1),
\end{equation}
where
\begin{equation}\label{e:model-vc-d-dim-2}
	d_0(u) = d_0( M(u)), \quad  d_+(u)= d_+( M(u) ), \quad d_-(u) = d_-( M(u) )
\end{equation}
are the numbers of zero, positive and negative eigenvalues of $M(u)$, respectively.

On the other hand, if
\begin{equation} \label{e:model-vc-d-dim-3r}
d(M(u)) = d_0(u) + \min(d_+(u) , d_-(u)) \equiv d^{*}, \;\; u \in (0,1),
\end{equation}
for some $d^{*}$, this does not necessarily mean that the dimension $d$ of the stationary subspace is $d^{*}$. The latter is because \eqref{e:model-vc-d-dim-3r} does not guarantee that
 \begin{equation} \label{e:model-vc-d-dim-4r}
 {\cal B}_1 (u) = {\cal P}(M(u) ) \equiv {\cal B}_1^{*}, \quad u \in (0,1),
 \end{equation}
for some ${\cal B}_1^{*}$, playing the role of a stationary subspace. But if \eqref{e:model-vc-d-dim-4r} holds, then \eqref{e:model-vc-d-dim-3r} does imply that $d^{*}$ is the dimension $d$ of the stationary subspace. The statistical tests developed in the subsequent sections will, in fact, be for testing \eqref{e:model-vc-d-dim-3r} and hence will lead to the dimension $d$ assuming \eqref{e:model-vc-d-dim-4r}. How testing can be done for \eqref{e:model-vc-d-dim-4r} remains an open question, though we shall also suggest new ways to estimate ${\cal B}_1$, based on developments in this section.

Finally, we illustrate the observations made above through a simple but instructive example. See also a subsequent remark.

\begin{example}\label{ex:Mu-analysis}
Consider a VC model with
\begin{equation}\label{e:ex-b2u}
	A^2(u) = \mbox{diag}(2+\sin(2\pi u),3-\sin(2\pi u),1+\sin(2\pi u)).
\end{equation}
Then, $\overline{A}^2 = \mbox{diag}(2,3,1)$ and
\begin{equation}\label{e:ex-b2u-bar}
	M(u) = A^2(u) - \overline{A}^2 = \sin(2\pi u)\cdot \mbox{diag}(1,-1,1).
\end{equation}
For fixed $u\neq 1/2$, the eigenvalues of $M(u)$ are $\sin(2\pi u)\cdot 1$ (of multiplicity 2) and $\sin(2\pi u)\cdot (-1)$. For further illustration, suppose $u\in (0,1/2)$, so that $\sin(2\pi u)>0$. Then, $\lambda_{+,1} = \lambda_{+,2} = \sin(2\pi u)\cdot 1$, $\lambda_{-,1} = \sin(2\pi u)\cdot (-1)$ and $d_+=2$, $d_-=1$, $d_0=0$, by using the notation in Section \ref{s:char-pnull} with $M = M(u)$. By Proposition \ref{p:rank_claim_1}, $d(u) = d(M(u)) = 0 + \min(1,2) =1$. The corresponding eigenvectors are $s_{+,1}= (1\ 0\ 0)'$, $s_{+,2}= (0\ 0\ 1)'$ and $s_{-,1}= (0\ 1\ 0)'$. By Proposition \ref{p:M-null}, a local stationary subspace or a pseudo null space of $M(u)$ can be expressed as
$$
{\cal B}_1(u) = {\cal P}(M(u)) = \mbox{lin}\{ \alpha_{-,1} s_{-,1} + \alpha_{+,1} s_{+,1} + \alpha_{+,2} s_{+,2} \}
$$
such that
$$
-\alpha_{-,1}^2 + \alpha_{+,1}^2 + \alpha_{+,2}^2  = 0
$$
and $\alpha_{-,1}^2 + \alpha_{+,1}^2 + \alpha_{+,2}^2=1$ for the norm to be $1$, where ``lin'' indicates a linear span. The latter two expressions yield $\alpha_{-,1} = (\alpha_{+,1}^2 + \alpha_{+,2}^2)^{1/2}$ (after choosing a positive sign for the square root) and $\alpha_{+,1}^2 + \alpha_{+,2}^2=1/2$. This further yields $\alpha:= \alpha_{+,1}\in [1/\sqrt{2},-1/\sqrt{2}]$, $\alpha_{+,2} = \pm (1/2 - \alpha^2)^{1/2}$ and $\alpha_{-,1} = 1/\sqrt{2}$. Thus, a pseudo null space can also be expressed as
\begin{equation}\label{e:ex-null}
	{\cal P}(M(u)) =  \mbox{lin}\{ (1/\sqrt{2}) s_{-,1} + \alpha s_{+,1} \pm (1/2 - \alpha^2)^{1/2} s_{+,2} \} =  \mbox{lin}\{ (\alpha,\ 1/\sqrt{2},\ \pm (1/2 - \alpha^2)^{1/2})' \},
\end{equation}
where $\alpha \in [1/\sqrt{2},-1/\sqrt{2}]$. Note that these spaces (vectors) are generally different for different $\alpha$'s. For example, for $\alpha=0$,
\begin{equation}\label{e:ex-null-1}
	{\cal P}(M(u)) =  \mbox{lin}\{ (0,\ 1/\sqrt{2},\ 1/\sqrt{2})' \} = \mbox{lin}\{ (0\ 1\ 1)' \}
\end{equation}
and for $\alpha=1/\sqrt{2}$,
\begin{equation}\label{e:ex-null-2}
	{\cal P}(M(u)) =  \mbox{lin}\{ (1/\sqrt{2},\ 1/\sqrt{2},\ 0)' \} = \mbox{lin}\{ (1\ 1\ 0)' \} .
\end{equation}
\end{example}

\begin{remark}\label{r:ex:Mu-analysis}
The fact that a pseudo null space in Example \ref{ex:Mu-analysis} is not unique should not be surprising from the following perspective. Let $X_t = (X_{1,t},X_{2,t},X_{3,t})'$ be a $3$-vector process following the VC model with (\ref{e:ex-b2u}). The pseudo eigenvectors $w_1 = (0\ 1\ 1)'$ in (\ref{e:ex-null-1}) and $w_2 = (1\ 1\ 0)'$ in (\ref{e:ex-null-2}) can be checked easily to be such that $w_1'X_t$ and $w_2'X_t$ are stationary.

It is also interesting and important to note here that the stationary dimension for this model is {\it not} $2$. For example, note that while $w_1'X_t$ and $w_2'X_t$ are indeed stationary, the $2$-vector process $(w_1'X_t, w_2'X_t)'$ is {\it not} stationary. Indeed, this is the case since e.g.\ $\EE (w_1'X_t) (w_2'X_t)=\EE(X_{2,t}+X_{3,t})(X_{1,t}+X_{2,t})=\EE X_{2,t}^2=3-\sin(2\pi t/T)$.
\end{remark}

\section{Inference of stationary subspace dimension}
\label{s:inf-dimens}

\noindent Here we discuss the statistical tests for the dimension of a stationary subspace at both the ``local'' and ``global'' levels along with their asymptotic properties.

According to (\ref{e:model-vc-d-dim})--(\ref{e:model-vc-d-dim-2}), the dimension of a stationary subspace of the VC model is the local dimension $d(u)$ or pseudo nullity $d(M(u))$ of the matrix $M(u)$ in (\ref{e:models-M}), assuming it is the same across $u$, which can further be expressed in terms of $d_0(u)$, $d_+(u)$ and $d_-(u)$. We are interested here in the statistical testing for $d_0(u)$, $d_+(u)$, $d_-(u)$ and hence also for $d(u) = d(M(u))$. We consider below two types of tests: local (that is, for fixed $u$) and global (that is, for an interval of $u$). Though it is the global test which is most relevant for (\ref{e:model-vc-d-dim}), we consider the local test since its statistic forms the basis of the global test and also because of independent interest.

For statistical inference, we obviously need an estimator of the matrix $M(u)$. It is set naturally as
\begin{equation}\label{e:model-vc-Mu-estim}
	\widehat M(u) = \widehat{A}^2(u) -  \widehat{\overline{A}}^2 :=  \frac{1}{T}\sum_{t=1}^{T}X_tX_t'K_h\Big(u - \frac{t}{T}\Big) - \frac{1}{T}\sum_{t=1}^{T}X_tX_t',
\end{equation}
where $K_h(u) = h^{-1}K(h^{-1}u)$, $K(\cdot)$ is a kernel function and $h$ denotes the bandwidth. A kernel is a symmetric function which integrates to $1$, with further regularity assumptions possibly made as well. In simulations and data application, we work with the triangle kernel $K(u) = 1-|u|$ if $|u|<1$ and $0$, otherwise.

\subsection{Local dimension test}
\label{s:inf-dimens-local}

In this section, $u$ is assumed to be fixed. A pseudo nullity $d(M(u))$ of a symmetric matrix $M(u)$ could be tested by adapting the matrix rank tests found in \citet{donald:2007rank}. For this, we need an asymptotic normality result for the estimator $\widehat M(u)$, which is stated next. The proof can be found in Appendix \ref{s:proofs-local}.

\begin{proposition}\label{p:M-loctest-vc-an}
Suppose that the assumptions for the proposition stated in Appendix \ref{s:proofs-local} hold, in particular, $T\to\infty$, $h\to 0$ so that $Th\to\infty$  and $Th^3\to 0$. Then, we have
\begin{equation}\label{e:M-loctest-vc-an}
	\frac{\sqrt{Th}}{\|K\|_2 \mu_4^{1/2}} A(u)^{-1} \Big( \widehat M(u) - M(u) \Big) A(u)'^{-1} \stackrel{d}{\to} {\cal Z}_p,
\end{equation}
where $\|K\|_2^2 = \int_\RR K(v)^2 dv$, $\mu_4 = \EE(Y_{i,t}^2-1)^2=\EE Y_{i,t}^4-1$ and ${\cal Z}_p$ is a symmetric $p\times p$ matrix having independent normal entries with variance $1$ on the diagonal and variance $1/\mu_4$ off the diagonal. Moreover, $\widehat{A}^k(u) \to_p A^k(u)$, $k=\pm 1,\pm 2$.
\end{proposition}

Note that the asymptotic normality result (\ref{e:M-loctest-vc-an}) can be expressed as
\begin{equation}\label{e:M-loctest-ass}
	a_T F(u)(\widehat M(u) - M(u)) F(u)' \stackrel{d}{\to} {\cal Z}_p,
\end{equation}
where
\begin{equation}\label{e:M-loctest-ass-def}
	a_T=\frac{\sqrt{Th}}{\|K\|_2 \mu_4^{1/2}},\quad F(u)=A(u)^{-1}.
\end{equation}
Furthermore, by Proposition \ref{p:M-loctest-vc-an}, we have
\begin{equation}\label{e:M-loctest-ass-2}
	\widehat F(u)\stackrel{p}{\to} F(u)
\end{equation}
with $\widehat F(u) = \widehat{A}(u)^{-1}$. In principle, $\mu_4$ in $a_T$ needs to be estimated as well. But for simplicity and better readability, we shall assume that $\mu_4$ is known. Construction of consistent estimators $\widehat \mu_4$ is discussed in the supplementary technical appendix (\citet{vc-appendix}), and would be sufficient for the results below to hold assuming that $\mu_4$ is estimated through $\widehat \mu_4$.
Furthermore, if one is willing to assume Gaussianity of $Y_t$, note that $\mu_4=2$ can also be used.

The convergence (\ref{e:M-loctest-ass}) with (\ref{e:M-loctest-ass-2}) is the setting for the matrix rank tests considered in \cite{donald:2007rank}, with one small difference that plays little role as noted in the proof of Corollary \ref{c:M-loctest-n0} below. We explain below how the results of \cite{donald:2007rank} can be adapted to test for $d_0(u)$, $d_+(u)$, $d_-(u)$ and hence also
for $d(M(u))$. But first it is convenient to gather some of the notation to be used on a number of occasions below. In view of \eqref{e:M-loctest-ass}, note that the matrix $F(u) = A(u)^{-1}$ plays the role of standardization. In this sense, the focus should not be so much on $M(u)$ but rather on $F(u)M(u)F(u)'$. Note that $d_0(u) = d_0(M(u))) = d_0( F(u) M(u) F(u)' )$ and also $d_{\pm}(M(u)) = d_{\pm}( F(u) M(u) F(u)' )$. We let

\begin{equation}\label{e:M-eigenvalues}
\gamma_1(u) \leq \hdots \leq \gamma_p(u) \quad \textrm{and} \quad \widehat \gamma_{1}(u) \leq \hdots \leq \widehat \gamma_{p}(u)
\end{equation}
be the ordered eigenvalues of $F(u)M(u)F(u)'$ and $\widehat F(u) \widehat M(u)  \widehat F(u)'$, respectively. We have
\begin{align} \label{e:M-eigenvalues-sorted}
\gamma_1(u) \leq \hdots \leq \gamma_{d_-(u)}(u)<0=\gamma_{d_-(u)+1}(u)= \hdots
 = \gamma_{d_-(u)+d_0(u)}(u)  \notag \\
 < \gamma_{d_-(u)+d_0(u)+1}(u)
 \leq \hdots \leq \gamma_{p} (u).
\end{align}
Let also
\begin{equation} \label{e:M-eigenvalues-squared}
0 \leq \gamma_{2,1}(u) \leq \hdots \leq \gamma_{2,p}(u) \quad \textrm{and} \quad \widehat \gamma_{2,1}(u) \leq \hdots \leq \widehat \gamma_{2,p}(u)
\end{equation}
be the ordered eigenvalues of $(F(u) M(u) F(u)')^2$ and $(\widehat F(u) \widehat M(u)  \widehat F(u)')^2$, respectively. We have $(\widehat \gamma_{i}(u))^2 = \widehat \gamma_{2,j(i)}(u)$ for some $j(i)$, and a similar expression with the hats and also

\begin{equation} \label{e:M:eigenvalues-squared-ordered}
0 = \gamma_{2,1}(u) = \hdots = \gamma_{2,d_0(u)}(u)  < \gamma_{2,d_0(u)+1}(u) \leq \hdots \leq \gamma_{2,p}(u).
\end{equation}
In particular, $d_0(u) = d_0( F(u)M(u)F(u)') = d_0 ( ( F(u)M(u)F(u)')^2 ) $. By combining these observations, we have
\begin{equation}\label{e:M-dpm-est}
d_{\pm}(u) = \# \{i: \; \gamma_i(u) \gtrless 0, \; ( \gamma_i(u) )^2 = \gamma_{2,k}(u) \; \textrm{for some} \; k=d_0(u)+1,\hdots,p \}.
\end{equation}
In the proofs for Section \ref{s:inf-dimens-global}, we shall also use eigenspaces associated with the eigenvalues above but these will not be discussed here.

We first look at the inference about $d_0(u)$. Since $d_0(u)$ is equal to $p - \mbox{rk}(M(u))$, where $\mbox{rk}(M(u))$ is the rank of $M(u)$, we can test directly for $d_0(u)$ by using the MINCHI2 test of \citet{donald:2007rank}. By the discussion above, $\mbox{rk}\{ M(u) \} = \mbox{rk}\{ (F(u)M(u)F(u)')^2 \} $. It is then natural to consider for $r=0,\ldots,p$, the test statistic
\begin{equation}\label{e:M-loctest-minchi2}
	\widehat \xi_r(u) = a_T^2 \sum_{i=1}^{r} \widehat \gamma_{2,i}(u) = \frac{Th}{\|K\|_2^{2} \mu_4}  \sum_{i=1}^{r} \widehat \gamma_{2,i}(u).
\end{equation}
Its asymptotics is described in the following result.

\begin{corollary}\label{c:M-loctest-n0}
Suppose that the assumptions of Proposition \ref{p:M-loctest-vc-an} hold. Then, under $H_0:d_0(u)=r$,
\begin{equation}\label{2:M-loctest-n0-1}
	\widehat \xi_r(u) \stackrel{d}{\to} \chi^2 ( r(r+1)/2 )
\end{equation}
and under $H_1:d_0(u)<r, \;	\widehat \xi_r(u) \to_p +\infty$.
\end{corollary}

\begin{proof}
The result follows from Theorem 4.7 in \citet{donald:2007rank}. Note that the variance in the off-diagonal of the matrix ${\cal Z}_p$ in (\ref{e:M-loctest-ass}) plays no role in the derivation of the theorem, since the argument is based on a trace and hence only on the variance on the diagonal.
\end{proof}

The corollary can be used to test for $d_0(u)$ in a standard way sequentially, namely, testing for $H_0:d_0(u)=r$ starting with $r=p$ and subsequently decreasing $r$ by $1$ till the null hypothesis is not rejected. Let $\widehat d_0(u)$ be the resulting estimator, which is consistent for $d_0(u)$ under a suitable choice of critical values in the sequential testing as the following result states.

\begin{corollary}\label{c:M-loctest-n0-2}
Suppose that the assumptions of Proposition \ref{p:M-loctest-vc-an} hold. Let $\widehat d_0(u)$ be the estimator of $d_0(u)$ defined above through the sequential testing procedure when using a significance level $\alpha=\alpha_T$ in testing. If $\alpha_T\to 0$ and $(-\log \alpha_T)/Th\to 0$, then $\widehat d_0(u)\to_p d_0(u)$.
\end{corollary}

\begin{proof}
The result can be proved as e.g.\ Theorem 4.3 in \citet{fortuna:2008local}.
\end{proof}

We now turn to inference of $d_+(u)$ and $d_-(u)$. By the discussion surrounding equations \eqref{e:M-eigenvalues}--\eqref{e:M:eigenvalues-squared-ordered}, these quantities are also $d_{\pm}(F(u)M(u)F(u)')$. In view of the relation \eqref{e:M-dpm-est}, since $\widehat d_{0}(u)$ estimates $d_0(u)$ consistently and the sample eigenvalues converge to their population analogues, it is natural to set
\begin{equation}\label{e:M-loctest-n+-}
\vspace{-0.1cm}
	\widehat d_{\pm}(u) = \# \{i: \widehat \gamma_{i}(u)\gtrless 0,\quad (\widehat \gamma_{i}(u))^2 = \widehat \gamma_{2,k}(u),\quad \mbox{for some}\ k = \widehat d_0(u)+1,\ldots,p\}.
\end{equation}
% \begin{equation}\label{e:M-loctest-n+}
% 	\widehat d_{+,u} = \# \{i: \widehat \gamma_{i,u}>0,\quad (\widehat \gamma_{i,u})^2 = \widehat \gamma_{2,k,u},\quad \mbox{for some}\ k = \widehat d_{0,u}+1,\ldots,p\}
% \end{equation}
% and
% \begin{equation}\label{e:M-loctest-n-}
% 	\widehat d_{-,u} = \# \{i: \widehat \gamma_{i,u}<0,\quad (\widehat \gamma_{i,u})^2 = \widehat \gamma_{2,k,u},\quad \mbox{for some}\ k = \widehat d_{0,u}+1,\ldots,p\}.
% \end{equation}

\begin{corollary}\label{c:M-loctest-npm}
Under the assumptions of Corollary \ref{c:M-loctest-n0-2}, we have $\widehat d_+(u)\to_p d_+(u)$ and $\widehat d_-(u)\to_p d_-(u)$.
\end{corollary}

\begin{proof}
The stated result follows from the following two observations. First, by Corollary \ref{c:M-loctest-n0-2}, $\widehat d_0(u)=d_0(u)$ a.s.\ for large enough $T$ (depending on $\omega$ in ``a.s.''). Second, since $(\widehat F(u)\widehat M(u) \widehat F(u)')^k\to_p (F(u)M(u)F(u)')^k$ for $k=1,2$, we have $\widehat \gamma_{2,j}(u)\to_p \gamma_{2,j}(u)$, $j=1,\ldots,p$, and $\widehat \gamma_i(u)\to_p \gamma_i(u)$, $i=1,\ldots,p$. Thus, in the limit, the right-hand side of \eqref{e:M-loctest-n+-} becomes that of \eqref{e:M-dpm-est}, and hence also its left-hand side, that is, $d_{\pm}(u)$.
\end{proof}

Finally, a natural estimator for $d(u) = d(M(u))$ is
\begin{equation}\label{e:M-loctest-d0}
	\widehat d(u) = \widehat d(M(u)) = \widehat d_0(u) + \min(\widehat d_+(u),\widehat d_-(u)).
\end{equation}
Corollaries \ref{c:M-loctest-n0-2}--\ref{c:M-loctest-npm} imply the consistency of this estimator, which is the main result of this section.

\begin{theorem}\label{t:M-loctest-}
Under the assumptions of Corollary \ref{c:M-loctest-n0-2}, we have $\widehat d(M(u))\to_p d(M(u))$.
\end{theorem}

\subsection{Global dimension test}
\label{s:inf-dimens-global}

In the previous section, we considered testing about $d_{0}(u)$, $d_{+}(u)$, $d_{-}(u)$ and $d(M(u))$ for fixed $u$. The relation (\ref{e:model-vc-d-dim}) of interest, however, concerns these quantities for all $u\in(0,1)$. We would thus like to be able to make inference ``globally,'' that is, for an interval of $u$. Since our approach to $d(M(u))$ goes through $d_{0}(u)$, $d_{+}(u)$ and $d_{-}(u)$, we shall make inference about these quantities first. To deal with the possibility that these quantities might differ across $u\in(0,1)$, we shall work under the assumption that
\begin{equation}\label{e:global-equal}
	d_{0}(u) \equiv d_0,\quad d_{+}(u) \equiv d_+,\quad d_{-}(u) \equiv d_-,\quad \forall \;  u \in {\cal H}\subset (0,1),
\end{equation}
and develop a global dimension test about $d_0$, $d_+$ and $d_-$ in (\ref{e:global-equal}) for fixed ${\cal H}$. In practice, we shall apply the developed test over refined dyadic partitions, first for $(0,1)$, then for $(0,1/2)$ and $(1/2,1)$, then for $(0,1/4),(1/4,1/2),(1/2,3/4)$ and $(3/4,1)$ etc. What to expect under this splitting of the interval $(0,1)$ is discussed in Section \ref{s:inf-splitting}, and will be illustrated in Sections \ref{s:simulation} and \ref{s:application}.

We first focus on inference about $d_0$ in (\ref{e:global-equal}). Our global test will be based on the quantity $\int_{\cal H} \widehat \xi_r(u) du$, where $\widehat \xi_r(u)$ is the statistic (\ref{e:M-loctest-minchi2}) used in the local dimension test. Its asymptotics are described in the next result, which also defines the global statistic $\widehat \xi_r$. Recall the notation $\mu_4 = \EE (Y_{i,t}-1)^4$ used in Proposition \ref{p:M-loctest-vc-an}; see also the discussion following \eqref{e:M-loctest-ass-2}. We also let $|{\cal H}|$ be the length of the interval ${\cal H}$ and $\overline K(u) = \int_\RR K(v) K(u-v) dv$ be the so-called convolution kernel.

\begin{proposition}\label{p:global-an}
Suppose that the assumptions for the proposition stated in Appendix \ref{s:proofs-global} hold. Then, under $H_0: d_{0}(u)\equiv r$ for all $u\in {\cal H}$,
\begin{equation}\label{e:global-an-1}
	\widehat \xi_r := \widehat \xi_{r, {\cal H}} := \frac{\int_{\cal H} \widehat \xi_r(u) du - |{\cal H}| \frac{r( \mu_4+r-1)}{ \mu_4}}{\sqrt{h \frac{\|\overline K\|_2^2}{\|K\|_2^4} \frac{2(r \mu_4^2+2r(r-1))}{ \mu_4^2} |{\cal H}|}} \stackrel{d}{\to} {\cal N}(0,1)
\end{equation}
and under $H_1: d_{0}(u) < r$ for some $u\in {\cal H}, \; \widehat \xi_r \to_p \infty$.
\end{proposition}

The proof of the proposition can be found in Appendix \ref{s:proofs-global}, and follows the approach taken in \citet{donald:2011local}.

As for the local tests (see the discussion around Corollary \ref{c:M-loctest-n0-2}), Proposition \ref{p:global-an} can be used to test for $d_0$ sequentially, namely, testing for $H_0:d_0=r$ starting with $r=p$ and subsequently decreasing $r$ by $1$ till the null hypothesis is not rejected. Let $\widehat d_{0}$ be the resulting estimator.

\begin{corollary}\label{c:global-d0-consistency}
Suppose that the assumptions of Proposition \ref{p:global-an} hold. Let $\widehat d_{0}$ be the estimator of $d_{0}$ defined above through the sequential testing procedure when using a significance level $\alpha=\alpha_T$ in testing. If $\alpha_T\to 0$ and $(-\log \alpha_T)/Th\to 0$, then $\widehat d_{0}\to_p d_{0}$.
\end{corollary}

\begin{proof}
The result can be proved as in e.g.\ Theorem 4.3 in \citet{fortuna:2008local}, by noting that the two conditions provided in the statement of the theorem are in fact equivalent.
\end{proof}

We now turn to inference about $d_+$ and $d_-$ in (\ref{e:global-equal}). Recall the definition of the eigenvalues $\widehat \gamma_{i}(u)$ in \eqref{e:M-eigenvalues}, whose squares are the eigenvalues $\widehat \gamma_{2,j}(u)$ in \eqref{e:M-eigenvalues-squared} entering the test statistics $\widehat \xi_r(u)$ and $\widehat \xi_r$. From the discussions surrounding \eqref{e:M-eigenvalues}--\eqref{e:M:eigenvalues-squared-ordered}, the $\widehat d_0$  consecutive eigenvalues $\widehat \gamma_{i}(u)$ can be thought to be associated with the $d_0$ zero eigenvalues of $F(u) M_u F(u)'$. If we can estimate the starting index for these consecutive eigenvalues, we could then deduce the numbers $d_\pm$ of positive and negative eigenvalues of $F(u) M(u) F(u)'$.
In the case $\widehat d_0 \geq 1$, the above suggests to consider

\begin{equation}\label{e:global-stat-d+}
\widehat \zeta_r = \Big| \int_{\cal H} (\widehat \gamma_{r}(u) + \ldots + \widehat \gamma_{r+\widehat d_0-1}(u)) du \Big|,\quad r=1,\ldots,p-\widehat d_0+1,
\end{equation}
that is, the quantities involving the sums of the $\widehat{d}_0$ consecutive eigenvalues $\widehat{\gamma}_i(u)$, and to estimate this starting index through

\begin{equation}\label{e:global-stat-d+-r}
	\widehat r = \mathop{\rm argmin}_{r=1,\ldots,p-\widehat d_0+1} \widehat \zeta_r.
\end{equation}

\noindent (Whenever $r$ does not match the starting index, a larger value of $\widehat \zeta_r$ is expected, since it will be driven by $\widehat \gamma_{i}(u) $ associated with the positive or negative eigenvalues of $F(u) M(u) F(u)'$.) With the estimated index $\widehat{r}$, it is then natural to set further
\begin{equation}\label{e:global-d+-}
	\widehat d_- = \widehat r - 1,\quad \widehat d_+ = p - \widehat d_0 - \widehat d_-.
\end{equation}

If $\widehat d_0=0$, the preceding argument does not apply and, in fact, the quantity (\ref{e:global-stat-d+}) is not even defined. In this case, we suggest to consider
\begin{equation}\label{e:global-stat-d+-2}
	\widehat \eta_r = \Big|\int_{\cal H} (\widehat \gamma_{1}(u) + \ldots + \widehat \gamma_{r}(u)) du \Big| + \Big|\int_{\cal H} (\widehat \gamma_{r+1}(u) + \ldots + \widehat \gamma_{p}(u)) du \Big|,\quad r=0,\ldots,p,
\end{equation}
and to set
\begin{equation}\label{e:global-d+-2}
	\widehat d_- =  \mathop{\rm argmax}_{r=0,\ldots,p}\, \widehat \eta_r,\quad \widehat d_+ = p - \widehat d_-.
\end{equation}
The idea behind this definition and further motivation for using \eqref{e:global-stat-d+}--\eqref{e:global-d+-} can be found in the proof of the following corollary.

\begin{corollary}\label{c:global-d+--consistency}
Under the assumptions of Corollary \ref{c:global-d0-consistency}, we have $\widehat d_{+}\to_p d_{+}$ and $\widehat d_{-}\to_p d_{-}$.
\end{corollary}

\begin{proof}
The result follows from the following two observations. First, by Corollary \ref{c:global-d0-consistency}, we may assume that for large enough $T$, $\widehat{d}_0 = d_0$. Second, the eigenvalues entering \eqref{e:global-stat-d+} and \eqref{e:global-stat-d+-2} converge to their population counterparts, so that the relations \eqref{e:global-stat-d+} and \eqref{e:global-stat-d+-2} become in the limit, respectively, the relations $\zeta_r = | \int_{\cal H} ( \gamma_{r}(u) + \ldots +  \gamma_{r+ d_0-1}(u)) du |$ and $\eta_r = |\int_{\cal H} ( \gamma_{1}(u) + \ldots + \gamma_{r}(u)) du | + |\int_{\cal H} ( \gamma_{r+1}(u) + \ldots +  \gamma_{p}(u)) du |$.  The conclusion follows by observing that the population quantities satisfy \eqref{e:global-stat-d+-r}, \eqref{e:global-d+-} and \eqref{e:global-d+-2} with all the hats removed.
\end{proof}

Finally, a natural estimator for $d := d_0 + \min(d_+,d_-)$ is
\begin{equation}\label{e:global-d}
	\widehat d = \widehat d_{0} + \min(\widehat d_{+},\widehat d_{-}).
\end{equation}
Corollaries \ref{c:global-d0-consistency}--\ref{c:global-d+--consistency} imply the consistency of this estimator, which is the main result of this section.

\begin{theorem}\label{t:global-d}
Under the assumptions of Corollary \ref{c:global-d0-consistency}, we have $\widehat d\to_p d$.
\end{theorem}

\subsection{Global dimension tests under interval splitting}
\label{s:inf-splitting}

The global pseudo dimension test was developed in Section \ref{s:inf-dimens-global} under the assumption \eqref{e:global-equal} for a subinterval $\mathcal{H} \subset (0,1)$. When global testing is to be performed on (0,1), we suggest to carry out the introduced global test over subintervals of refined partitions. In this section, we describe how the procedure is carried out and the results can be interpreted.

We shall need the following basic property of the global test statistic $\widehat \xi_{r, {\cal H}}$ in \eqref{e:global-an-1}. Suppose $\mathcal{H}_1$ and $\mathcal{H}_2$ are two disjoint intervals such that
\begin{equation}
\mathcal{H} = \mathcal{H}_1 + \mathcal{H}_2.
\end{equation}
Then, since $\int_{\mathcal{H}} = \int_{\mathcal{H}_1} + \int_{\mathcal{H}_2}  $ and $|\mathcal{H}| = |\mathcal{H}_1| + |\mathcal{H}_2|$, it follows from the definition \eqref{e:global-an-1} of $\widehat \xi_{r, {\cal H}}$ that
\begin{equation} \label{e:inf-split-1}
\widehat \xi_{r, {\cal H}} = \widehat \xi_{r, \mathcal{H}_1} \Big( \frac{|\mathcal{H}_1|}{\mathcal{ |H| }} \Big)^{1/2} + \widehat \xi_{r, \mathcal{H}_2} \Big(\frac{|\mathcal{H}_2|}{\mathcal{|H|}} \Big)^{1/2}.
\end{equation}
In particular, in view of Proposition \ref{p:global-an}, under $d_0(u) \equiv r, \; u \in \mathcal{H}_i$,
\begin{equation}\label{e:inf-split-2-dist}
 \widehat \xi_{r, \mathcal{H}_i} \stackrel{d}{\to} N(0,1) \; =: \; Z_i, \; i=1,2,
\end{equation}
where $Z_i$'s can be considered independent because the VC model involves independent variables across $\mathcal{H}_1$ and $\mathcal{H}_2$. The relations \eqref{e:inf-split-1} and \eqref{e:inf-split-2-dist} then imply
\begin{equation} \label{e:inf-split-3-dist2}
\widehat \xi_{r, {\cal H}} \stackrel{d}{\to} Z_1 \Big( \frac{|\mathcal{H}_1|}{\mathcal{|H|}} \Big)^{1/2} + Z_2 \Big( \frac{|\mathcal{H}_2|}{\mathcal{|H|}} \Big)^{1/2} \stackrel{d}{=} N(0,1),
\end{equation}
which is consistent with what is expected under $d_0(u) \equiv r, \; u \in \mathcal{H} = \mathcal{H}_1 + \mathcal{H}_2$ by Proposition \ref{p:global-an}. Such considerations will allow having some consistency over refined partitions in the sense described below. We first consider the case of $d_0(u)$, and then discuss those of $d_{\pm}(u)$ and $d(u) = d(M(u))$.

Thus, let
\begin{equation}\label{e:inf-split-4-interval-split-defin}
\mathcal{H}_i^{(k)} = \Big( \frac{i-1}{2^k} , \frac{i}{2^k} \Big], \; i=1,
\hdots,2^k,
\end{equation}
form refined partitions of $(0,1 \rbrack$ for $k \geq 0$. Let $\widehat{d}_{0,i}^{(k)}$ be the global estimator of $d_0$ over $\mathcal{H}_i^{(k)}$ by using the test statistic $\widehat \xi_{r, \mathcal{H}_i^{(k)}}$. In view of \eqref{e:inf-split-1}--\eqref{e:inf-split-3-dist2} and for the sake of consistency, when estimating $d_0$ over finer partitions, we suggest to adjust the normal critical value for comparing $\widehat \xi_{r, \mathcal{H}_i^{(k)}}$. More precisely, if $c_{\alpha}^{(0)} = c_{\alpha}$ is a normal critical value used at the level $k=0$, we use the critical value $c_{\alpha}^{(k)} = 2^{-k/2}c_{\alpha}$ at level $k$.

As a consequence of the choice of the critical values, there are only the following three possibilities for estimators $\widehat{d}_{0,i}^{(k)}$ when refining a partition:
\begin{itemize}

\item[(P1)] $\widehat{d}_{0,i}^{(k)} = r, \; \widehat{d}_{0,2i-1}^{(k+1)} = r,\; \widehat{d}_{0,2i}^{(k+1)} = r $ ;

\item[(P2)] $\widehat{d}_{0,i}^{(k)} = r$ and either $\widehat{d}_{0,2i-1}^{(k+1)} = r, \; \widehat{d}_{0,2i}^{(k+1)} < r$ or $\widehat{d}_{0,2i-1}^{(k+1)} < r, \; \widehat{d}_{0,2i}^{(k+1)} = r$ ;

\item[(P3)] $\widehat{d}_{0,i}^{(k)} = r$ and either $\widehat{d}_{0,2i-1}^{(k+1)} > r, \; \widehat{d}_{0,2i-1}^{(k+1)} > \widehat{d}_{0,2i}^{(k+1)} = r$ or $\widehat{d}_{0,2i}^{(k+1)} > r, \; \widehat{d}_{0,2i}^{(k+1)} > \widehat{d}_{0,2i-1}^{(k+1)}$.

\end{itemize}

\noindent Indeed, let us explain the first two of these possibilities, and also indicate a case which is not one of the possibilities.

The possibility (P1) arises in the following scenario: one has $\widehat{d}_{0,i}^{(k)} = r$ when $\widehat \xi_{j, \mathcal{H}_i^{(k)}} > 2^{-k/2}c_{\alpha}$ for $j=p,p-1,\hdots,r+1$ and $\widehat \xi_{r, \mathcal{H}_i^{(k)}} \leq 2^{-k/2}c_{\alpha}$. Similarly, $\widehat{d}_{0,2i-1}^{(k+1)} = r $ and  $\widehat{d}_{0,2i}^{(k+1)} = r$ when $\widehat \xi_{j, \mathcal{H}_{2i-1}^{(k+1)}} , \widehat \xi_{r, \mathcal{H}_{2i}^{(k+1)}}  > 2^{-(k+1)/2}c_{\alpha} $ for $j=p,p-1,\hdots,r+1$ and $\widehat \xi_{r, \mathcal{H}_{2i-1}^{(k+1)}} , \widehat \xi_{r, \mathcal{H}_{2i}^{(k+1)}} \leq 2^{-(k+1)/2}c_{\alpha}$. This is consistent with the special case of \eqref{e:inf-split-1},

\begin{equation}\label{e:inf-split-5}
\widehat \xi_{j, \mathcal{H}_i^{(k)}} = \widehat \xi_{j, \mathcal{H}_{2i-1}^{(k+1)}} \cdot 2^{1/2} + \widehat \xi_{j, \mathcal{H}_{2i}^{(k+1)}} \cdot 2^{1/2},
\end{equation}
in the sense that the relationships of the two summands of \eqref{e:inf-split-5} to the respective critical values is the same as that for their sum.

The possibility (P2), on the other hand, corresponds to the scenario when $\widehat \xi_{i, \mathcal{H}_{2i-1}^{(k+1)}} , \widehat \xi_{j, \mathcal{H}_{2i}^{(k+1)}}  > 2^{-(k+1)/2}c_{\alpha}$ for $j=p,p-1,\hdots,r+1$, but then either $\widehat \xi_{r, \mathcal{H}_{2i-1}^{(k+1)}}  \leq  2^{-(k+1)/2}c_{\alpha}$ and $\widehat \xi_{r, \mathcal{H}_{2i}^{(k+1)}}  >  2^{-(k+1)/2}c_{\alpha}$ or $\widehat \xi_{r, \mathcal{H}_{2i-1}^{(k+1)}}  >  2^{-(k+1)/2}c_{\alpha}$ and $\widehat \xi_{r, \mathcal{H}_{2i}^{(k+1)}}  \leq 2^{-(k+1)/2}c_{\alpha}$. The case that is excluded from the possibilities listed above, is $\widehat{d}_{0,2i-1}^{(k+1)} < r$ and  $\widehat{d}_{0,2i}^{(k+1)} < r$ (and $\widehat{d}_{0,i}^{(k)} = r$) which would happen if $\widehat \xi_{r, \mathcal{H}_{2i-1}^{(k+1)}} , \widehat \xi_{r, \mathcal{H}_{2i}^{(k+1)}}   >  2^{-(k+1)/2}c_{\alpha}$ but is impossible in view of \eqref{e:inf-split-5}. In our experiments with simulated and real data, the possibilities (P1) and (P2) seem to be the most common.

Figure \ref{fig:splitting_population}, left plot, presents global testing results under splitting for one realization of a VC model (Model 5 from Section \ref{s:simulation_dimension}). In this plot, the estimate of $\widehat{d}_{0,1}^{(0)}$ for $(0,1 \rbrack$ is presented at the point 1/2 of the $x$-axis which is the midpoint of $(0,1 \rbrack$. The estimator $\widehat{d}_{0,1}^{(1)}$ and $\widehat{d}_{0,2}^{(1)}$ are presented at points 1/4 and 3/4 respectively, which are the midpoints of the intervals $(0,1/2 \rbrack$ and $(1/2,1 \rbrack$. The presentation is continued in the same way, till the level $k=2$ is reached and the estimates $\widehat{d}_{0,1}^{(2)}, \widehat{d}_{0,2}^{(2)} , \widehat{d}_{0,3}^{(2)}, \widehat{d}_{0,4}^{(2)}$ are presented at the finest considered level.

\begin{figure}[t]
\begin{center}
\includegraphics[scale=0.38]{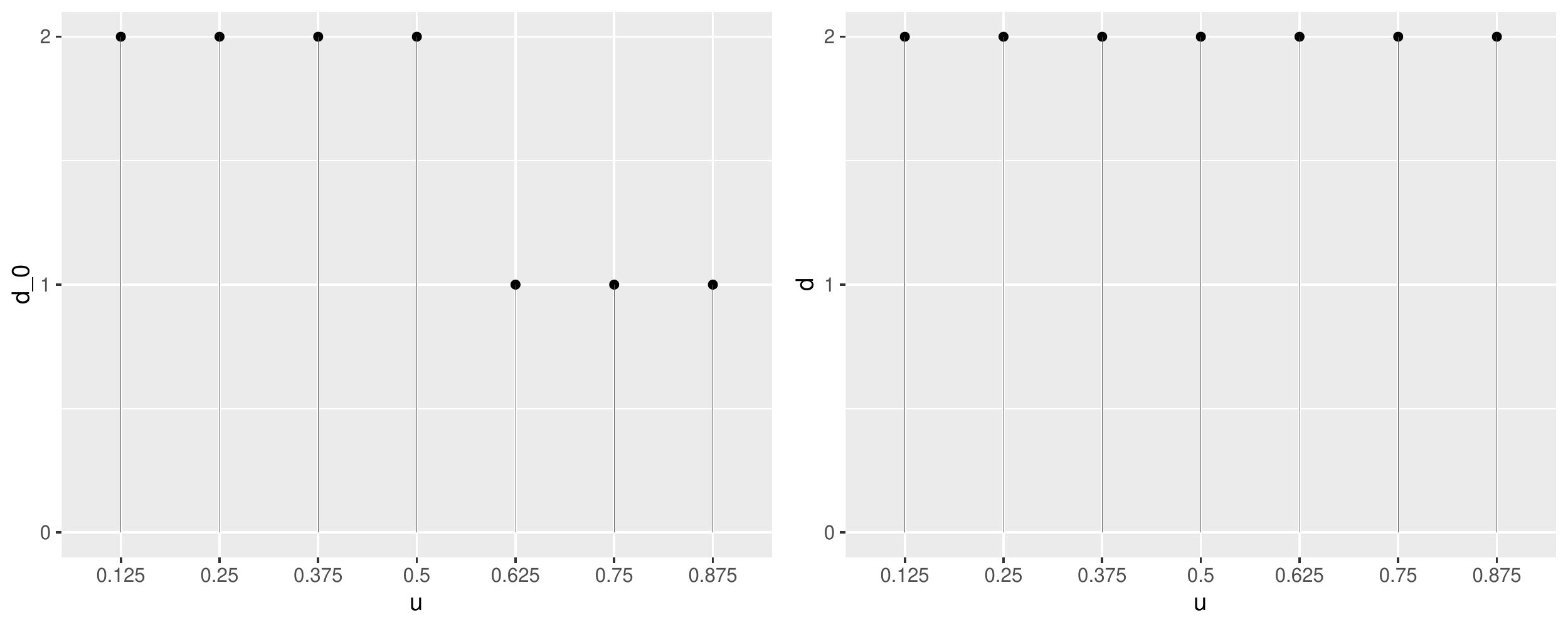}
\end{center}
\vspace{-0.8cm}
\caption{Left: Visualization of the estimates of $\widehat{d}_{0,i}^{(k)}$ for $k=0,1,2$ from one realization of a VC model. Right: Similar visualization of $\widehat{d}_i^{(k)}$. }\label{fig:splitting_population}
\end{figure}

The estimators $\widehat{d}_{0,i}^{(k)}$ lead to the corresponding estimates $\widehat{d}_{\pm,i}^{(k)}$ and $\widehat{d}_{i}^{(k)}$. The latter could be presented similarly as in the described left plot of Figure \ref{fig:splitting_population}. This is illustrated in the right plot of Figure \ref{fig:splitting_population}.  

\section{Estimation of stationary subspace}
\label{s:inf-subspace}

We turn here to the estimation of a stationary subspace of the VC model, which is related to pseudo null spaces, given in Definition \ref{d:descr-psudo-nullity}, of matrices $M(u)$ through (\ref{e:model-vc-dB1}). We shall not provide here any formal statistical tests related to a stationary subspace but rather make a number of related comments, inspired by the developments in Section \ref{s:char-pnull}. In Section \ref{s:11-subspace}, we introduce a particular class of local stationary subspaces. In Section \ref{s:cluster_subspace}, a graph-based method is provided that aims at forming clusters of the many local stationary subspaces introduced in Section  \ref{s:11-subspace}. Finally, a technique to select one subspace out of the clusters is provided and this serves as our estimate of a stationary subspace.  

\subsection{(1,1)-local stationary subspaces} 
\label{s:11-subspace}

The results of Section \ref{s:char-pnull} show that pseudo null spaces have a rich structure and are typically not unique, for a given matrix. Motivated by the proof of Proposition \ref{p:rank_claim_1}, we shall restrict our attention to their special cases, given in the following definition. The definition and subsequent developments use the notation of Section \ref{s:char}, namely that of $M$, $d_0$, $d_\pm$, $\lambda_i$, $\lambda_{0,i}$, $\lambda_{\pm,i}$, $s_i$, $s_{0,i}$, $s_{\pm,i}$.

\begin{definition}\label{d:11-null-spaces}
Let $M$ be a symmetric matrix and suppose $\min(d_+,d_-)\geq 1$. A {\it $(1,1)$-pseudo null space} of $M$ is defined as
\begin{equation}\label{e:11-null-spaces}
	{\cal P}_{(1,1)}(M) = \mbox{\rm lin}\{ s_{0,1},\ldots,s_{0,d_0},\alpha_i s_{+,p(i)} + \beta_i s_{-,n(i)},\ i=1,\ldots,\min(d_+,d_-) \},
\end{equation}
for a fixed $\alpha_i,\beta_i$, where $p(i)\in \{1,\ldots,d_+\}$ are different across $i$, $n(i)\in \{1,\ldots,d_-\}$ are different across $i$ and 
\begin{equation}\label{e:11-null-spaces-2}
	\alpha_i^2 \lambda_{+,p(i)} + \beta_i^2 \lambda_{-,n(i)}=0,\ i=1,\ldots,\min(d_+,d_-).
\end{equation}
When $M = M(u)$ with $M(u)$ as in \eqref{e:models-M}, a (1,1)-pseudo null space $\mathcal{P}(M(u))$ will be called a (1,1)-\textit{local stationary subspace} and denoted as $\mathcal{B}_{(1,1)}(u)$.
\end{definition}

The fact that ${\cal P}_{(1,1)}(M)$ defines a pseudo null space for $M$ follows as in the proof of Proposition \ref{p:rank_claim_1}.

\begin{remark}\label{r:11-number}
The total number of $(1,1)$-pseudo null spaces of $M$ is
\begin{equation}\label{e:11-number}
	n(M)  \; = \; { d_+ \choose \min(d_+,d_-)} \cdot { d_- \choose \min(d_+,d_-)} \cdot (\min(d_+,d_-)!),
\end{equation}
where the first two terms account for the selection of eigenvectors associated with the positive and negative eigenvalues, and the last term for pairing them off. Depending on the values of $d_\pm$, the total number (\ref{e:11-number}) can be quite large: e.g.\ with $d_+=3$ and $d_-=5$, the number is $60$. 
\end{remark}

\begin{remark}\label{r:11-expl}
The prefix ``$(1,1)$-'' in Definition \ref{d:11-null-spaces} refers to the fact that a pseudo eigenvector of a pseudo null space is constructed by taking one (1) eigenvector $s_{+,p(i)}$ associated with the positive eigenvalues and one (1) eigenvector $s_{-,n(i)}$ associated with the negative eigenvalues. More elaborate constructions are possible as well, for example, by taking two (2) eigenvectors associated with the positive eigenvalues and one (1) eigenvector with the negative ones, as in Example \ref{ex:Mu-analysis}, which could be called a $(2,1)$-pseudo null space. In this work though, we shall consider only $(1,1)$-pseudo null spaces.
\end{remark}

If $\widehat M$ estimates $M$ with the analogous estimators $\widehat d_0$, $\widehat d_\pm$, $\widehat \lambda_i$, $\widehat \lambda_{0,i}$, $\widehat \lambda_{\pm,i}$, $\widehat s_i$, $\widehat s_{0,i}$, $\widehat s_{\pm,i}$ of the respective quantities, we would similarly like to have the sample counterparts of the $(1,1)$-pseudo null space in (\ref{e:11-null-spaces}). Definition \ref{d:11-null-spaces} may not, however, extend directly to the sample quantities since $\widehat d_\pm$ may not necessarily represent the actual number of positive/negative eigenvalues $\widehat \lambda_{\pm,i}$ and hence the condition (\ref{e:11-null-spaces-2}) may not be satisfied. To deal with this possibility, we define the sample counterparts $\widehat {\cal P}_{(1,1)}(M)$ in the same way as in  (\ref{e:11-null-spaces}) by using the quantities with the hats, except that $\widehat d_+$ and $\widehat d_-$ are replaced by $\widetilde d_+$ and $\widetilde d_-$, which are defined as
\begin{equation}\label{e:d+-tilde}
	\widetilde d_+ = \max\{r: r\leq \widehat d_+,\ 0<\widehat \lambda_{p-r+1} \leq \ldots \leq \widehat \lambda_p \}
\end{equation}
and 
\begin{equation}\label{e:d--tilde}
	\widetilde d_- = \max\{r: r\leq \widehat d_-,\ \widehat \lambda_1 \leq \ldots \leq \widehat \lambda_r<0 \},
\end{equation} 
where we assumed that the eigenvalues $\widehat \lambda_i$ appear in the non-decreasing order. That is, we define a sample $(1,1)$-pseudo null space as
\begin{equation}\label{e:11-null-spaces-sample}
	\mathcal{P}_{(1,1)}( \widehat{M} ) = \mbox{\rm lin}\{ \widehat s_{0,1},\ldots,\widehat s_{0,\widehat d_0},\alpha_i \widehat s_{+,p(i)} + \beta_i \widehat s_{-,n(i)},\ i=1,\ldots,\min(\widetilde d_+,\widetilde d_-) \},
\end{equation}
where $p(i)\in \{1,\ldots,\widetilde d_+\}$ are different across $i$, $n(i)\in \{1,\ldots,\widetilde n_+\}$ are different across $i$ and 
\begin{equation}\label{e:11-null-spaces-2-sample}
	\alpha_i^2 \widehat \lambda_{+,p(i)} + \beta_i^2 \widehat \lambda_{-,n(i)}=0,\ i=1,\ldots,\min(\widetilde d_+,\widetilde d_-).
\end{equation}

\noindent Replacing $\widehat{M}$ above by $\widehat{M}(u)$ from \eqref{e:model-vc-Mu-estim} for $u \in (0,1)$, estimators $\widehat{\mathcal{B}}_{(1,1)}(u) = \mathcal{P}_{(1,1)}(\widehat{M}(u))$ can be defined for (1,1)-local stationary subspaces. Techniques to cluster these (1,1)-local stationary subspaces is discussed next. 

\subsection{Clustering (1,1)-local stationary subspaces} 
\label{s:cluster_subspace}

According to \eqref{e:model-vc-dB1}, there is a stationary subspace $\mathcal{B}_1$ for a VC model if there is at least one identical stationary subspace for all $u \in (0,1)$. In Section \ref{s:11-subspace}, we defined (1,1)-local stationary subspaces $\mathcal{B}_{(1,1)}(u)$ whose number can already be quite large for a fixed $u$; see Remark \ref{r:11-number}. A natural possibility in defining a candidate for a stationary subspace $\mathcal{B}_1$ is to consider all (1,1)-local stationary subspaces across $u$'s and select a subspace representing their ``majority'' in some suitable sense. In light of this observation, using clustering seems natural and this approach is pursued here on the estimated (1,1)-local stationary subspaces. 

More specifically, we discuss a graph-based approach to clustering the many (1,1)-local stationary subspaces, or equivalently, the many (1,1)-pseudo null spaces using distances that are computed based on  canonical angles between spaces. In addition, a method to select one (1,1)-local stationary subspace out of the many is discussed and this is considered as our estimate of a  stationary subspace. 

As in Section \ref{s:11-subspace}, we let $\widehat{\mathcal{B}}_{(1,1)}(u) = \mathcal{P}_{(1,1)}( \widehat{M}(u))$ denote a sample (1,1)-local stationary subspace of the matrix $\widehat M(u)$ from (\ref{e:model-vc-Mu-estim}). Let $\widehat{\mathcal{B}}_{(1,1)} = \bigcup_u  \widehat{\mathcal{B}}_{(1,1)}(u)$ be the union of all (1,1)-local stationary subspaces over $u$. In practice, the union $\bigcup_u$ is replaced by a set of discrete points $\{u_1,u_2,\hdots,u_{n_u} \}$ in $(0,1)$. 

Consider a graph $G=(V,E)$ where every vertex $v \in V$ corresponds to a (1,1)-local stationary subspace in $\widehat {\cal B}_{(1,1)}$. The adjacency matrix $E$ will be defined in terms of a   distance between (1,1)-local stationary subspaces using canonical angles. Let $\mathscr{B}_1$ and $\mathscr{B}_2$ be two (1,1)-local stationary subspaces in $\widehat {\cal B}_{(1,1)}$ of dimensions $d_1$ and $d_2$, respectively. Letting $d =   \min( d_1,d_2)$, the canonical angles computed between these two spaces are given by $\theta_1 \leq \theta_2 \leq \hdots \leq \theta_d$, where 

\begin{gather}\label{e:canonical_angle_definition}
\theta_1 \; = \; \min_{ x_1 \in \mathscr{B}_1, \; y_1 \in \mathscr{B}_2 } \arccos \Big( \frac{x_1 y_1}{||x_1|| \cdot ||y_1|| } \Big), \\
\theta_j \; = \; \min_{ \substack{  x_j \in \mathscr{B}_1, \; y_j \in \mathscr{B}_2; \\ x_j \perp x_1,x_2, \hdots ,x_{j-1}, \; y_j \perp y_1,y_2,..,y_{j-1}  }  }  \arccos \Big( \frac{x_j y_j}{ ||x_j|| \cdot  ||y_j||  } \Big), \quad j=2,3,\hdots,d.
\end{gather}
The vectors $x_1,x_2,\hdots,x_d$ and $y_1,y_2, \hdots ,y_d$ are called canonical vectors. We measure the  distance between spaces $\mathscr{B}_1$ and $\mathscr{B}_2$ as $\max_{1 \leq i \leq d} \theta_i, \;$ and define the adjacency matrix $E = (e_{i,j})$ for $i,j=1,2,\hdots,|V|$, as
\begin{equation} \label{e:adjacency_matrix}
e_{i,j} \; = \; 1, \;\; \textrm{if} \;\; \theta(v_i,v_j)<\theta_0,
\end{equation}

\noindent where $\theta(v_i,v_j)$ is the maximum canonical angle between the (1,1)-local stationary subspaces corresponding to vertices $v_i$ and $v_j$, and $\theta_0$ is a threshold. The choice of $\theta_0$ dictates the number of clusters estimated with more being formed for lower values of $\theta_0$. In our numerical work, we set $\theta_0 =20^{\circ}$. 

Finally, in order to obtain the clusters of vertices in the graph $G$, we utilize the Walktrap algorithm of \cite{pons_walktrap}. Here, a transition probability matrix $P=(p_{ij})$ is constructed with $p_{ij} = \frac{A_{ij}}{d(i)}$ where $A=(A_{ij})$ denotes the adjacency matrix of $G$ and $d(i)$ denotes the degree of vertex $v_i$. A random walk process defined on this graph $G$ is based on the powers of the matrix $P$, that is, the probability of moving from vertex $v_i$ to $v_j$ through a random walk of length $t$ is given by $P^{t}_{ij}$. The closeness of vertices in the graph is measured by these probabilities from the observation that if two vertices $v_i$ and $v_j$ are in the same cluster, $P^t_{ij}$ must be high.  

Let $C_1,C_2, \hdots ,C_L $ be the $L$ cluster of vertices produced by the Walktrap algorithm that have a size of at least 3 vertices. We first obtain the centers $\{c_1,c_2, \hdots ,c_L \}$ of these $L$ clusters by computing the sine of the maximum canonical angle, 
\begin{equation} \label{e:cluster_center}
c_l \; = \; \min_{v \in C_l} \; \sum_{v' \neq v} \; \sin (\theta(v' , v)), \quad  l=1,2,\hdots,L.
\end{equation}     
The (1,1)-local stationary subspaces corresponding to the $L$ cluster centers are considered as the candidate stationary subspaces returned by our method. Additionally, in order to select a single (1,1)-local stationary subspace (our stationary subspace estimator) out of these $L$ representative subspaces, we assess the ``denseness'' of each cluster through 
%\begin{equation} \label{e:vetrex_closeness}
%T_v \; = \; \frac{d(v)}{\sum_{v' \neq v} \sin (\theta(v' , v))} \;\;, \;\; v \in C_j
%\end{equation}
%as a measure of 'closeness' of vertex $v$ within cluster $C_j$ where $d(v)$ denotes the degree of vertex $v$. Then we compare the average 'closeness' measure of vertices in each of the $K$ clusters. More precisely, we have
%\begin{equation} \label{e:avg_cluster_closeness}
%\overline{T}_{C_j} \; = \; \sum_{v \in C_j} T_v 
%\end{equation}  
\begin{equation} \label{e:avg_cluster_closeness}
T ( C_l ) \; = \; \frac{  \frac{1}{|C_l|} \sum_{v \in C_l} \log( d(v)) }{ \frac{1}{|C_l|} \sum_{v \in C_l} \Big( \frac{1}{|C_l|-1} \sum_{v' \in C_l, \; v' \neq v }  \sin (\theta(v' , v) \Big)   }, \quad l=1,2,\hdots,L, 
\end{equation}  
 where $d(v)$ denotes the degree of vertex $v$ within cluster $C_l$. We then identify the cluster with maximum $T(\cdot)$ among the $L$ clusters and select the final (1,1)-local stationary subspace corresponding to the center of that cluster. That is, we select our stationary subspace estimate as the (1,1)-local stationary subspace corresponding to the cluster center $c_s$ in cluster $C_s$, where
\begin{equation}\label{e:subspace_select_graph}
C_s \; = \; \underset{l}{\operatorname{arg max}} \;  T(C_l).   
\end{equation}

\begin{example} \label{ex:subspace-cluster}
To illustrate the above technique, we consider the VC model \eqref{e:intro-vcmodel} with $p=3$, i.i.d.\ $Y_t \sim N(0,I_3)$, $T=1000$ and $A^2(u) = \mbox{diag}(2+0.5 \sin (2 \pi u) , 3 - \sin (2 \pi u) , 1.5 + \sin (2 \pi u) )$. The stationary subspace dimension for this model is $d=1$ and the true (1,1)-local stationary subspaces are given by $(\frac{1}{\sqrt{2}},\frac{1}{\sqrt{2}},0)$, $(0,\frac{1}{\sqrt{2}},\frac{1}{\sqrt{2}})$, $(\frac{1}{\sqrt{2}},-\frac{1}{\sqrt{2}},0)$ and $(-\frac{1}{\sqrt{2}},\frac{1}{\sqrt{2}},0)$ for all $u$'s. We generated one realization of the series $X_t$ based on this model and obtained all the estimated (1,1)-local stationary subspaces across a set of points $\{ 0.04, 0.08,0.12,\hdots,0.96 \}$. Figure \ref{fig:subspaces_model_1} depicts a 3D plot that includes the population (1,1)-local stationary subspaces (solid circles), the estimated (1,1)-local stationary subspaces (crosses), the cluster centers $c_l$ for $l=1,2,3,4$ of the 4 clusters (open circles) and the selected (1,1)-local stationary subspace based on \eqref{e:subspace_select_graph} (solid square, marked as VC Final). Finally, the estimated stationary subspace from DSSA (\cite{sundararajan:2017}) is plotted along with the other spaces (diamond). 

In Table \ref{tab:cluster-sizes}, we list the cluster centers of the 4 main clusters in Figure \ref{fig:subspaces_model_1}, the sizes of those 4 clusters and the proportions of $u \in \{ 0.04, 0.08,0.12,\hdots,0.96 \}$ with points in the respective clusters.  Observe that the (1,1)-local stationary subspace selected as the stationary subspace based on \eqref{e:subspace_select_graph} lies in the most ``dense'' cluster and the DSSA stationary subspace also lies in the same cluster. Here, the 4 biggest clusters formed by the method comprise roughly 84\% of the total number of (1,1)-local stationary subspaces. The (1,1)-local stationary subspace selected as the stationary subspace and identified as the solid square point in Figure \ref{fig:subspaces_model_1} lies in the biggest cluster that contains roughly 24\% of the (1,1)-local stationary subspaces. This subspace also lies in the cluster with the highest proportion of $u$'s with points in that cluster. 
\end{example}

\begin{figure}[t]
\begin{center}
\includegraphics[scale=0.87]{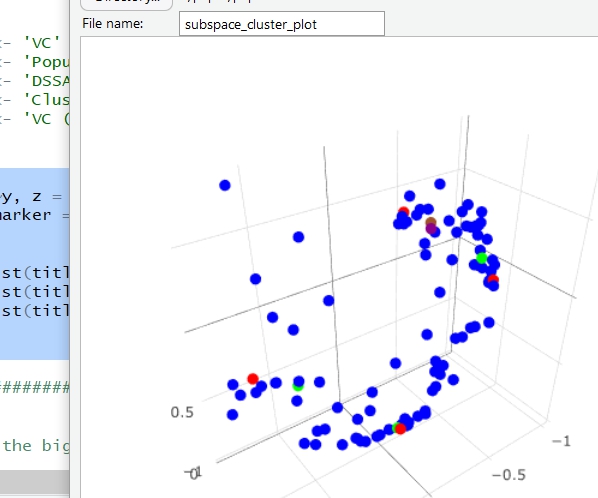}
\end{center}
\vspace{-0.55cm}
\caption{Local and global stationary subspaces in Example \ref{ex:subspace-cluster}: 3D plot of the population (1,1)-local stationary subspaces (solid circles), the estimated (1,1)-local stationary subspaces (crosses), the cluster centers of the 4 clusters (open circles),  the stationary subspace based on \eqref{e:subspace_select_graph} (solid square, marked as VC Final), the estimated stationary subspace from DSSA (diamond).}\label{fig:subspaces_model_1}
\end{figure}

\begin{table}[t] 
\begin{center}
\begin{tabular}{|c |c |c | }
\hline
Cluster center &  Cluster size (in \%) & Proportion of $u$'s  \\
\hline 
$(0.661,  0.750 , -0.003)$ & 19.79  & 79.16   \\
\hline
$(0.029, -0.641 , 0.767)$ & 23.95  & 91.67  \\
\hline 
$(-0.114 , 0.778 , 0.618)$ & 23.95 & 95.83  \\
\hline 
$(0.938, -0.346, -0.018)$ & 17.70  & 70.85  \\
\hline 
\end{tabular}
\caption{ The 4 cluster centers, cluster sizes (in percentage), and the proportions of $u$'s for the 4 biggest clusters in Figure \ref{fig:subspaces_model_1}. The selected (1,1)-local stationary subspace identified as the purple point in Figure \ref{fig:subspaces_model_1} corresponds to $(-0.114 , 0.778 , 0.618)$. } \label{tab:cluster-sizes}
\end{center}
\end{table}

\section{Simulation study}
\label{s:simulation}

In Section \ref{s:simulation_dimension}, we evaluate the empirical performance of the proposed method in estimating the dimension of a stationary subspace for several VC models. In Section \ref{s:simulation_subspace}, we assess the ability of the proposed method in estimating a  stationary subspace using a few discrepancy measures.

\subsection{Dimension estimation comparison}
\label{s:simulation_dimension}

We first consider several VC models \eqref{e:intro-vcmodel}, characterized through $A(u)$, for which the dimensions $d_0(u) = d_0$, $\min (d_+(u) , d_-(u) ) = \min (d_+ , d_-)$ and $d(u) \equiv d$ do not depend on $u$ (and neither do the local stationary subspaces). The model matrices $A(u)$, the respective matrices in \eqref{e:models-M} and the dimensions $p$, $d_0$, $ \min (d_+ , d_-)$ are: \\

\noindent \textbf{Model 1:} $p=3$, $d_0=0$, $\min (d_+, d_-) = 1$, $d=1$,
\begin{gather*}
A^2(u) = \mbox{diag} (2 + 0.5\sin(2 \pi u) , 3 - \sin(2 \pi u) , 1.5 + \sin(2 \pi u) ), \\
M(u) = \mbox{diag} (0.5\sin(2 \pi u) ,  - \sin(2 \pi u) , \sin(2 \pi u)). \; \; \quad \quad \quad \; \quad
\end{gather*}

\vspace{0.5cm}

\noindent \textbf{Model 2:} $p=3$, $d_0=2$, $\min (d_+, d_-) = 0$, $d=2$,
\begin{equation*}
A^2(u) = \mbox{diag} (3 -2u , 3 , 4 ), \; M(u) = \mbox{diag} (1-2u,0,0).
\end{equation*}

\vspace{0.5cm}

\noindent \textbf{Model 3:} $p=4$, $d_0=1$, $\min (d_+, d_-) = 1$, $d=2$,
\begin{equation*}
A^2(u) =
\begin{pmatrix}
e^1 &&  1 && 0 && 0 \\
1 &&  2+ \sin(2 \pi u)  && 0 && 0 \\
0 && 0 && 3 - 2u && 0.5 \\
0 && 0 && 0.5 && 3 - \sin(2 \pi u)
\end{pmatrix} ,
\end{equation*}
$M(u) = \mbox{diag} (0 , \sin(2 \pi u) , 1-2u , -\sin(2 \pi u))$.

\vspace{0.5cm}

\noindent \textbf{Model 4:} $p=4$, $d_0=3$, $\min (d_+, d_-) = 0$, $d=3$,
\begin{equation*}
A^2(u) =
\begin{pmatrix}
e^u &&  \rho && \rho^2 && \rho^3 \\
\rho &&  1  && \rho && \rho^2 \\
\rho^2 && \rho && 4 && \rho \\
\rho^3 && \rho^2 && \rho && e^1
\end{pmatrix}, \;  \rho=0.5, \; M(u) = \mbox{diag} ( e^u-e-1 , 0,0,0  ).
\end{equation*}

\vspace{0.5cm}

Models 1 and 2 are 3-dimensional and have diagonal $A^2(u)$. Models 3 and 4 are 4-dimensional with non-diagonal $A^2(u)$.

We suppose that the models above are Gaussian with i.i.d.\ $ N(0,I_p)$ vectors $Y_t$. In the simulations we take  $T \in \{ 200, 500, 1000, 2000 \}$ as the sample sizes. In estimating $A^2(u)$ in \eqref{e:model-vc-Mu-estim}, bandwidth choices $h$ ranging from $T^{-0.1}$ to $T^{-0.5}$ were attempted and the best results were obtained for $h \in (T^{-0.3}, T^{-0.4})$. We fix  $h =  T^{-0.35}$ and present the results for this choice.  We focus on testing for $\mathcal{H} = (0,1)$ only and use 100 Monte Carlo replications.

We compare the performance of the proposed dimension estimator $\widehat{d}$ given in \eqref{e:global-d} with the sequential estimation procedure provided in Section 2.2.5 of \cite{sundararajan:2017}. The method found in the latter work will be referred to as DSSA and the proposed method will be denoted as VC. The estimation results for the two methods and the four considered models are presented through violin plots in Figure \ref{fig:all_model_violin}. Violin plots are intended as a somewhat qualitative visualization of the results -- perhaps slightly more informative histograms for the results can be found in \citet{vc-appendix}. The plots show that estimation improves with increasing sample size $T$, and that the proposed VC method performs better than the competing DSSA method in detecting the true dimension $d$.

\begin{figure}[t]
\begin{center}
\includegraphics[scale=0.42]{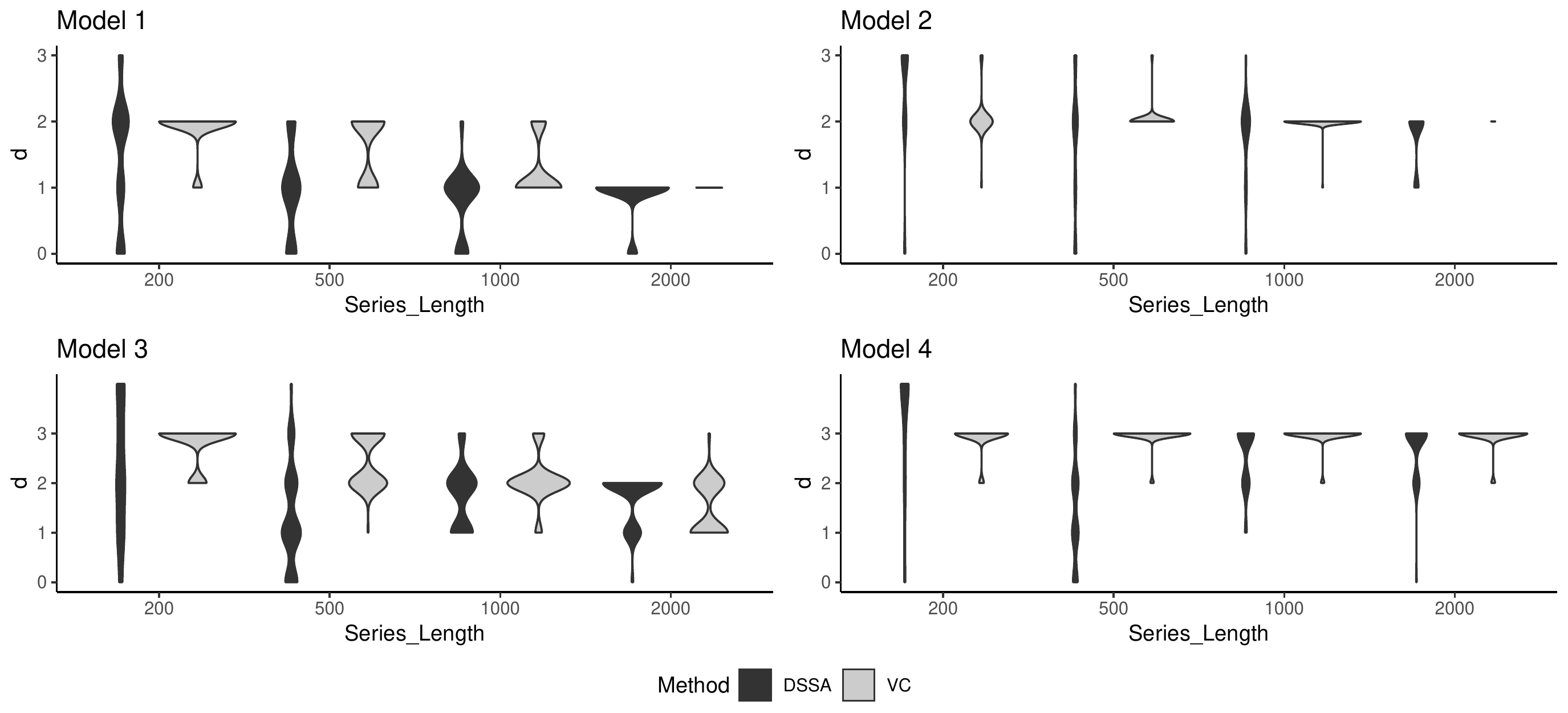}
\end{center}
\caption{Models 1-4 : Violin plots of the estimates of $d$ for the indicated sample sizes for the two competing methods: DSSA and VC (proposed method).  } \label{fig:all_model_violin}
\end{figure}

%\subsubsection{Estimation over smaller intervals}
%\label{s:simulation-splitting}

We now turn to VC models whose dimensions depend on $u$'s. We consider the following models:

\vspace{0.6cm}

\noindent \textbf{ Model 5:} $p=3$,  $d_{0}(u)=2$ and $d(u)=2$ if $u<0.5$ and $d_{0}(u)=1$ and $d(u) = 2$ if $u>0.5$,

\begin{equation*}
A^2(u) = 1_{(0,0.5)}(u) \; \mbox{diag} (2 + \sin(2 \pi u) , 2.0901 , 3) +  1_{(0,0.5)}(u) \;  \mbox{diag} (2 + \sin(2 \pi u) , 4 + 3\sin(2 \pi u) , 3).
\end{equation*}

\vspace{0.5cm}

\noindent \textbf{ Model 6:} $p=3$,  $d_{0}(u)=3$ and $d_u=3$ if $u<0.5$ and $d_0(u)=2$ and $d(u) = 2$ if $u>0.5$,

\begin{equation*}
A^2(u) =
1_{(0,0.5)}(u) \; \begin{pmatrix}
4 &&  0.5 && 0 \\
0.5 &&  3.125  && 0 \\
0 && 0 && 1
\end{pmatrix} +
1_{(0.5,1)}(u) \; \begin{pmatrix}
4 &&  0.5 && 0 \\
0.5 &&  3u^2 + 2u  && 0 \\
0 && 0 && 1
\end{pmatrix} .
\end{equation*}

\vspace{0.5cm}

\noindent The entries 2.0901 in Model 5 and 3.125 in Model 6 ensure smoothness of $A^2(u)$ at $u=0.5$.

We report the estimation results for the models in Figures \ref{fig:all_splitting_plots__model_5} and \ref{fig:all_splitting_plots__model_6}. The plots across the two figures are analogous and hence only those in Figure \ref{fig:all_splitting_plots__model_5} will be explained. The top two plots in Figure \ref{fig:all_splitting_plots__model_5} concern the estimates of $d_0(u)$, and the bottom two plots are those of $d(u)$. The structure of the left plots is similar to that of Figure \ref{fig:splitting_population}. That is, the estimates over $\mathcal{H} = (0,1)$ are depicted at $u=0.5$, over $\mathcal{H} = (0,0.5)$ at $u=0.25$ and those over $\mathcal{H} = (0.5,1)$ at $u=0.75$. The only difference here is that the results are reported over 100 realizations in the form of violin plots. The plots on the right, on the other hand, present the local estimates at $u$'s indicated on the horizontal axis, according to the methods discussed in Section \ref{s:inf-dimens-local}.

Several observations are in place regarding Figures \ref{fig:all_splitting_plots__model_5} and \ref{fig:all_splitting_plots__model_6}. Note that the estimates are generally sensitive to the choice of the subinterval, (0,0.5) or (0.5,1), in the direction of the true dimension value. The estimates of the local dimensions also tend to capture their variability across the different $u$'s.

\begin{figure}[t]
\begin{center}
\includegraphics[scale=0.37]{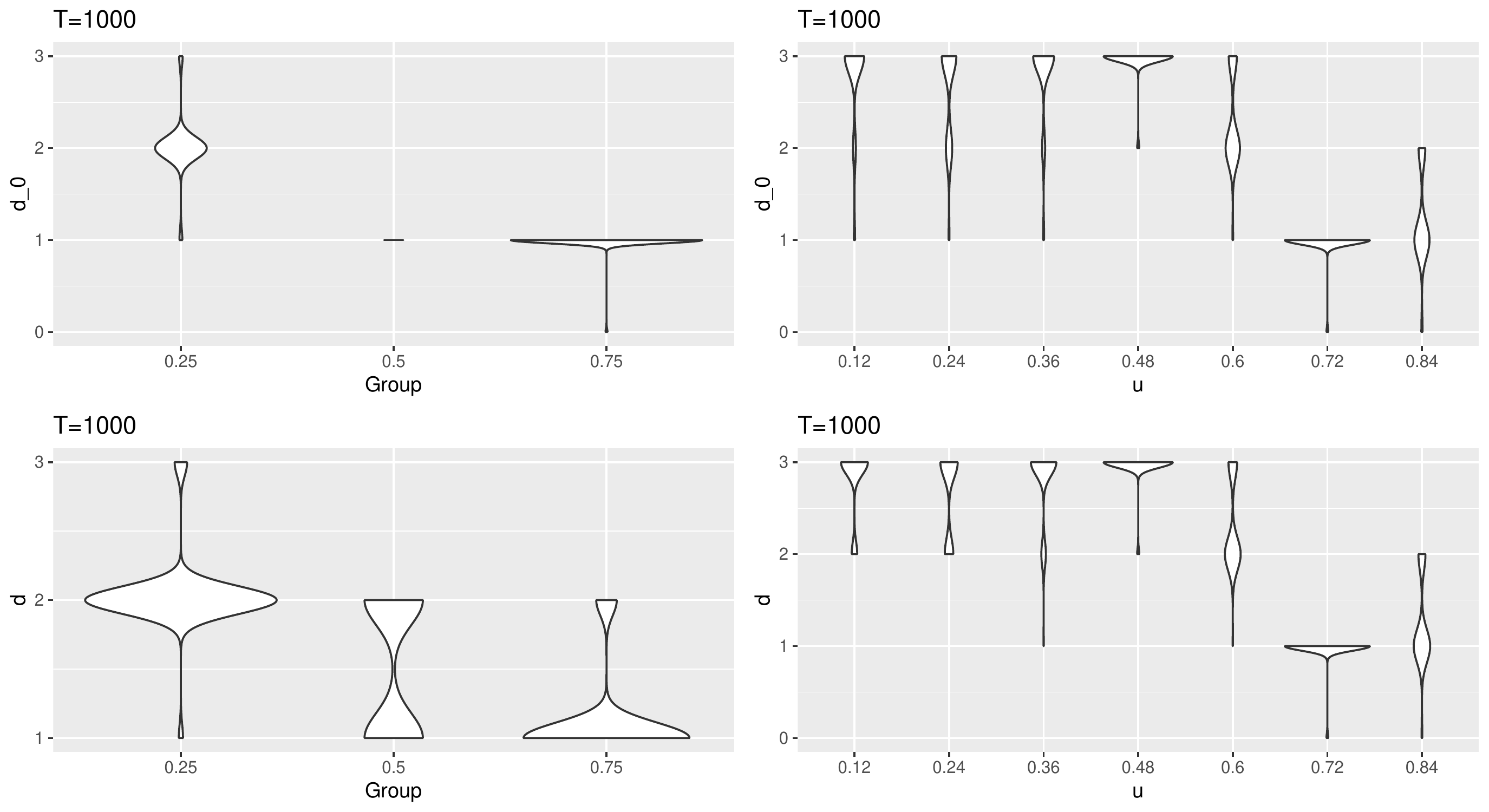}
\end{center}
\caption{Model 5: Top left: Violin plots of the estimates of $d_0(u)$ obtained over (0,0.5), (0,1), (0.5,1) depicted at $u=0.25$, $u=0.5$ and $u=0.75$, respectively. Top right: Violin plots of the local estimates of $d_0(u)$ obtained over indicated $u$ in $(0,1)$. Bottom left and right: Analogous plots but for the estimates of $d(u)$. }\label{fig:all_splitting_plots__model_5}
\end{figure}

\begin{figure}[t]
\begin{center}
\includegraphics[scale=0.37]{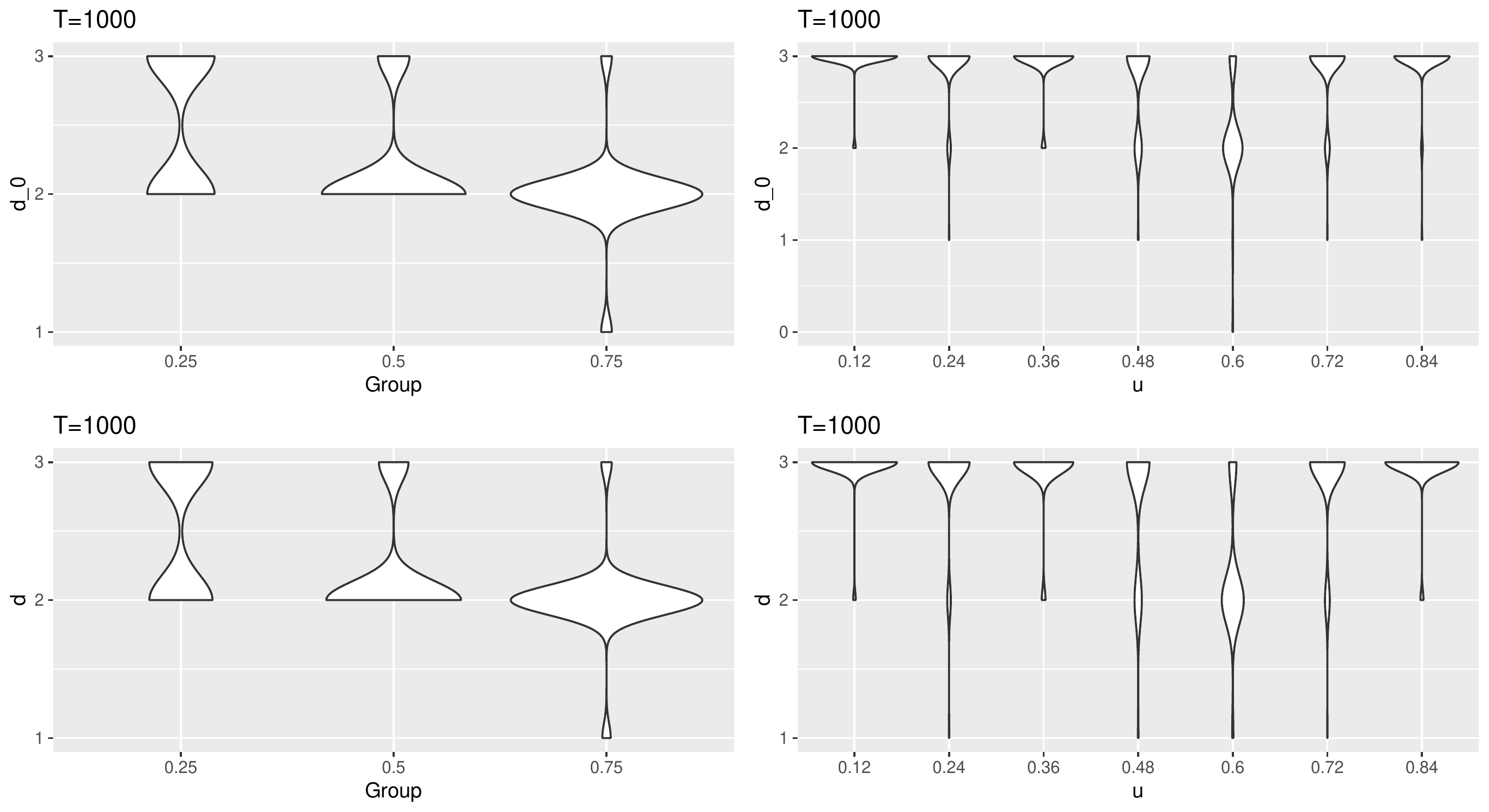}
\end{center}
\caption{Model 6: Top left: Violin plots of the estimates of $d_0(u)$ obtained over (0,0.5), (0,1), (0.5,1) depicted at $u=0.25$, $u=0.5$ and $u=0.75$, respectively. Top right: Violin plots of the local estimates of $d_0(u)$ obtained over indicated $u$ in $(0,1)$. Bottom left and right: Analogous plots but for the estimates of $d(u)$. }\label{fig:all_splitting_plots__model_6}
\end{figure}

\subsection{Subspace estimation comparison}
\label{s:simulation_subspace}

We compare here the performance of the proposed method in estimating a subspace (from Section \ref{s:cluster_subspace}) to DSSA of \cite{sundararajan:2017}, in terms of  three discrepancy measures. The first measure concerns  departure from stationarity based on the sizes of the DFT covariances as given in Eq. (9) of \cite{sundararajan:2017}. More precisely, for an estimated stationary subspace process $Y_t = \widehat{B}_1^{'} X_t$, we set
\begin{equation}\label{e:dssaoptim}
D_1(\widehat{B}_1) \; = \; \sum_{r=1}^{m} \; ||\; \Re ( \widehat{\Gamma}_r^{Y} )\;||_F^2 \; + \;||\; \Im (\widehat{\Gamma}_r^{Y} )\;||_F^2,
\end{equation}
where $|| A ||_F $ denotes the Frobenius norm of a matrix $A \in \mathbb{R}^{d \times d}$, $\Re(\cdot)$ and $\Im(\cdot)$ denote the entrywise real and imaginary parts respectively, and  $\widehat{\Gamma}_r^{Y}$ is the $d \times d$ lag-$r$ DFT sample autocovariance matrix given by
\begin{equation}\label{e:dftcov}
\widehat{\Gamma}_r^Y = \frac{1}{T}\; \sum_{k=1}^{T} J_{Y}(\omega_k) J_{Y}(\omega_{k+r})^{\ast},
\end{equation}
with $\omega_k = 2 \pi k / T$ referring to a Fourier frequency,  $J_Y(\cdot)$ being the discrete Fourier transform (DFT) of the $d$-variate series $ Y_{t}$ and $J_{Y}(\cdot)^{\ast}$ denoting the complex conjugate transpose. The number of DFT covariance lags $m$ in \eqref{e:dssaoptim} is fixed to 3.

The second measure is based on the relation \eqref{e:model-vc-ss-2}. For any candidate subspace $\widehat{B}_1$, we set
\begin{equation}\label{e:vcoptim}
D_2(\widehat{B}_1) = \sum_k || \widehat{B}_1^{'} (A^2(u_k)  -  \overline{A}^2) \widehat{B}_1 ||_F = \\ \notag
\sum_{k} || \widehat{B}_1^{'} ( M(u_k) ) \widehat{B}_1 ||_F.
\end{equation}

The last measure is based on canonical angles computed between the population and estimated subspaces $\mathcal{B}_1$ and $\widehat{\mathcal{B}}_1$ that are spanned by the columns of  the $p \times d$ matrices $B_1$ and $\widehat{B}_1$, respectively. As in \eqref{e:canonical_angle_definition}, let $\theta_1 \leq \theta_2 \leq \hdots \leq \theta_d$ be the $d$ canonical angles between the spaces $\mathcal{B}_1$ and $\widehat{\mathcal{B}}_1$. Then, set

%Let $B_1 \hat{B}_1^{'} = USV$ be the singular value decomposition where the columns of the $d %%\times d$ matrix $U$ contains the singular vectors and $S = diag(\sigma_1,\sigma_2,\hdots,\sigma_d)$ %where $\sigma_1 \geq \sigma_2 \geq \hdots \geq \sigma_d$ are the singular values.  The $d$ canonical %angles, $\theta_1 \leq \theta_2 \leq \hdots \leq \theta_d $, between the subspaces spanned by the rows %of $B_1$ and $\hat{B}_1$ are given by
%\begin{equation}
%cos(\theta_j) \; = \; \sigma_j \; , \; j=1,2,\hdots,d.
%\end{equation}

\begin{equation}
D_3(\widehat{B}_1) = ( \sum_{j=1}^{d} \; \sin^2( \theta_j ) )^{1/2}. 
\end{equation}

In the simulations here, we consider the same four models, Models 1--4, as in Section \ref{s:simulation_dimension}. We first present estimation results for the measure $D_3$ in Table \ref{tab:compare_d3}. For our method, labeled as VC in the table, the stationary subspace estimate is taken as discussed in Section \ref{s:cluster_subspace}. At the population level, the corresponding stationary subspace is  also selected using the same technique. In the case when such multiple subspaces are available at the population level, we compute the distance $D_3$ to all of them and then take the minimum. For the DSSA method, we take the subspace estimate as defined in \cite{sundararajan:2017}.

As seen from the table, the VC method generally performs better than DSSA in yielding smaller average distances over multiple realizations. We stress here again that the VC method is computationally much less intensive than DSSA.

\begin{table}[t]
\begin{center}
\begin{tabular}{ |c|c|c|c|c |c |c|c|c|}
\hline
 & \multicolumn{2}{|c|}{Model 1} & \multicolumn{2}{|c|}{Model 2} & \multicolumn{2}{|c|}{Model 3} & \multicolumn{2}{|c|}{Model 4}  \\
\hline
$T$ & DSSA & VC & DSSA & VC & DSSA & VC & DSSA & VC \\
\hline
\multirow{2}{*}{200} &  0.3277  &  0.3489  &  0.7548  &  0.4963  &  0.3633  &  0.3489 &  0.7453  & 0.6580 \\
 & (0.1429) & (0.1479) & (0.0993) & (0.2348) & (0.1693) & (0.1401)  & (0.0914) & (0.1761) \\
\hline
\multirow{2}{*}{500} &  0.2976  &  0.2878  &  0.6927  &  0.4145  &  0.3575  &  0.3301  &  0.6866  &  0.6196  \\
 & (0.1399) & (0.1288) & (0.1002) & (0.2078) & (0.1657) & (0.1329)  & (0.1348) & (0.1931) \\
\hline
\multirow{2}{*}{1000} &  0.2757  &  0.2968 &  0.6049  &  0.3393 &  0.3241  & 0.2649 &  0.6806  &  0.5655   \\
 & (0.1306) & (0.1301) & (0.1070) & (0.1953) & (0.1453) & (0.1228)  & (0.1175) & (0.2004)  \\
\hline
\multirow{2}{*}{2000} &  0.2678  &  0.2374 &  0.5094  &  0.2751 &  0.2461  &  0.1931 &  0.5932  &  0.4447 \\
 & (0.1323) & (0.1311)  & (0.1015) & (0.1905) & (0.1347) & (0.1202) & (0.1311) & (0.2261)  \\
\hline
\end{tabular}
\end{center}
\caption{Models 1-4: The canonical angle based measure $D_3(\widehat{B}_1)$, between the population and estimated subspace. The empirical standard errors of the respective measures are provided in brackets.} \label{tab:compare_d3}
\end{table}

We now turn to the measures $D_1$ and $D_2$. Here, we shall examine the VC and DSSA methods through these measures from a different perspective. The relevant plots can be found in Figure \ref{fig:m1_compare_d1d2} for Model 1,  associated with the indicated sample sizes. In Figure \ref{fig:m1_compare_d1d2}, top panel, the horizontal solid circles  in the plots represent the measure $D_1(\widehat{B}_1)$  for the DSSA estimate $\widehat{B}_1$, averaged over multiple realizations. The other two curves correspond to the  measure $D_1(\widehat{B}_1(u))$ computed for the estimates of $\widehat{B}_1(u)$ of (1,1)-local stationary subspaces from Section \ref{s:inf-subspace}, either averaged over multiple realizations (triangles) or with the minimum value taken (squares). The interpretation of the plots of Figure \ref{fig:m1_compare_d1d2}, bottom panel, is analogous but for the measure $D_2$. The figures show that VC method performs better than DSSA even ``locally'' for most values of $u$ under measure $D_2$ whereas DSSA performs better under measure $D_1$. Analogous figures for Models 2--4 can be found in \cite{vc-appendix}. The results are similar to Figure \ref{fig:m1_compare_d1d2}, especially for larger sample sizes.

%%%%%%%%%%%%%%%%%%%%%%%%%%%%%%%%%%%%%%%%%%%%%%%%%%%%%%%%%%%%%%%%%%%%%%%%%%%%%%%%%%%%%%
%
%
% % Model 1
%

\begin{figure}[!htb]
\minipage{0.24\textwidth}
  \includegraphics[width=\linewidth]{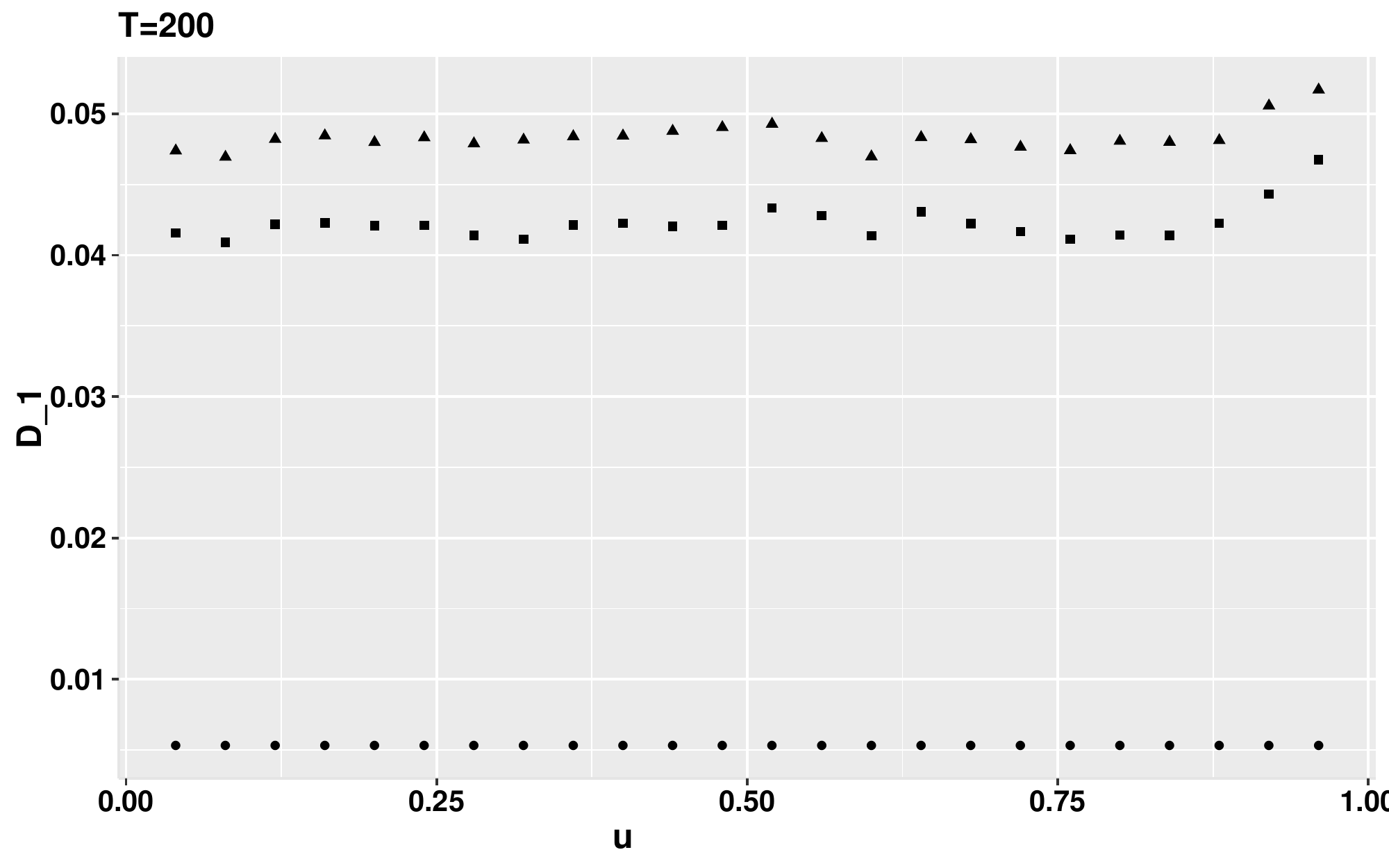}
\endminipage\hfill
\minipage{0.24\textwidth}
  \includegraphics[width=\linewidth]{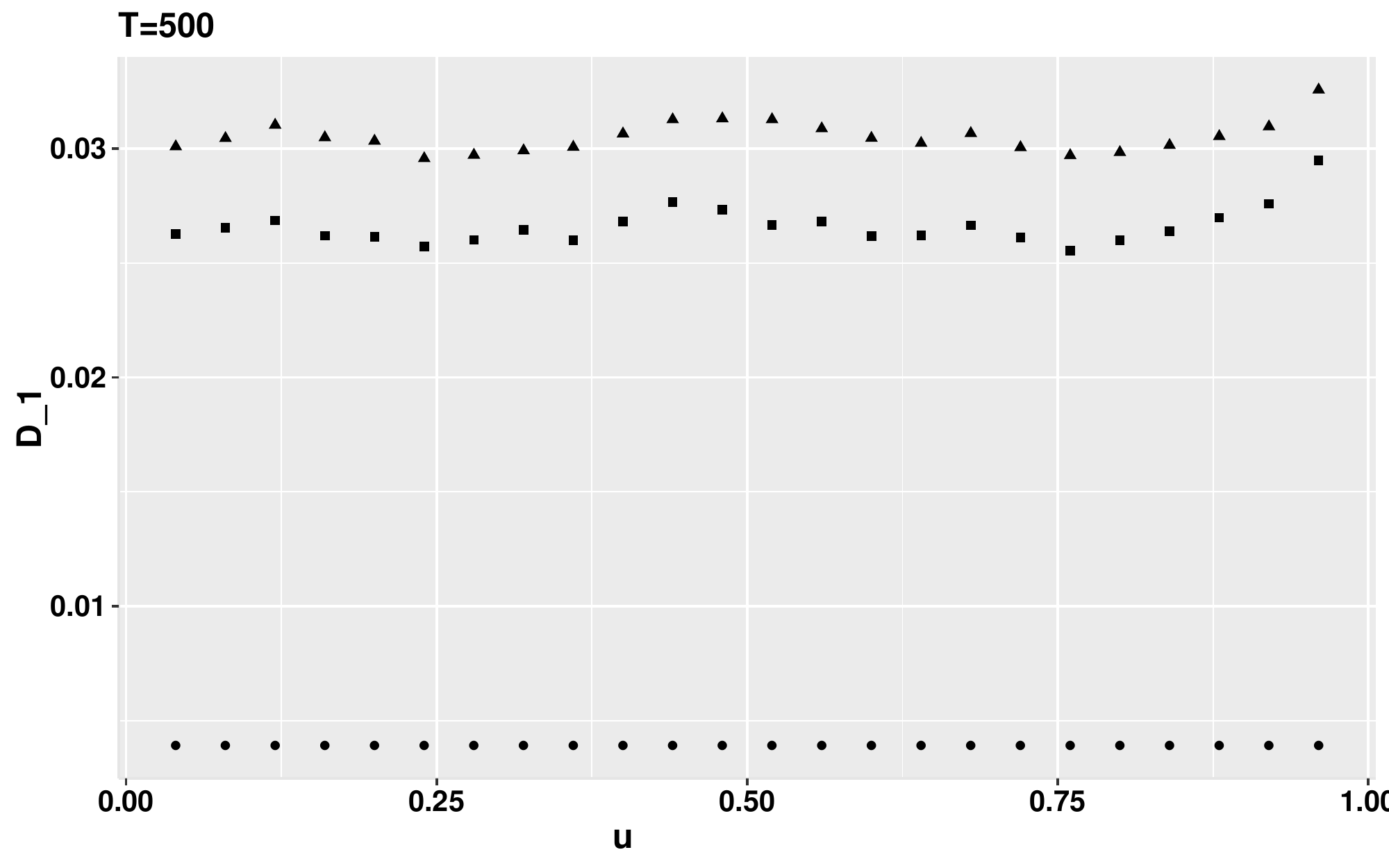}
\endminipage\hfill
\minipage{0.24\textwidth}%
  \includegraphics[width=\linewidth]{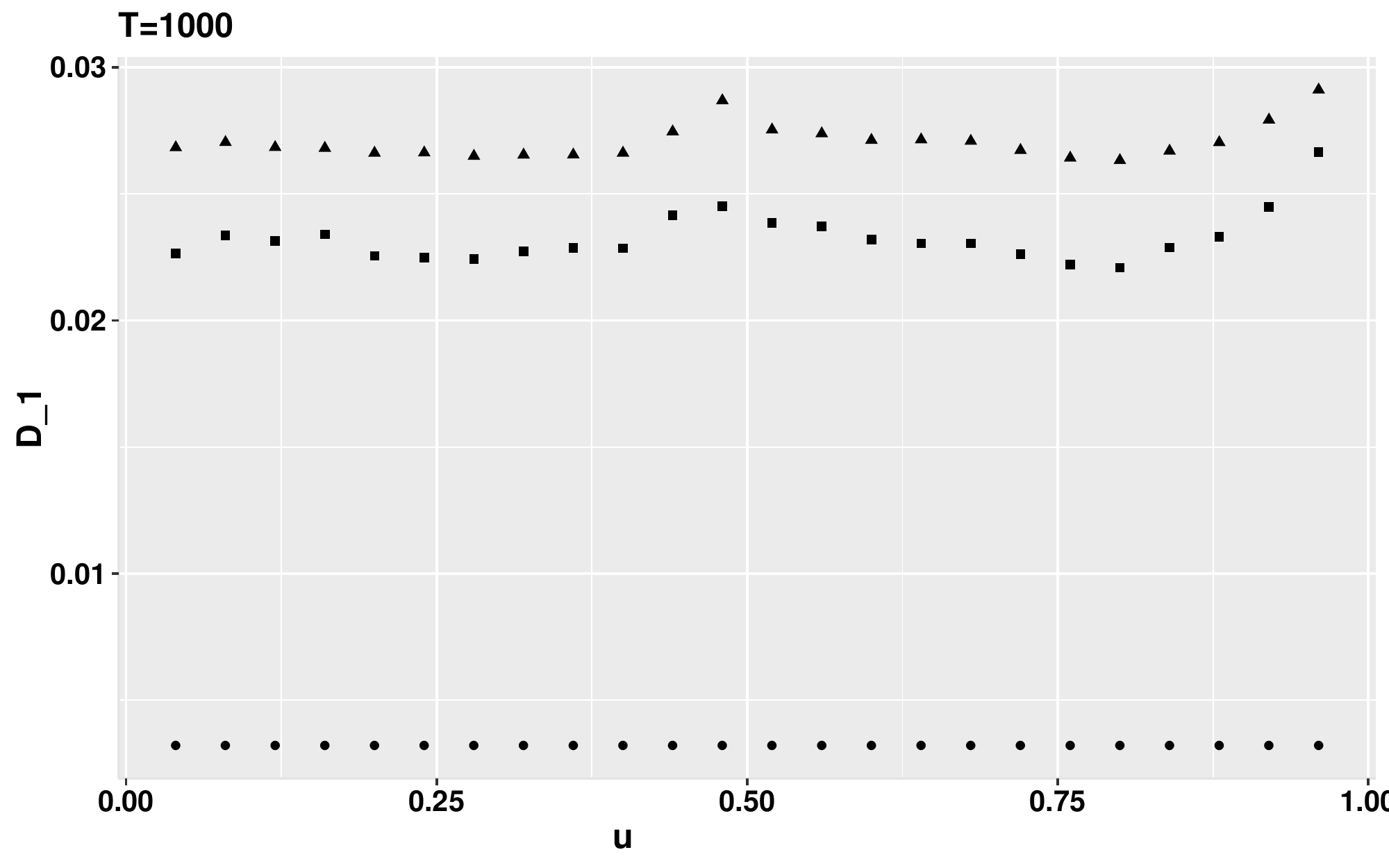}
\endminipage
\minipage{0.24\textwidth}%
  \includegraphics[width=\linewidth]{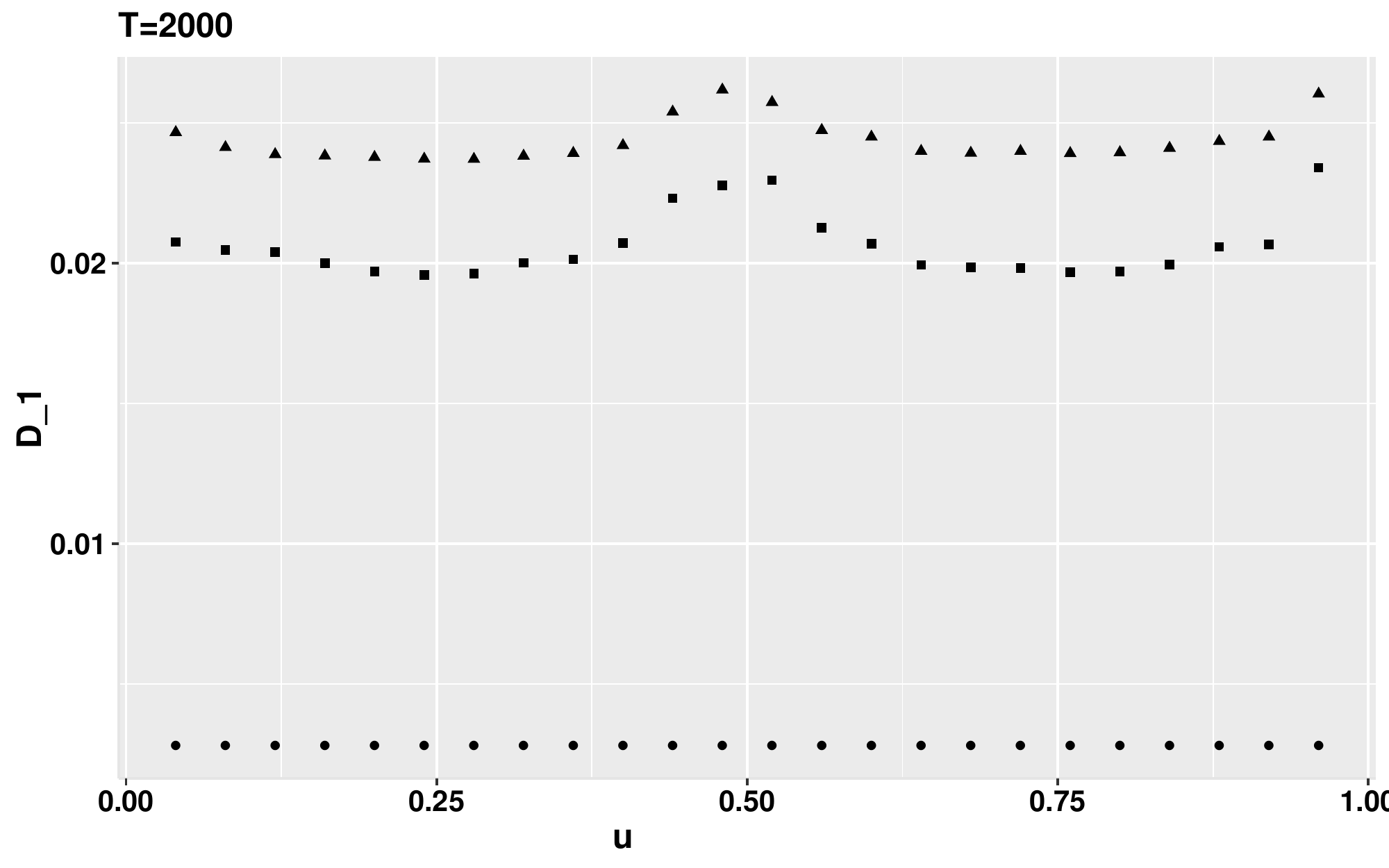}
\endminipage
% \caption{Model 1: Plot of $D_1(\widehat{B}_{1} (u) )$ against $u$ for the competing methods DSSA %and VC and several sample sizes.  } \label{fig:m1_compare_d1}
%\end{figure}
%\vspace{-0.5cm}
\hfill
%\begin{figure}[H]
\minipage{0.24\textwidth}
  \includegraphics[width=\linewidth]{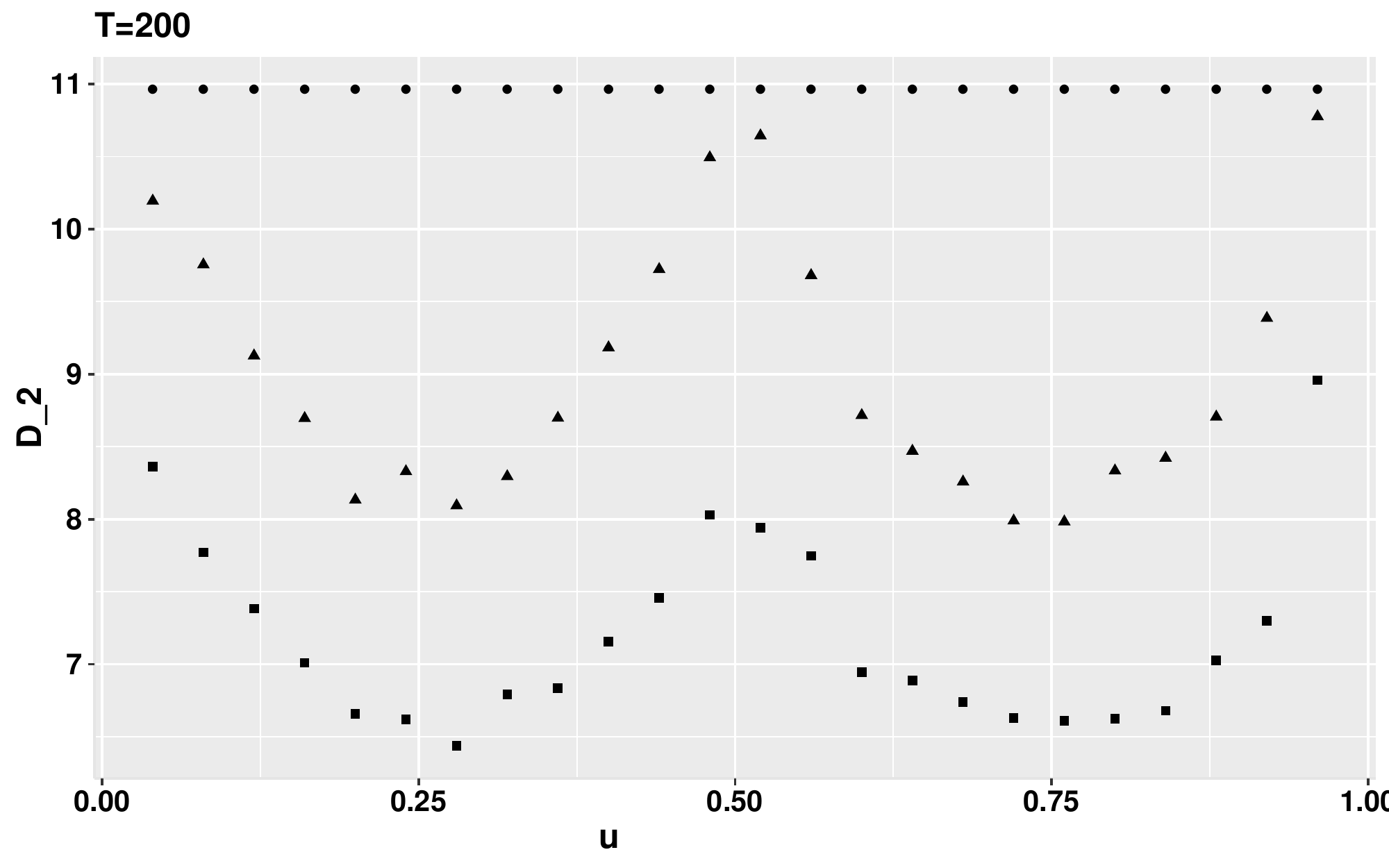}
\endminipage\hfill
\minipage{0.24\textwidth}
  \includegraphics[width=\linewidth]{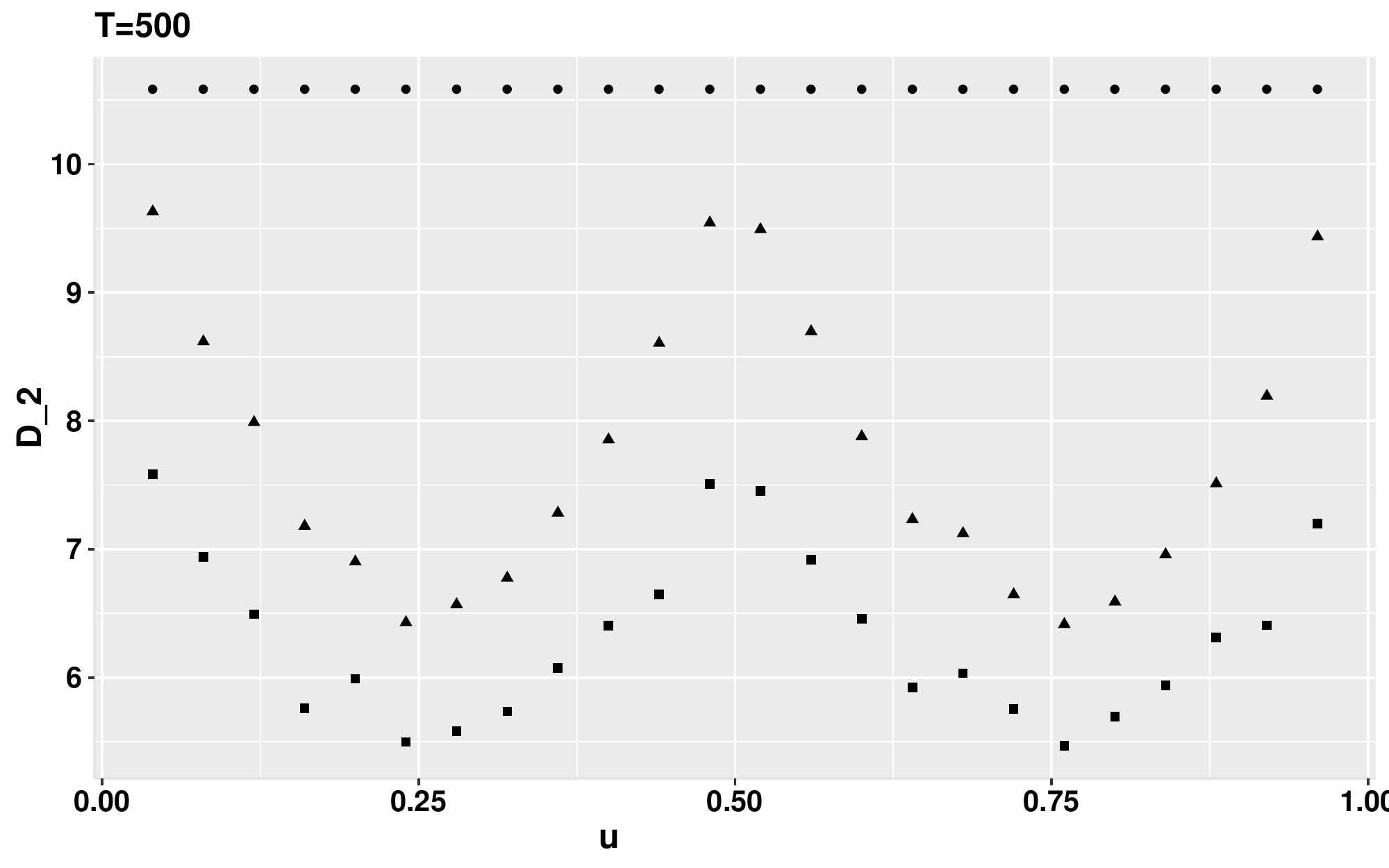}
\endminipage\hfill
\minipage{0.24\textwidth}%
  \includegraphics[width=\linewidth]{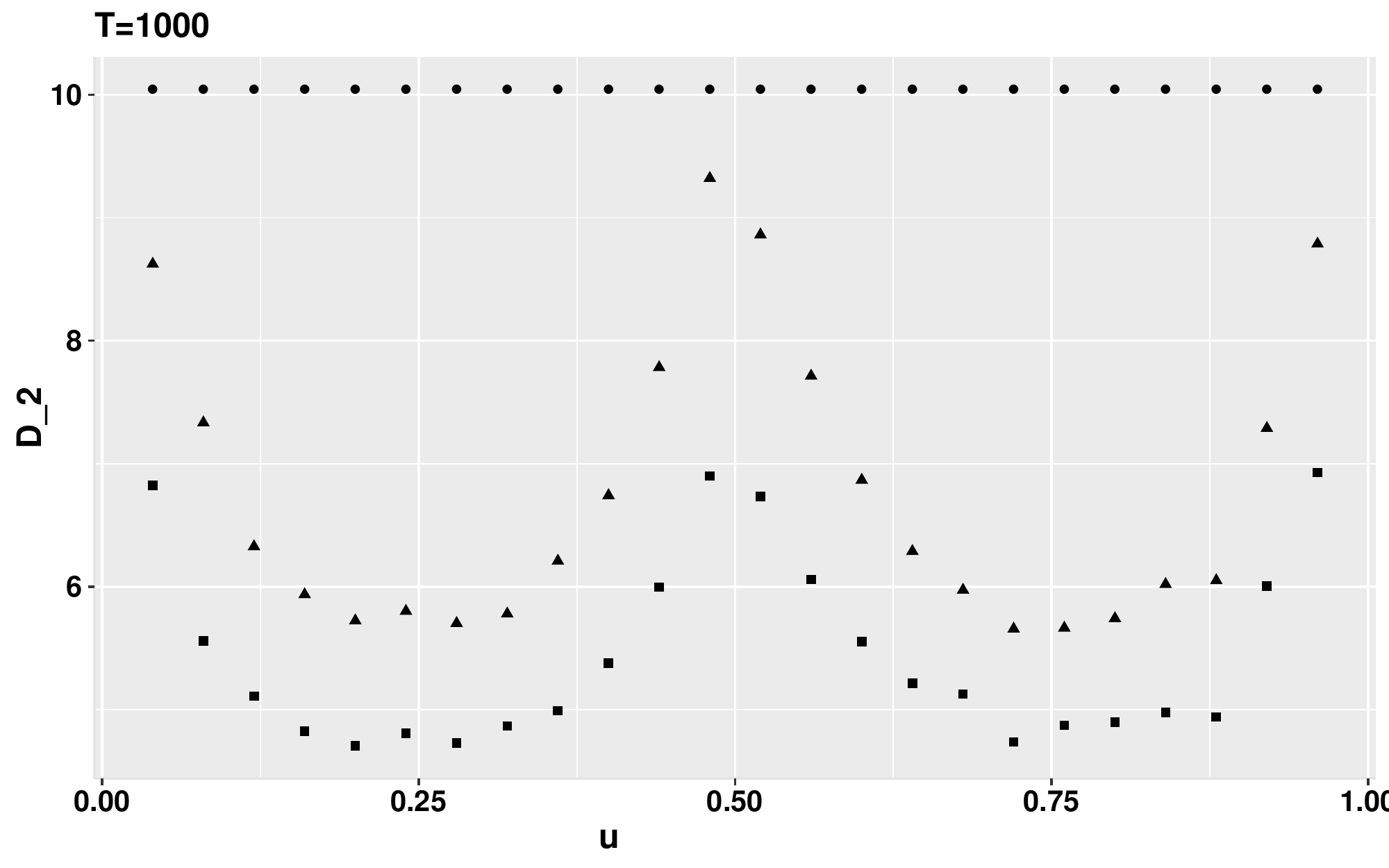}
\endminipage
\minipage{0.24\textwidth}%
  \includegraphics[width=\linewidth]{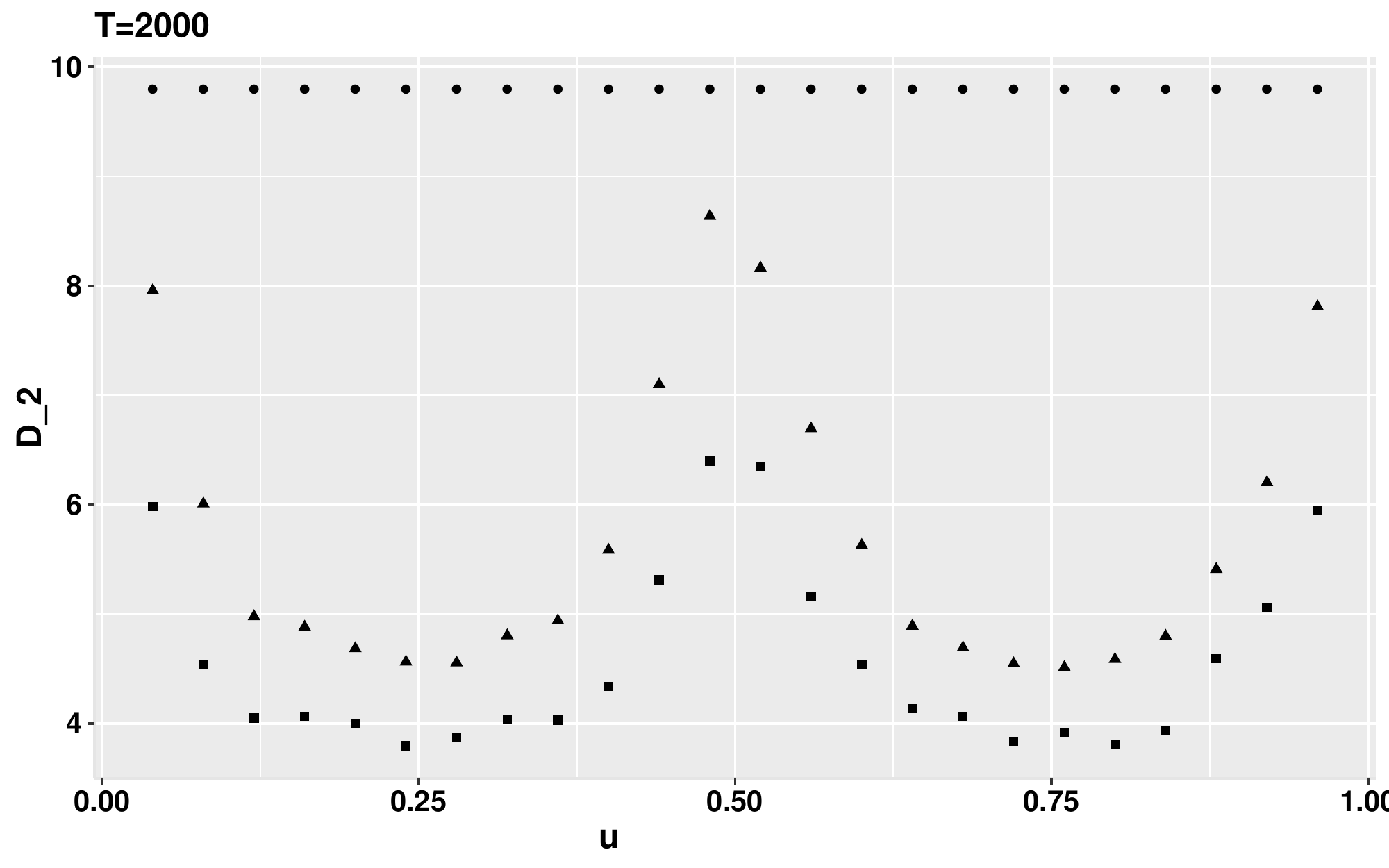}
\endminipage
\caption{Model 1 - Top:  Plot of $D_1(\widehat{B}_{1} (u) )$ against $u$ for the competing methods DSSA and VC and several sample sizes. VC (avg.) in triangles, VC (min.) in squares and DSSA in solid circles. Bottom: Analogous plot but with measure $D_2(\widehat{B}_{1} (u) )$.  } \label{fig:m1_compare_d1d2}
\end{figure}

%
%
% %%%%%%%%%%%%%%%%%%%%%%%%%%%%%%%%%%%%%%%%%%%%%%%%%%%%%%%%%%%%%%%%%%%%%%%
%%%%%%%%%%%%%%%%%%%%%%%%%%%%%%%%%%%%%%%%%%%%%%%%%%%%%%%%%%%%%%%%%%%%%%%% 

\section{Application to BCI and EEG data}
\label{s:application}

Brain-Computer Interface (BCI) aims at connecting the human brain and the computer in a non-invasive manner. During the BCI study used here, individuals are asked to imagine movements with their left and right hands and these are referred to as trials. The trials are interspersed with break periods. The multivariate EEG brain signal is recorded during the entire course of the experiment and the objective is to associate the movements imagined by the individuals with the corresponding EEG signal. Note that the EEG signal is recorded through different channels (locations on the scalp) and each channel constitutes a component of the multivariate signal.

Several works including \citet{slex_ombao}, \citet{ssa09}, \citet{sundararajan:2017} and \citet{sundararajan:neuro} have treated the multivariate EEG signal from such experiments as a nonstationary time series. The advantages of finding a stationary subspace of the observed EEG signal are discussed in the above mentioned works. The key observation made was that nonstationary sources in the brain signal were associated with variations in the mental state that are unrelated to the motor imagery task. Hence, relying solely on stationary components was seen to improve classification accuracy.

We study the classification performance of the proposed  VC method  using the BCI Competition IV \footnote{See http://www.bbci.de/competition/iv/} dataset II in
  \citet{bci4data}. It consists of EEG signals from 9 subjects  performing 4 different
   motor imagery tasks: 1-left hand, 2-right hand, 3-feet and 4-tongue. We analyze
    the EEG signals only from classes 1 and 2 and treat the problem as a two-group
     classification. The continuous signal was sampled at discrete time
     steps at the sampling rate of 250 Hz where 1 second corresponds to 250 observations on the digital signal scale. The signal was recorded through 22 electrodes over the course of the experiment and the signal was band-pass filtered as in \citet{lotte11}. The experiment involved 144 trials for each subject wherein 72 trials belonged to Class 1 (left hand) and the other 72 to Class 2 (right hand). Every trial is followed by an adequate resting period for the subject before the start of the next one. In each trial, we use the observations between 0.5 seconds to 2.5 seconds after the cue instructing the subject to perform the motor imagery task. More precisely, for trial $j$ where $j=1,2 , \hdots , 144$, this interval comprises of 500 observations on the digital signal scale, denoted by $X_t^{(j)}$, $t =1,2,\hdots,500$.

We first restrict attention to 5 EEG electrodes and treat the input signal to have dimension $p=5$. These are 5 locations that  can be viewed as representatives of the different regions on the brain, namely, Frontal (Fz), Pre-Frontal (Pz) and Cortical (C3, C4, Cz).\footnote{See http://www.bbci.de/competition/iv/desc\_ 2a.pdf for additional details on the dataset.} We use the VC method to obtain a $d$ dimensional stationary process where $d<p$. For every trial $j=1,2,\hdots,144$, we have $\{ X_t^{(j)} \} \in \mathbb{R}^{p}$ as the observed multivariate signal and $ \{ Y_t^{(j)} \} \in \mathbb{R}^{d}$ as the stationary transformation.

 We first report on the estimated pseudo dimension $d$ for the observed signals for the 9 subjects in this study labeled S1, S2,$\hdots$, S9. For subjects S3, S5 and S8, the percentage of times the candidate dimensions ($d=0,1,2,3,4$) were estimated by the 2 competing methods, DSSA or VC, out of the 144 trials is provided in the left panel of Figure \ref{fig:bci_dim_plot}. This plot also includes the estimates of $d$ over the first and second halves of the data denoted as VC (First) and VC (Second). It is noted that DSSA always provides a lower estimate of $d$ than the VC method. Similar plots for all 9 subjects can be found \cite{vc-appendix}.

%\begin{table}[H]
%\begin{center}
%\begin{tabular}{|c |c |c |c| c| c |c| c | c| c |}
%\hline
% &  S1 & S2 & S3 & S4 & S5 & S6 & S7 & S8 & S9 \\
%\hline
%VC & 1 & 1 & 1 & 1 & 1 & 1 & 1 & 1 & 1 \\
%\hline
%DSSA & 1 & 1 & 1 & 1 & 1 & 1 & 1 & 1 & 1 \\
%\hline
%\end{tabular}
%\caption{ \footnotesize Mode of the estimated pseudo dimension $\hat{d}$ over 144 trials for %subjects S1 to S9 based on VC Method and DSSA.} \label{tab:mode}
%\end{center}
%\end{table}

\begin{figure}[t]
\begin{subfigure}{.45\textwidth}
\begin{center}
\includegraphics[scale=0.3]{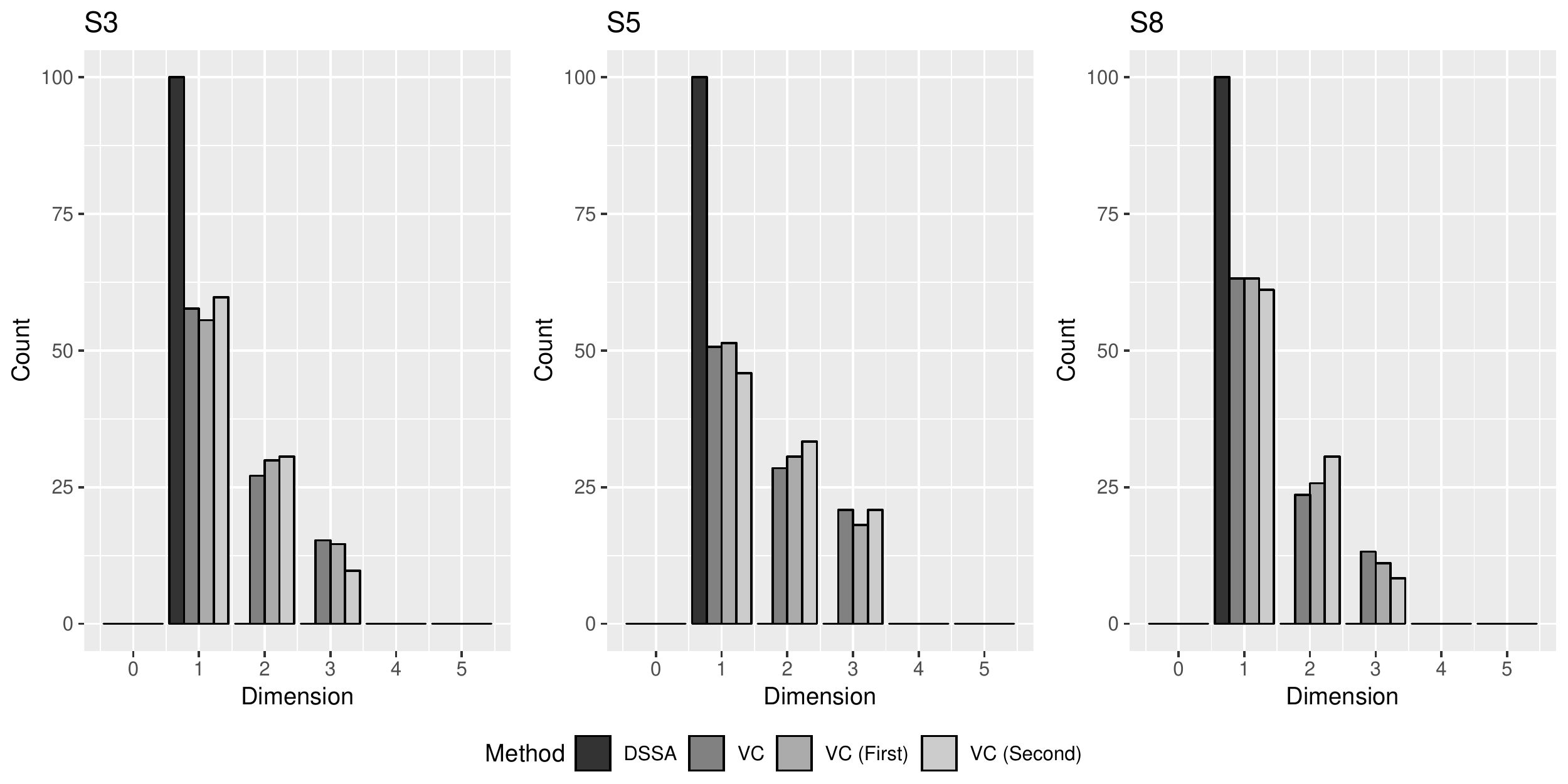}
\end{center}
%\caption{ Histogram of the dimension estimates $d$ by the two competing methods based on the 144 trials.   %} \label{fig:bci_dim_plot}
\end{subfigure}
\begin{subfigure}{.5\textwidth}
\begin{center}
\includegraphics[scale=0.3]{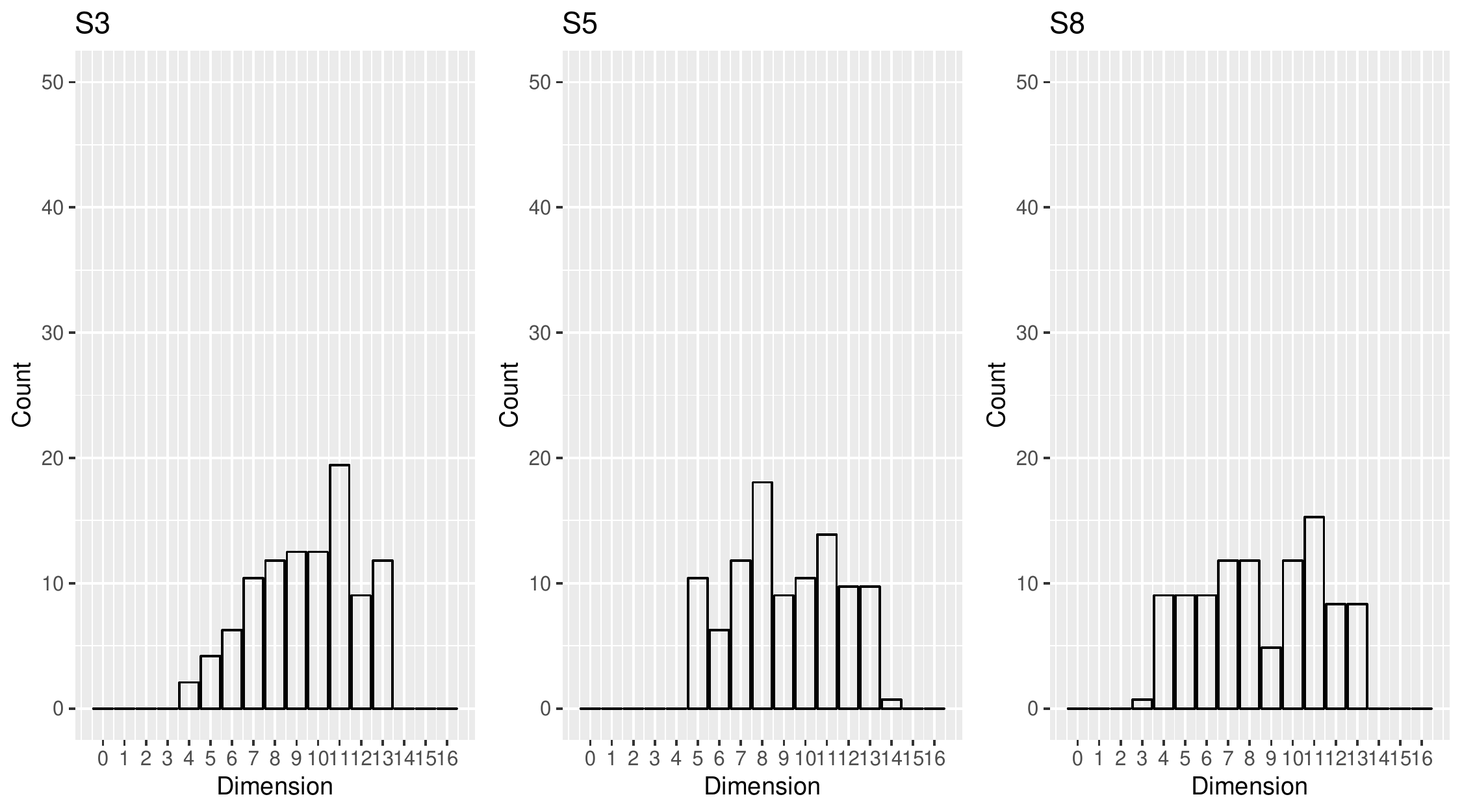}
\end{center}
%\caption{ $p=22$. Histogram of the stationary subspace dimension estimates based on the VC method based %on the 144 trials.   } \label{fig:bci_dim_plot_p-22}
\end{subfigure}
\caption{ Left: ($p=5$) Histogram of the dimension estimates $d$ by the two competing methods based on the 144 trials. Right: ($p=22$) Histogram of the stationary subspace dimension estimates for the VC method based on the 144 trials.   } \label{fig:bci_dim_plot}
\end{figure}

\begin{center}
\begin{table}[ht]
\begin{minipage}{0.5\hsize}
\begin{center}
\begin{tabular}{ | c | c | c | c | c | }
\hline
$d$ &  &  S3 & S5 & S8  \\
\hline
\multirow{2}{*}{$1$} & DSSA  & 56.25  &  49.03   &  51.11   \\
 & VC   &  50.69 & 52.08 & 46.15  \\
\hline
\multirow{2}{*}{$2$} & DSSA  &  54.86  & 54.16 & 56.45  \\
 & VC   &   48.61 & 55.56 &   45.05   \\
\hline
\multirow{2}{*}{$3$} & DSSA  &  59.02  & 59.33 & 64.39  \\
 & VC   &  47.22 & 57.63 &   54.54   \\
\hline
\multirow{2}{*}{$4$} & DSSA  & 56.25 & 62.50 & 66.28  \\
 & VC  &   65.97  & 64.58 &    59.44  \\
\hline
\end{tabular}
  \end{center}
  \end{minipage}
\begin{minipage}{0.45\hsize}
\begin{center}
\begin{tabular}{ |c | c |c |c| c|}
\hline
$d$ &  & S3 & S5 & S8 \\
\hline
\multirow{2}{*}{7} &  DSSA  &  70   & 60.48  & 68.30  \\
 & VC & 70.83 & 73.61 & 68.53 \\
\hline
\multirow{2}{*}{9} & DSSA &  77.90  & 65.50   & 70.38 \\
 & VC &  72.22 & 81.25 &  71.32 \\
\hline
\multirow{2}{*}{11} & DSSA &   72.92  &  69.58   & 71.33 \\
& VC & 74.30 & 88.19 &  81.25 \\
 \hline
\multirow{2}{*}{13} & DSSA &   85.10  &  70.83   & 78.38  \\
& VC & 85.41  & 89.58 &  84.72 \\
\hline
\end{tabular}
\end{center}
\end{minipage}
\caption{Left: Out-of-sample classification accuracy (in \%) for the 3 subjects S3, S5 and S8 for the two
  indicated methods with $p=5$. Right: Out-of-sample classification accuracy (in \%) for 3 subjects S3, S5 and S8 corresponding to $d$= 7, 9, 11, and 13 for the VC method with $p=22$.}  \label{tab:bci_acc_rate}  
\end{table}
\end{center}

%\begin{center}
%\begin{table}[t]
%\begin{center}
%\begin{tabular}{|c |c | c | c |c|}
%\hline
%Distance Measure &  &  S3 & S5 &  S8  \\
%\hline
%\multirow{2}{*}{$D_3$} & Median &  0.75  &  0.74 & 0.74  \\
% & Mean &  0.67 &  0.67 & 0.69  \\
%\hline
%\multirow{2}{*}{$D_4$} & Median & 1.02  & 1 & 1  \\
% & Mean &  0.97 & 0.94 & 1  \\
%\hline
%\end{tabular}
%\caption{  Mean and median distance (in radians) between the estimated subspaces using DSSA and VC %over 144 trials for subjects S3, S5 and  S8. Distances are based on $D_3$ and $D_4$ from Section %\ref{s:simulation_subspace}.} \label{tab:bci_subspace_distance}
%\end{center}
%\end{table}
%\end{center}

Given the $d$-variate stationary
  processes  $Y_t^{(j)}$, we aim to find differences between the two classes (0 and 1) based on the covariance structure. For a given subject, this is achieved by computing the average spectral density matrices for the two classes over the Fourier frequencies:
 \begin{equation} \label{e:app-distance-measure}
\overline{g^i}(\omega_k) \; = \; \frac{1}{n_i} \; \sum_{j \in \textrm{Class } i} g_j(\omega_k), \;\; i=0,1,
\end{equation}
where $g_j(\omega_k)$ is the estimated $d \times d$ spectral matrix for trial $j$ using observations $\{ Y_t^{(j)} \}$, $n_i=72$, for $i=0,1$ and $\omega_k = \frac{2 \pi k}{500}$, $k=1,2,\hdots,500$, are the Fourier frequencies.

In order to train a classifier, for every trial $j=1,2,\hdots,144$, a distance vector  $p_{j,AB} = (p_{0,j,AB},p_{1,j,AB})$ is computed, where
\begin{equation} \label{e:app-logistic-model}
p_{i,j,AB} \; = \frac{1}{250} \; \sum_{k=1}^{250} \; || \; g_j(\omega_k) - \overline{g^i}(\omega_k)  \; ||_F^2, \; \; i=0,1,
\end{equation}
 and $||\cdot||_F$ is the Frobenius norm of a matrix. It measures the distance to the center of each of the two classes. This distance measure serves as our two-dimensional feature vector to be used in constructing
 a logistic regression classifier and assessing its out-of-sample
 classification accuracy.

Table \ref{tab:bci_acc_rate} (left) shows the out-of-sample classification accuracies for three subjects corresponding to  $d = 1, 2, 3, 4 $. A similar table containing results for all 9 subjects can be found in \cite{vc-appendix}. The accuracy rates reflect a comparable average performance in the two competing methods. Finally, the accuracy rate increases as the pseudo dimension $d$ increases from 1 to 4, a phenomenon witnessed and discussed in \citet{ssa09} and \citet{sundararajan:2017}. In dealing with brain signals from healthy individuals in this experiment, the nonstationarity is believed to be caused by artifacts such as as fatigue, physical movement, blinking. Hence, more stationary sources means greater elimination of nonstationary sources in the signal that are unrelated to the experimental task at hand.

The results in Figure \ref{fig:bci_dim_plot} (right) and Table \ref{tab:bci_acc_rate} (right) are analogous to those on their respective left panels but taking the input signal to have dimension $p=22$, that is, without restricting attention to 5 EEG electrodes. For the  out-of-sample classification accuracies, the values  $d = 7, 9, 11, 13 $ are considered. The results in Table \ref{tab:bci_acc_rate} (right) indicate a better performance of VC in comparison to results on the same dataset (with $p=22$) in \citet{sundararajan:2017}.

%\begin{figure}[t]
%\begin{center}
%\includegraphics[scale=0.4]{bci_dim_plot_vc_p=22_only3.pdf}
%\end{center}
%\caption{ $p=22$. Histogram of the stationary subspace dimension estimates based on the VC method based on the 144 trials.   } \label{fig:bci_dim_plot_p-22}
%\end{figure}

%\begin{table}[H]
%\begin{center}
%\begin{tabular}{ |c  |c |c| c|}
%\hline
%$d$ &  S3 & S5 & S8 \\
%\hline
%7 &  70.83 & 73.61 & 68.53   \\
%\hline
%9 &  72.22 & 81.25 &  71.32   \\
%\hline
%11 & 74.30 & 88.19 &  81.25  \\
%\hline
%13 & 85.41  & 89.58 &  84.72  \\
%\hline
%\end{tabular}
%\caption{ Right: Out-of-sample classification accuracy (in \%) for 3 subjects S3, S5 and S8 corresponding %to $d$= 7, 9, 11, and 13 for the VC method with $p=22$.}  \label{tab:bci_acc_rate}
%\end{center}
%\end{table}

It must be noted that the results in Table \ref{tab:bci_acc_rate} is based on an estimated stationary subspace for a pre-fixed dimension $d$. The estimates of the dimensions of the stationary subspaces from \eqref{e:global-d}, namely $\widehat{d}_0$, $\widehat{d}_+$ and $\widehat{d}_-$, could potentially be different across different trials. For each trial, we obtained these estimates and compared them with the fixed $d$ and investigated results for several different combinations of  $\widehat{d}_0$, $\widehat{d}_+$ and $\widehat{d}_-$ that results in $d = \widehat{d}_0 + \min(\widehat{d}_+ , \widehat{d}_-)$. Having witnessed very similar results for the various combinations, we only present, in Table \ref{tab:bci_acc_rate}, the case wherein $d = \widehat{d_0}$.

\section{Concluding remarks}
\label{s:conclusion}

Our goal in this work is to (i) study existence of  linear combinations of components of a multivariate nonstationary
process which are stationary, and (ii) find the number of such stationary linear combinations. The true nature of the problem and richness of
its solution present themselves naturally when the general dependence setup in \cite{sundararajan:2017} is abandoned in favor of heterogenous independent
observations. In this simplified setup (\citet{ssa09}), solution of  the problem reduces to the  study of inertia or signs of the eigenvalues and the corresponding eigenvectors of certain symmetric time-varying matrices constructed  from varying covariance or heterogeneity of the vector observations.
This enable us to provide a direct linear-algebraic method to construct stationary subspaces which
outperforms the earlier computationally more expensive optimization-based SSA solutions.

Several directions related to this work could be explored in the future. The developed framework involving pseudo spaces and dimensions is general enough to apply to zero mean locally stationary processes, when working with their time-frequency spectra. Another possibility is to try to combine the models of this work and \citet{ssa-means}, so that both mean and the covariance are allowed to vary in time. Yet another direction is to explore connections to multivariate stochastic volatility models that are concerned with modeling the changing covariance across time. 

%\begin{remark}\label{r:inf-special-case-no-kernel}

%In a simpler setup where $M(u)$ from equation \eqref{e:models-M} is given by
%\begin{equation}\label{e:inf-M-special-case}
%M(u) = \int_0^{0.5} A^2(u) du - \int_{0.5}^1 A^2(u) du,
%\end{equation}
%the equivalent to condition \eqref{e:model-vc-ss-2} becomes
%\begin{equation}\label{e:inf-vc-special-case}
%B_1 ( \; \int_0^{0.5} A^2(u) du - \int_{0.5}^1 A^2(u) du \; ) B_1^{'} = 0.
%\end{equation}
%Note that here the estimator for the matrix $M(u)$ does not require a kernel and it's associated %problems such as bandwidth selection. \citet{ssa09} involves a setup wherein the time interval is %partitioned into $N$ blocks and the aim is to reach a stationary process that has constant variance across the $N$ blocks. The formulation in \eqref{e:inf-vc-special-case} is similar wherein the goal is to find the stationary subspace that leads to equal variance in the first and second half (two blocks) of the time interval.
%\end{remark}

\appendix
\section{Proofs}
\label{s:proofs}

The appendix concerns the technical aspects of this work, including various assumptions and proofs of some of the stated results.

\subsection{Assumptions and proof of Proposition \ref{p:M-loctest-vc-an}}
\label{s:proofs-local}

We use the following assumptions for Proposition \ref{p:M-loctest-vc-an}, labeled according to the quantities they concern.  
\begin{itemize}
	\item [(Y1)] $Y_t$, $t=1,\ldots,T$, are i.i.d.\ random vectors with i.i.d.\ entries, $\EE(Y_t)=0$ and $\EE(Y_t Y_t') = I_p$. 
	\item [(Y2)] The entries of $Y_t$ have finite absolute moments of order $4+\epsilon$ for some $\epsilon>0$.
	\item [(K)] The kernel function $K$ is even (i.e.\ $K(u)=K(-u)$, $u\in\RR$), has bounded support $(-S,S)$, where it is positive, and is continuously differentiable on $(0,S)$. Furthermore, $\int_\RR K(u) du =1$.
	\item [(A1)] The matrix $A^2(u)$ is positive definite for $u\in (0,1)$.
	\item [(A2)] The matrix $A^2(u)$ is continuously differentiable for $u\in (0,1)$.  
\end{itemize}
We note that (A1) and (A2) imply the continuous differentiability of $A(u)$, $A(u)^{-1}$ and $A(u)^{-2}$ for $u\in (0,1)$.

\medskip

\noindent {\sc Proof of Proposition \ref{p:M-loctest-vc-an}:} In view of (\ref{e:models-M}) and (\ref{e:model-vc-Mu-estim}), write
\begin{equation}\label{e:proofs-local-MS}
	\widehat M(u) - M(u) = S_1 + S_2 - S_3 - S_4,	
\end{equation}
where 
\begin{eqnarray}
	S_1 & = & \frac{1}{T}\sum_{t=1}^T (X_tX_t' - A^2(\frac{t}{T})) K_h(u-\frac{t}{T}),\nonumber \\
	S_2 & = & \frac{1}{T}\sum_{t=1}^T A^2(\frac{t}{T}) K_h(u-\frac{t}{T}) - A^2(u),\nonumber \\
	S_3 & = & \frac{1}{T}\sum_{t=1}^T (X_tX_t' - A^2(\frac{t}{T})),\nonumber \\
	S_4 & = & \frac{1}{T}\sum_{t=1}^T A^2(\frac{t}{T}) - \int_0^1 A^2(u) du. \label{e:proofs-local-MS-list}
\end{eqnarray}
We will show that $(Th)^{1/2}S_j\to 0$ for $j=2,4$, and $\to_p 0$ for $j=3$, and that $S_1$ yields the normal limit in (\ref{e:M-loctest-vc-an}). This shall prove the convergence (\ref{e:M-loctest-vc-an}). In what follows, $C$, $C'$ and so on, denote generic positive constants that can change from line to line.

Since the derivatives of the entries of $A(u)$ (and $A^2(u)$) are assumed continuous and hence bounded on $(0,1)$, it follows that $|S_4|\leq C/T$ and hence indeed $(Th)^{1/2}S_4\to 0$, since $T\to\infty$ and $h\to 0$. Since by assumptions, $Y_t$ are i.i.d.\ and their entries have finite fourth moments, and $A(u)$ is bounded on $(0,1)$, we have
$$
\EE \|S_3\|_2^2 = \frac{1}{T^2} \EE \Big\| \sum_{t=1}^T A(\frac{t}{T}) (Y_tY_t'-I_p)A(\frac{t}{T})'\Big\|_2^2 \leq \frac{C}{T^2} \sum_{t=1}^T  \EE \|Y_t Y_t'-I_p\|_2^2 = \frac{C'}{T}.
$$
Thus, $S_3 = O_p(T^{-1/2})$ and $(Th)^{1/2}S_3 = O_p(h^{1/2})=o_p(1)$. For $S_2$, write $S_2 = S_{2,1}+S_{2,2}$, where 
$$
S_{2,1} = \frac{1}{T} \sum_{t=1}^T (A^2(\frac{t}{T}) - A^2(u)) K_h(u-\frac{t}{T}),\quad S_{2,2} = A^2(u) \Big( \frac{1}{T} \sum_{t=1}^T  K_h(u-\frac{t}{T}) - 1 \Big).
$$
Since the entries of $A^2(u)$ have bounded first derivatives, we have
$$
|S_{2,1}| \leq \frac{C}{T} \sum_{t=1}^T |u - \frac{t}{T}| K_h(u-\frac{t}{T}) = \frac{Ch}{T} \sum_{t=1}^T G_{1,h}(u-\frac{t}{T}) = O(h),
$$
where $G_1(u) = |u| K(u)$, $G_{1,h}(u) = h^{-1}G_1(h^{-1}u)$ and we used Lemma \ref{l:sum-int-kernel} below in the last step. Thus, $(Th)^{1/2}S_{2,1}\to 0$ since $Th^3\to 0$ by assumption. Similarly, $S_{2,2}=O((Th)^{-1})$. 

Finally, we show that $S_1$ gives the desired asymptotic limit. For this, we consider $A(u)^{-1}S_1A(u)'^{-1}$ and express it as
$$
A(u)^{-1}S_1A(u)'^{-1} = S_{1,1} + S_{1,2} + S_{1,3} + S_{1,4},
$$
where 
\begin{eqnarray*}
	S_{1,1} & = & \frac{1}{T} \sum_{t=1}^T (A(u)^{-1} - A(\frac{t}{T})^{-1}) (X_tX_t'-A^2(\frac{t}{T})) (A(u)'^{-1} - A(\frac{t}{T})'^{-1}) K_h(u-\frac{t}{T}), \\
	S_{1,2} & = & \frac{1}{T} \sum_{t=1}^T A(\frac{t}{T})^{-1} (X_tX_t'-A^2(\frac{t}{T})) (A(u)'^{-1} - A(\frac{t}{T})'^{-1}) K_h(u-\frac{t}{T}), \\
	S_{1,3} & = & \frac{1}{T} \sum_{t=1}^T (A(u)^{-1} - A(\frac{t}{T})^{-1}) (X_tX_t'-A^2(\frac{t}{T})) A(\frac{t}{T})'^{-1} K_h(u-\frac{t}{T}), \\
	S_{1,4} & = & \frac{1}{T} \sum_{t=1}^T A(\frac{t}{T})^{-1} (X_tX_t'-A^2(\frac{t}{T}))  A(\frac{t}{T})'^{-1} K_h(u-\frac{t}{T})  = \frac{1}{T} \sum_{t=1}^T (Y_tY_t'-I_p)  K_h(u-\frac{t}{T}).
\end{eqnarray*}
We will argue that $S_{1,1}$, $S_{1,2}$ and $S_{1,3}$ are asymptotically negligible for the limit, and that $S_{1,4}$ gives the desired normal limit.

For $S_{1,1}$, let $K_2(u)=K^2(u)/\|K\|_2^2$ be a kernel function obtained from the square of $K$ and $K_{2,h}(u) =  h^{-1} K_2( h^{-1} u)$. As in dealing with $S_2$ above, note that $\EE\|X_tX_t' - A^2(t/T)\|_2^2$ is bounded by a constant uniformly over $t=1,\ldots,T$. Then,
$$
\EE \|S_{1,1}\|_2^2 = \frac{1}{T^2}\sum_{t=1}^T \EE \Big\| (A(u)^{-1} - A(\frac{t}{T})^{-1}) (X_tX_t'-A^2(\frac{t}{T})) (A(u)'^{-1} - A(\frac{t}{T})'^{-1}) \Big\|_2^2 K_h^2(u-\frac{t}{T})
$$
$$
\leq \frac{\|K\|_2^2}{Th} \frac{1}{T} \sum_{t=1}^T \EE 
\Big\| X_tX_t'-A^2(\frac{t}{T}) \Big\|_2^2 
\Big\| A(u)^{-1} - A(\frac{t}{T})^{-1} \Big\|_2^4 K_{2,h}(u-\frac{t}{T})
$$
$$
\leq \frac{C}{Th} \frac{1}{T} \sum_{t=1}^T 
\Big\| A(u)^{-1} - A(\frac{t}{T})^{-1} \Big\|_2^4 K_{2,h}(u-\frac{t}{T})
$$
$$
\leq \frac{C}{Th} \frac{1}{T} \sum_{t=1}^T 
(u - \frac{t}{T})^2 K_{2,h}(u-\frac{t}{T}) = \frac{Ch}{T} \frac{1}{T} \sum_{t=1}^T 
G_{2,h}(u-\frac{t}{T}),
$$
where $G_2(u)=u^2K_2(u)$, $G_{2,h}(u) = h^{-1} G_2( h^{-1} u)$ and we used the assumed smoothness of $A(u)^{-1}$. By using Lemma \ref{l:sum-int-kernel} below, we have $S_{1,1}=O_p((h/T)^{1/2})$ and hence $(Th)^{1/2}S_{1,1}=O_p(h)=o_p(1)$. One can show similarly that 
$S_{1,2}=O_p(1/T^{1/2})$ and $S_{1,3}=O_p(1/T^{1/2})$, and hence that $(Th)^{1/2}S_{1,2}$ and $(Th)^{1/2}S_{1,3}$ are negligible as well. We next analyze the last term $S_{1,4}$.

To establish the asymptotic normality of $(Th)^{1/2}S_{1,4}$, it is enough to check the Lyapunov condition and to make sure that the variances converge to the desired quantities. For the Lyapunov condition, note that, with $\delta>0$,
$$
\sum_{t=1}^T \EE \Big\| \frac{(Th)^{1/2}}{T} (Y_tY_t'-I_p) K_h(u-\frac{t}{T}) \Big\|_{2+\delta}^{2+\delta} \leq \frac{C(Th)^{1+\delta/2}}{T^{2+\delta}} \sum_{t=1}^T  K^{2+\delta}_h(u-\frac{t}{T}) 
$$ 
$$
= \frac{C'}{(Th)^{\delta/2}} \frac{1}{T} \sum_{t=1}^T K_{2+\delta,h}(u-\frac{t}{T}), 
$$
where $K_{2+\delta}(u)=K(u)^{2+\delta}/\|K\|_{2+\delta}^{2+\delta}$. The last bound converges to $0$ since $Th\to\infty$ and since the sum in the bound converges to $\int_\RR K_{2+\delta}(u)du$ by Lemma \ref{l:sum-int-kernel} below. For the convergence of the variances, note that the different entries of the symmetric matrix $Y_tY_t'-I_p$ are independent, which translates into them being uncorrelated in the limit, as stated in the proposition. We thus need to consider the variances of 
$$
\frac{(Th)^{1/2}}{T} \sum_{t=1}^T (Y_{1,t}^2 - 1) K_h(u-\frac{t}{T}) =: S_{\scriptsize \rm diag}, \; 
\frac{(Th)^{1/2}}{T} \sum_{t=1}^T Y_{1,t} Y_{2,t} K_h(u-\frac{t}{T}) =: S_{\scriptsize \rm off}
$$
for the diagonal and off-diagonal asymptotic variances, respectively. For $S_{\scriptsize \rm diag}$, we have
$$
\EE S_{\scriptsize \rm diag}^2 = \frac{Th}{T^2} \sum_{t=1}^T \EE(Y_{1,t}^2 - 1)^2 K_h^2(u-\frac{t}{T}) = \mu_4 \|K\|_2^2 \sum_{t=1}^T K_{2,h}(u-\frac{t}{T}), 
$$
where again, the kernel $K_2(u) = K^2(u)/\|K\|_2^2$. By Lemma \ref{l:sum-int-kernel} below, $\EE S_{\scriptsize \rm diag}^2 \to \mu_4 \|K\|_2^2$, which matches the diagonal variance in (\ref{e:M-loctest-vc-an}). For $S_{\scriptsize \rm off}$, we have similarly
$$
\EE S_{\scriptsize \rm off}^2 = \frac{Th}{T^2} \sum_{t=1}^T \EE Y_{1,t}^2 \EE Y_{2,t}^2 K_h^2(u-\frac{t}{T}) = \|K\|_2^2 \sum_{t=1}^T K_{2,h}(u-\frac{t}{T})\to \|K\|_2^2.  
$$
This concludes the proof of the convergence (\ref{e:M-loctest-vc-an}).

Finally, the consistency of $\widehat A^2(u)$ follows from the proved convergence to $0$ of the terms $S_1$ and $S_2$ above. This implies the consistency of $\widehat A^{-2}(u)$ since $A^2(u)$ is assumed to be positive-definite. Since the square root operation is a continuous functional, we also get the consistency of $\widehat A(u)$ and $\widehat A^{-1}(u)$.
\quad \quad $\Box$

\medskip

Part $(i)$ of the following auxiliary lemma was used in the proof above. Its part $(ii)$ will be used in Appendices \ref{s:proofs-global} and \ref{s:proofs-global-2} below.

\begin{lemma}\label{l:sum-int-kernel}
Let $G:\RR\to\RR$ be even ($G(-v)=G(v)$, $v\in \RR$), have bounded support $(-S,S)$ and be continuously differentiable on $(0,S)$. Let also $G_h(v) = \frac{1}{h}G(\frac{v}{h})$, $v\in \RR$, with $h>0$ such that $h\to 0$ and $Th\to\infty$. Then:

$(i)$ for $u\in (0,1)$,
$$
\Big| \frac{1}{T} \sum_{t=1}^T G_h(u-\frac{t}{T}) - \int_{-S}^S G(v) dv\Big| \leq \frac{C}{Th}. 
$$

$(ii)$ The bound also hold uniformly for $u\in {\cal H}$, where ${\cal H}$ is a closed subinterval in $(0,1)$.
\end{lemma}

\begin{proof}
The result $(i)$ follows from: for large enough $T$,
$$
\frac{1}{T} \sum_{t=1}^T G_h(u-\frac{t}{T}) = \frac{1}{Th} \sum_{t=1}^T G\Big(\frac{u-\frac{t}{T}}{h}\Big) = \frac{1}{Th} \sum_{t=1}^T G\Big(\frac{Tu-t}{Th}\Big) 
$$
$$
= \frac{1}{Th} \sum_{s=T(u-1)}^{Tu-1} G\Big(\frac{s}{Th}\Big) =  \sum_{s=-STh}^{STh} G\Big(\frac{s}{Th}\Big) \frac{1}{Th},
$$
since, for example, $Tu-1>STh$ for large enough $T$. The last sum above converges to $\int_{-S}^S G(v) dv$ no slower than the given rate since $Th\to\infty$ and $G$ is continuously differentiable. The part $(ii)$ follows in the same way since $Tu-1>STh$ for large enough $T$ uniformly for $u\in{\cal H}$.
\end{proof}

\subsection{Assumptions and proof of Proposition \ref{p:global-an}}
\label{s:proofs-global}

The proof of Proposition \ref{p:global-an} follows the path taken by \cite{donald:2011local}; see, in particular, its technical appendix online. We shall try to minimize repetitions by indicating only the key needed assertions. Some of the developments will be somewhat simpler since many of the considered matrices are symmetric. But we shall also need several new auxiliary results to account for the key difference from \cite{donald:2011local} in that smoothing through a kernel is carried out here in time $t$ while this was over the values of a random variable in \cite{donald:2011local}. Some of the auxiliary results for the proof of Proposition \ref{p:global-an} can be found in Appendix \ref{s:proofs-global-2} and the reader interested in proofs might want to look at them first, before going through the arguments in this section. 

We assume implicitly as stated in Section \ref{s:inf-dimens-global} that ${\cal H}$ is a closed subinterval of $(0,1)$. We shall also need to strengthen Assumption (A2) in Appendix \ref{s:proofs-local} to:
\begin{itemize}
	\item [(A2g)] The entries of the matrix $A^2(u)$ are real analytic for $u\in {\cal H}$.
\end{itemize}
According to one possible definition, a function $f$ is real analytic if its Taylor series converges to the function $f$ in a neighborhood of each point. As noted for similar assumptions G4, G5 in \cite{donald:2011local}, analyticity is assumed to have smoothness of the eigenvectors of analytic matrices involving $A^2(u)$. It is well known that smoothness of a matrix is not sufficient to have smooth eigenvectors (see, e.g., \cite{kato:1976}, \cite{bunse:1991numerical}). Alternatively, the smoothness of the eigenvectors of interest can be assumed. 
 
A number of comments concerning the  notation are also in place. To simplify the notation, we shall drop the dependence on $u$, and write $\widehat A$, $A$, $\widehat M$, $M$, $\widehat \xi_r$, $\widehat \gamma_{2,i}$, etc.\ instead of $\widehat A(u)$, $A(u)$, $\widehat M(u)$, $M(u)$, $\widehat \xi_r(u)$, $\widehat \gamma_{2,i}(u)$, etc.\ Similarly, $\sup$ will denote $\sup_{u\in {\cal H}}$ and $\xi = O_{p,{\scriptsize \rm sup}}(b_T)$ will stand for $\sup|\xi| = O_p(b_T)$. As throughout this work, $G_h(u)$ will stand for a scaled kernel function $h^{-1}G(h^{-1}u)$. The kernel function $G$ will be normalized to integrated to $1$ when it is important. We shall use both $K$ in Assumption (K) of Appendix \ref{s:proofs-local} and other functions related to $K$, which will be denoted as
$$
\overline K(u) = \int_\RR K(u-v) K(v) dv,\quad K^{(q)}_p(u) = C |u|^q K^p(u),
$$
where $C$ is such that $\|K^{(q)}_p\|_1 = 1$. Note that $K^{(0)}_1(u)=K(u)$. When $q=0$, we shall simply write $K_{p}(u)$.

\medskip
\noindent {\sc Proof of Proposition \ref{p:global-an}:} We shall first prove (\ref{e:global-an-1}). The first step consists of showing that the ordered eigenvalues $0\leq \widehat \gamma_{2,1}\leq \ldots \leq \widehat \gamma_{2,r}$ entering $\widehat \xi_r$ can be replaced in the asymptotic limit by the ordered eigenvalues $0\leq \widehat \lambda_{2,1}\leq \ldots \leq \widehat \lambda_{2,r}$ of the $r\times r$ matrix
\begin{equation}\label{e:proofs-global-CM}
	D'(\widehat M-M)D^2 (\widehat M-M)D = D' \widehat M D^2 \widehat M D,
\end{equation}
where $D$ is a $p\times r$ matrix described below. The key here is that the matrix (\ref{e:proofs-global-CM}) is $r\times r$ so that the sum of its $r$ eigenvalues is just its trace, which is amenable to easier manipulations. The matrix $D$ will also play an important role of standardization.

The $p\times r$ matrix $D$ and another $p\times (p-r)$ matrix $\widetilde D$ enter into a $p\times p$ matrix $D_0 = (\widetilde D ,  D)$ characterized as follows: $D_0$ consists of the ``eigenvectors'' associated with the eigenvalues $\gamma_{2,i}$ through the characteristic equation 
\begin{equation}\label{e:proofs-global-FM}
	|(FMF)^2 - \gamma_{2,i}I_p| = |FMF^2MF - \gamma_{2,i}I_p| = |MF^2M-\gamma_{2,i}F^{-2}|=0
\end{equation}
that satisfy
\begin{equation}\label{e:proofs-global-CF}
	D_0'F^{-2}D_0= F^{-1}D_0^2 F^{-1} =I_p.
\end{equation}
Said differently, $F^{-1}D_0$ consists of the eigenvectors of $(FMF)^2$, that is,
\begin{equation}\label{e:proofs-global-MCF}
	F MF^2M D_0 = F^{-1}D_0 \mbox{diag}(\gamma_{2,p},\ldots,\gamma_{2,1}).
\end{equation}
Since we deal with squared symmetric matrices, $F^{-1}D_0$ also consists of the eigenvectors of $FMF$, whose squared eigenvalues are $\gamma_{2,i}$, so that 
\begin{equation}\label{e:proofs-global-MCF-2}
	F M D_0 = F^{-1}D_0 \mbox{diag}(\gamma_{p'},\ldots,\gamma_{1'}),
\end{equation} 
where $(\gamma_{i'})^2=\gamma_{2,i}$ and the prime indicates that the order is not necessarily that of the increasing order in $\gamma_{1}\leq \ldots\leq \gamma_p$. This fact will be used below. Note also that $\gamma_{2,r}=\ldots=\gamma_{2,1}=0$ by assumption and hence that $D'M=0$.

The next result will justify the replacement of the eigenvalues $\widehat \gamma_{2,i}$.

\begin{lemma}\label{l:proof-global-gamma-replace}
With the above notation and under Assumptions (A.1), (A.2g), we have for $i=1,\ldots,r$,
\begin{equation}\label{e:proof-global-gamma-replace}
	\sup|a_T^2 \widehat \gamma_{2,i} - a_T^2 \widehat \lambda_i| = O_p\Big( (Th/\ln T)^{1/2} \Big).
\end{equation}	
\end{lemma}

\medskip

\noindent {\sc Proof:} Let $|B|$ denote the determinant of a matrix $B$. As on p.\ 173 of \cite{robin:2000tests}, we have
$$
0 = \Big| \widehat M\widehat F^2\widehat M - \widehat \gamma_{2,i} \widehat F^{-2} \Big|
= \Big| (\widetilde D ,  a_TD)'(\widehat M\widehat F^2\widehat M - \widehat \gamma_{2,i} \widehat F^{-2})  (\widetilde D ,  a_TD) \Big|
$$

\begin{equation*}
= \left|
\left(
\begin{array}{cc}
	\widetilde D^{'}( \widehat M \widehat F^2 \widehat M - \widehat \gamma_{2,i} \widehat F^{-2})\widetilde D & a_T \widetilde D^{'} ( \widehat M\widehat F^2\widehat M - \widehat \gamma_{2,i} \widehat F^{-2}) D \\
	a_T D^{'}(\widehat M \widehat F^2 \widehat M - \widehat \gamma_{2,i} \widehat F^{-2})\widetilde D & a_T^2 D'(\widehat M\widehat F^2\widehat M - \widehat \gamma_{2,i} \widehat F^{-2}) D
\end{array}
\right)
\right|.
\end{equation*}

By using the relation $|(B_{11}\ B_{12};\ B_{21}\ B_{22})|=|B_{11}|\cdot |B_{22}-B_{21}B_{11}^{-1}B_{12}|$, we have further that
\begin{equation}\label{e:proofs-global-2-det0}
	0 = |\widehat S|\cdot |\widehat W - a_T^2 \widehat \gamma_{2,i} \widehat V^{-1}|,
\end{equation}
where 
\begin{eqnarray*}
	\widehat S & = & \widetilde D'(\widehat M\widehat F^2\widehat M - \widehat \gamma_{2,i} \widehat F^{-2})\widetilde D,\\
	\widehat W & = & a_T^2 D'\widehat M\widehat F^2\widehat M D - a_T^2 D'\widehat M\widehat F^2\widehat M \widetilde D \widehat S^{-1} \widetilde D'\widehat M\widehat F^2\widehat M D,\\
	\widehat V^{-1} & = & D'\widehat F^{-2}D + \widehat \gamma_{2,i} D'\widehat F^{-2}\widetilde D \widehat S^{-1} \widetilde D'\widehat F^{-2}D\\
	& & - D'\widehat M\widehat F^2\widehat M \widetilde D \widehat S^{-1} \widetilde D'\widehat F^{-2}D - D'\widehat F^{-2}\widetilde D\widehat S^{-1} \widetilde D'\widehat M\widehat F^2\widehat M D. 
\end{eqnarray*}

By Proposition \ref{p:proofs-global-2-AM} and the smoothness of $\widetilde D$ by Proposition \ref{p:proofs-global-2-analytic}, note that 
$$
\widehat S = \widetilde D' MF^2M\widetilde D + O_{p,{\scriptsize \rm sup}}((Th/\ln T)^{-1/2}) = \mbox{\rm diag}(\gamma_{2,p},\ldots,\gamma_{2,r+1}) + O_{p,{\scriptsize \rm sup}}((Th/\ln T)^{-1/2}), 
$$
where the second equality follows from (\ref{e:proofs-global-MCF}) and (\ref{e:proofs-global-CF}). This shows that, asymptotically, $|\widehat S|>0$. Hence, in view of (\ref{e:proofs-global-2-det0}), we may suppose without loss of generality that $a_T^2\widehat \gamma_{2,i}$ are the eigenvalues of the matrix $\widehat W\widehat V$. The matrix $\widehat V$ is symmetric and its eigenvalues are positive asymptotically since $\widehat V\to_p D'F^{-2}D$. Then, $\widehat V$ may be assumed to be positive definite, and $a_T^2\widehat \gamma_{2,i}$ be taken as eigenvalues of $\widehat V^{1/2}\widehat W\widehat V^{1/2}$. Since the matrix is symmetric, by applying the Wielandt-Hoffman theorem (e.g.\ \cite{golub:2012matrix}), we get that 
$$
\sup |a_T^2 \widehat \gamma_{2,i} - a_T^2 \widehat \lambda_i| \leq \sup \Big| 
\widehat V^{1/2}\widehat W\widehat V^{1/2} - a_T^2 D'(\widehat M-M) D^2 (\widehat M-M) D
\Big|.
$$
By using Proposition \ref{p:proofs-global-2-AM} and the fact that the square root is a continuous operation on positive definite matrces, $\widehat V^{1/2} = (D'F^{-2}D)^{1/2} + O_{p,{\scriptsize \rm sup}}((Th/\ln T)^{-1/2}) = I_r + O_{p,{\scriptsize \rm sup}}((Th/\ln T)^{-1/2})$. Similarly, for $\widehat W$, we have $\widehat W = a_T^2 D'\widehat M D^2\widehat M D + O_{p,{\scriptsize \rm sup}}((Th/\ln T)^{-1/2}) = a_T^2 D'(\widehat M-M) D^2 (\widehat M-M) D+ O_{p,{\scriptsize \rm sup}}((Th/\ln T)^{-1/2})$, where the first equality follows from the relation 
$$
F^2 - F^2 M\widetilde D \mbox{\rm diag}(\gamma_{2,p},\ldots,\gamma_{2,r+1})^{-1} \widetilde D' M F^2  = D^2.
$$
The latter relation holds by the following argument. Note that it is equivalent to 
$$
I = F M\widetilde D \mbox{\rm diag}(\gamma_{2,p},\ldots,\gamma_{2,r+1})^{-1} \widetilde D' M F + F^{-1}D^2 F^{-1}
$$
and in view of (\ref{e:proofs-global-MCF}), follows from $F M\widetilde D \mbox{\rm diag}(\gamma_{2,p},\ldots,\gamma_{2,r+1})^{-1} \widetilde D' M F = F^{-1}\widetilde D^2 F^{-1}$ or 
$$
F^2 M\widetilde D \mbox{\rm diag}(\gamma_{2,p},\ldots,\gamma_{2,r+1})^{-1} \widetilde D' M F^2 = \widetilde D^2,
$$
which is a consequence of (\ref{e:proofs-global-MCF-2}).  \quad \quad $\Box$

\medskip

By Lemma \ref{l:proof-global-gamma-replace}, instead of working with $\int \widehat \xi_r du$, we can focus instead on 
\begin{equation}\label{e:proofs-global-L}
	\widehat L_r = a_T^2 \sum_{i=1}^r \widehat \lambda_i = a_T^2 \mbox{tr}\{ D' \widehat M D^2 \widehat M D \} = a_T^2 \mbox{tr}\{ D' (\widehat M -M) D^2 (\widehat M -M) D \} 
\end{equation}
and $\int \widehat L_r du$. Write $\widehat M- M = S_1+S_2-S_3-S_4$ as in (\ref{e:proofs-local-MS}) of the proof of Proposition \ref{p:M-loctest-vc-an}. Then,
\begin{equation}\label{e:proofs-global-L-dec}
	\int \widehat L_r du = \sum_{j,k=1}^4 (\pm 1) \int \widehat L_{r,jk} du =: \sum_{j,k=1}^4 \widehat G_{r,jk},
\end{equation}
where 
\begin{equation}\label{e:proofs-global-L-dec-0}
	\widehat L_{r,jk} = a_T^2 \mbox{tr}\{ D' S_j D^2 S_k D \} 
\end{equation}
and $(\pm 1)$ in (\ref{e:proofs-global-L-dec}) accounts for the signs of $S_j, S_k$ in the decomposition (\ref{e:proofs-local-MS}). The next two lemmas concern the asymptotics of $\widehat G_{r,jk}$.

\begin{lemma}\label{l:proofs-global-G11}
Under Assumptions (Y1), (Y2), (K), (A1), (A2g), and 
$$
T\to\infty,\ h\to 0,\ Th^{3/2}\to\infty,\quad T^{\epsilon/(4+\epsilon)} h^{1/2}\to\infty,
$$
we have
\begin{equation}\label{e:proofs-global-G11}
\frac{\widehat G_{r,11}- |{\cal H}| \frac{r(  \mu_4+r-1)}{ \mu_4}}{\sqrt{h \frac{\|\overline K\|_2^2}{\|K\|_2^4} \frac{2(r \mu_4^2+2r(r-1))}{ \mu_4^2} |{\cal H}|}} \stackrel{d}{\to} {\cal N}(0,1).
\end{equation}
\end{lemma}

As noted following Assumption (A2g), this assumption is needed to apply Proposition \ref{p:proofs-global-2-analytic} to have the smoothness of $D=D(u)$.

\medskip

\noindent {\sc Proof:} In view of the definition of $S_1$ following (\ref{e:proofs-local-MS}), we can write
$$
\widehat G_{r,11} = \int_{\cal H} a_T^2 \mbox{tr}\{ D' S_1 D^2 S_1 D \} du = I_1 + I_2
$$
$$
:= \frac{a_T^2}{T^2} \sum_{t=1}^T \int_{\cal H} \mbox{tr}\Big\{ D' (X_tX_t' - A^2(\frac{t}{T})) D^2 (X_tX_t' - A^2(\frac{t}{T})) D \Big\} K_h^2(u - \frac{t}{T}) du 
$$
$$
+ \frac{2a_T^2}{T^2} \sum_{t_1<t_2} \int_{\cal H} \mbox{tr}\Big\{ D' (X_{t_1}X_{t_1}' - A^2(\frac{t_1}{T})) D^2 (X_{t_2}X_{t_2}' - A^2(\frac{t_2}{T})) D \Big\} K_h(u - \frac{t_1}{T})K_h(u - \frac{t_2}{T}) du. 
$$
We shall argue first that $D=D(u)$ in $I_1$ and $I_2$ above can be replaced by $D(t/T)$, and shall denote the respective terms by $\widetilde I_1$ and $\widetilde I_2$. The relation (\ref{e:proofs-global-CF}) will then be used to simplify $\widetilde I_1$ and $\widetilde I_2$. Finally, we will show that $\widetilde I_1$ produces the centering in (\ref{e:proofs-global-G11}) and $\widetilde I_2$ yields the asymptotic normality in (\ref{e:proofs-global-G11}).

To see why $I_1$ can be replaced by $\widetilde I_1$, that is, $D=D(u)$ be replaced by $D(t/T)$, consider one of the terms in the difference between $I_1$ and $\widetilde I_1$, namely,
$$
R_1 = \frac{a_T^2}{T^2} \sum_{t=1}^T \int_{\cal H} \mbox{tr}\Big\{ (D-D(\frac{t}{T}))' (X_tX_t' - A^2(\frac{t}{T})) D^2 (X_tX_t' - A^2(\frac{t}{T})) D \Big\} K_h^2(u - \frac{t}{T}) du.
$$
(The other terms in the difference can be dealt with similarly.) Then, by using the smoothness of $D(u)$ by Proposition \ref{p:proofs-global-2-analytic}, the expression (\ref{e:M-loctest-ass-def}) for $a_T$ and Lemma \ref{l:sum-int-kernel}, $(ii)$, we have 
$$
|R_1| \leq \frac{C}{T} \sum_{t=1}^T \|X_tX_t' - A^2(\frac{t}{T})\|_2^2 \int_{\cal H} \|D(u) - D(\frac{t}{T})\| K_{2,h}(u - \frac{t}{T}) du
$$ 
$$
\leq \frac{C'h}{T} \sum_{t=1}^T \|X_tX_t' - A^2(\frac{t}{T})\|_2^2 \int_{\cal H}  K_{2,h}^{(1)}(u - \frac{t}{T}) du
$$
$$
\leq \frac{C''h}{T} \sum_{t=1}^T \|X_tX_t' - A^2(\frac{t}{T})\|_2^2  \leq C(w) h,
$$
for a random constant $C(w)$. The term $R_1$ would then not affect the asymptotics (\ref{e:proofs-global-G11}) since $h/\sqrt{h}\to 0$.

For $I_2$ and $\widetilde I_2$, consider similarly a term from their difference given by 
$$
R_2 = \sum_{t_1<t_2} b_{T,t_1,t_2}(X_{t_1}X_{t_1}',X_{t_2}X_{t_2}')
$$
$$
:=  \sum_{t_1<t_2} \frac{2a_T^2}{T^2} \int_{\cal H} \mbox{tr}\Big\{ (D-D(\frac{t}{T}))' (X_{t_1}X_{t_1}' - A^2(\frac{t_1}{T})) D^2 (X_{t_2}X_{t_2}' - A^2(\frac{t_2}{T})) D \Big\} K_h(u - \frac{t_1}{T})K_h(u - \frac{t_2}{T}) du.
$$
We shall use Proposition \ref{p:proofs-global-2-Ustat-result} to obtain a convergence rate for $R_2$. To apply the proposition, a number of (conditional) expectations involving $b_{T,t_1,t_2}(X_{t_1}X_{t_1}',X_{t_2}X_{t_2}')$ need to be evaluated. Note that $\EE b_{T,t_1,t_2}(X_{t_1}X_{t_1}',X_{t_2}X_{t_2}') = 0$ and $\EE (b_{T,t_1,t_2}(X_{t_1}X_{t_1}',X_{t_2}X_{t_2}')|X_{t_1}X_{t_1}')=0$ and similarly when conditioning on $X_{t_2}X_{t_2}'$. We thus only need to consider $\EE (b_{T,t_1,t_2}(X_{t_1}X_{t_1}',X_{t_2}X_{t_2}'))^2$. By using the generalized Minkowski's inequality and the smoothness of $D(u)$ by Proposition \ref{p:proofs-global-2-analytic}, note that
$$
\Big( \EE (b_{T,t_1,t_2}(X_{t_1}X_{t_1}',X_{t_2}X_{t_2}'))^2 \Big)^{1/2}
$$
$$
\leq \int_{\cal H} \Big( \EE  \Big( \frac{2a_T^2}{T^2} \mbox{tr}\Big\{ (D-D(\frac{t}{T}))' (X_{t_1}X_{t_1}' - A^2(\frac{t_1}{T})) D^2 (X_{t_2}X_{t_2}' - A^2(\frac{t_2}{T})) D \Big\} K_h(u - \frac{t_1}{T})K_h(u - \frac{t_2}{T}) \Big)^2  \Big)^{1/2} du
$$
$$
\leq \frac{Ca_T^2}{T^2}   \int_{\cal H} \|D (u) - D(\frac{t_1}{T})\|_2 K_h(u - \frac{t_1}{T})K_h(u - \frac{t_2}{T}) du
$$
$$
\leq \frac{C'h^2}{T}   \int_{\cal H} K_{h}^{(1)}(u - \frac{t_1}{T})K_h(u - \frac{t_2}{T}) du,
$$
where we used the definiton of $a_T$. It follows from Lemma \ref{l:proofs-global-2-kernel-result} that 
$$
\sum_{t_1<t_2} \EE (b_{T,t_1,t_2}(X_{t_1}X_{t_1}',X_{t_2}X_{t_2}'))^2 \leq \frac{Ch^4}{T^2} \sum_{t_1<t_2} \Big(\int_{\cal H} K_{h}^{(1)}(u - \frac{t_1}{T})K_h(u - \frac{t_2}{T}) du \Big)^2
\leq  C'h^3.
$$
Hence, by Proposition \ref{p:proofs-global-2-Ustat-result}, $R_2$ is of the order $O_p(h^{3/2})$ and hence does not affect the asymptotics (\ref{e:proofs-global-G11}) since $h^{3/2}/\sqrt{h}\to 0$.

We can thus replace $I_1$ and $I_2$ by $\widetilde I_1$ and $\widetilde I_2$, respectively, which in view of (\ref{e:proofs-global-CF})  can be expressed as 
$$
\widetilde I_1 = \frac{a_T^2}{T^2} \sum_{t=1}^T \mbox{tr}\Big\{ (\widetilde Y_t\widetilde Y_t' - I_r) (\widetilde Y_t\widetilde Y_t' - I_r) \Big\} \int_{\cal H}  K_h^2(u - \frac{t}{T}) du, 
$$
$$
\widetilde I_2 = \frac{2a_T^2}{T^2} \sum_{t_1<t_2}  \mbox{tr}\Big\{ (\widetilde Y_{t_1}\widetilde Y_{t_1}' - I_r) (\widetilde Y_{t_2}\widetilde Y_{t_2}' - I_r) \Big\}  \int_{\cal H} K_h(u - \frac{t_1}{T})K_h(u - \frac{t_2}{T}) du,
$$
where $\widetilde Y_t=D^{'}(t/T) X_t$. 

For $\widetilde I_1$, write it as 
$$
\widetilde I_1 = \frac{1}{\mu_4T} \sum_{t=1}^T \sum_{i,j=1}^r (\widetilde Y_t\widetilde Y_t' - I_r)_{ij}^2  \int_{\cal H}  K_{2,h}(u - \frac{t}{T}) du 
$$
$$
= \frac{1}{\mu_4T} \sum_{t=1}^T \sum_{i,j=1}^r \Big( (\widetilde Y_t\widetilde Y_t' - I_r)_{ij}^2 - \EE (\widetilde Y_t\widetilde Y_t' - I_r)_{ij}^2  \Big)  \int_{\cal H}  K_{2,h}(u - \frac{t}{T}) du 
$$
$$
+ \sum_{i,j=1}^r \EE (\widetilde Y_1\widetilde Y_1' - I_r)_{ij}^2 \frac{1}{\mu_4T} \sum_{t=1}^T  \int_{\cal H}  K_{2,h}(u - \frac{t}{T}) du  = \widetilde I_{1,1}  + \widetilde I_{1,2}.
$$
It can be checked that
$$
\sum_{i,j=1}^r \EE (\widetilde Y_1\widetilde Y_1' - I_r)_{ij}^2 = r\mu_4 + r(r-1).
$$
Hence, by Lemma \ref{l:sum-int-kernel}, $(ii)$,
$$
\widetilde I_{1,2} = |{\cal H}| \frac{r\mu_4 + r(r-1)}{\mu_4}  + O(\frac{1}{Th}).
$$
For $\widetilde I_{1,1}$, setting $\xi_t = \sum_{i,j=1}^r ( (\widetilde Y_t\widetilde Y_t' - I_r)_{ij}^2 - \EE (\widetilde Y_t\widetilde Y_t' - I_r)_{ij}^2 )$, $p=1+\epsilon/4$ and using the von Bahr-Essen inequality (\cite{von:1965inequalities}), we obtain that
$$
\EE|\widetilde I_{1,1}|^p = C \EE \Big| \frac{1}{T} \sum_{t=1}^T \xi_t  \int_{\cal H}  K_{2,h}(u - \frac{t}{T}) du  \Big|^p \leq \frac{C' \EE|\xi_1|^p}{T^p} \sum_{t=1}^T \Big| \int_{\cal H}  K_{2,h}(u - \frac{t}{T}) du  \Big|^p \leq \frac{C''}{T^{p-1}}
$$
and hence $\widetilde I_{1,1} = O_p(T^{1/p-1})$. By assumption, $Th^{3/2}\to \infty$ and $T^{1-1/p}h^{1/2} = T^{\epsilon/(4+\epsilon)}h^{1/2}\to\infty$. This shows that the error term in $\widetilde I_{1,1}$ and $\widetilde I_{1,2}$ do not affect the asymptotics (\ref{e:proofs-global-G11}), and that $\widetilde I_{1,2}$ produces the desired asymptotic mean.

For $\widetilde I_2$, consider 
$$
\frac{\widetilde I_2}{\sqrt{h}} = \sum_{t_1<t_2} b_{T,t_1,t_2},
$$
where
$$
b_{T,t_1,t_2} = \frac{2h^{1/2}}{\|K\|_2^2 \mu_4 T}\mbox{tr}\Big\{ (\widetilde Y_{t_1}\widetilde Y_{t_1}' - I_r) (\widetilde Y_{t_2}\widetilde Y_{t_2}' - I_r)\Big\} \int_{\cal H} K_h(u - \frac{t_1}{T})K_h(u - \frac{t_2}{T}) du.
$$
We will argue that $\widetilde I_2/\sqrt{h}$ is asymptotically normal with the desired limiting variance. By Proposition 3.2 in \cite{dejong:1987central}, it is enough to show that 
\begin{enumerate}
	\item $\mbox{Var}(\frac{\widetilde I_2}{\sqrt{h}}) \to \frac{\|\overline{K}\|_2^2}{\|K\|_2^4} \frac{2(r\mu_4^2+2r(r-1))}{\mu_4^2} |{\cal H}|$;
	\item $G_{T,i}=o(1)$, $i=1,2,4$, where
	\begin{eqnarray*}
		G_{T,1} & = & \sum_{t_1<t_2} \EE b_{T,t_1,t_2}^4,\\
		G_{T,2} & = & \sum_{t_1<t_2<t_3} \EE b_{T,t_1,t_2}^2 b_{T,t_1,t_3}^2 + \EE b_{T,t_1,t_2}^2 b_{T,t_2,t_3}^2 + \EE b_{T,t_1,t_3}^2 b_{T,t_2,t_3}^2 , \\
		G_{T,4} & = & \sum_{t_1<t_2<t_3<t_4} \EE b_{T,t_1,t_2} b_{T,t_1,t_3} b_{T,t_2,t_4} b_{T,t_3,t_4} + \EE b_{T,t_1,t_2} b_{T,t_1,t_4} b_{T,t_2,t_3} b_{T,t_3,t_4} \\
		& & \quad \quad \quad \quad \quad \quad  + \EE b_{T,t_1,t_3} b_{T,t_1,t_4} b_{T,t_2,t_4} b_{T,t_2,t_4}.
	\end{eqnarray*}
\end{enumerate}
The first point above follows from 
$$
\mbox{Var}(\frac{\widetilde I_2}{\sqrt{h}}) = \frac{4h}{\|K\|_2^4 \mu_4^2 T^2} \sum_{t_1<t_2} \EE \mbox{tr}^2\Big\{ (\widetilde Y_{t_1}\widetilde Y_{t_1}' - I_r) (\widetilde Y_{t_2}\widetilde Y_{t_2}' - I_r)\Big\}^2 \Big( \int_{\cal H} K_h(u - \frac{t_1}{T})K_h(u - \frac{t_2}{T}) du \Big), 
$$
the observation that
$$
\EE \mbox{tr}^2\Big\{ (\widetilde Y_{t_1}\widetilde Y_{t_1}' - I_r) (\widetilde Y_{t_2}\widetilde Y_{t_2}' - I_r)\Big\} = r\mu_4^2 + 2r(r-1)
$$
and Lemma \ref{l:proofs-global-2-kernel-result}. For the second point above, $G_{T,1}$ can be bounded up to a constant by
$$
\frac{h^2}{T^4} \sum_{t_1<t_2} \Big( \int_{\cal H} K_h(u - \frac{t_1}{T})K_h(u - \frac{t_2}{T}) du \Big)^4.
$$
For example, the first term in the sum of $G_{T,2}$ can be bounded up to a constant by
$$
\frac{h^2}{T^4} \sum_{t_1<t_2<t_3} \Big( \int_{\cal H} K_h(u - \frac{t_1}{T})K_h(u - \frac{t_2}{T}) du \Big)^2\Big( \int_{\cal H} K_h(u - \frac{t_1}{T})K_h(u - \frac{t_3}{T}) du \Big)^2.
$$
For example, the first term in the sum of $G_{T,4}$ can be bounded up to a constant by
$$
\frac{h^2}{T^4} \sum_{t_1<t_2<t_3<t_4} \Big( \int_{\cal H} K_h(u - \frac{t_1}{T})K_h(u - \frac{t_2}{T}) du \Big)\Big( \int_{\cal H} K_h(u - \frac{t_1}{T})K_h(u - \frac{t_3}{T}) du \Big)\times
$$
$$
\times\ \Big( \int_{\cal H} K_h(u - \frac{t_2}{T})K_h(u - \frac{t_4}{T}) du \Big)\Big( \int_{\cal H} K_h(u - \frac{t_3}{T})K_h(u - \frac{t_4}{T}) du \Big) .
$$
The rate $o(1)$ for each of these bounds follows from Lemma \ref{l:proofs-global-2-kernel-result}. This completes the proof of the lemma. \quad \quad $\Box$

\medskip

\begin{lemma}\label{l:proofs-global-Gjk}
Under Assumptions (Y1), (Y2), (K), (A1), (A2g), and 
$$
T\to\infty,\ h\to 0,\ Th^{5/2}\to\infty,\quad T h^{3}\to 0,
$$
we have, for $(j,k)\neq (1,1)$,
\begin{equation}\label{e:proofs-global-Gjk}
\widehat G_{r,jk} = \int \widehat L_{r,jk} du = o_p(h^{1/2}).
\end{equation}
\end{lemma}

In fact, the proof of lemma establishes sharper rates of convergence of $\widehat G_{r,jk}$ to 0. The rate given in the lemma is what is needed to conclude the convergence (\ref{e:global-an-1}). Indeed, the latter now follows immediately from the arguments above and, in particular, Lemmas \ref{l:proofs-global-G11}  and \ref{l:proofs-global-Gjk} .

\medskip

\noindent {\sc Proof:} 
We consider only the cases $(j,k)=(2,2)$, $(3,3)$, $(4,4)$, $(1,2)$, $(1,3)$ and $(1,4)$. The other mixed cases can be dealt with similarly. 

For $(j,k)=(2,2)$, we have 
$\int \widehat L_{r,22} du = a_T^2 \int \mbox{tr}\{D'S_2D^2S_2D\}du$ where $S_2$ is defined in (\ref{e:proofs-local-MS-list}). As in the proof of Proposition \ref{p:M-loctest-vc-an}, $S_2 = O(h+(Th)^{-1})$ uniformly in $u$, where the latter follows by using Lemma \ref{l:sum-int-kernel}, $(ii)$. Then, by using the smoothness of $D=D(u)$ by Proposition \ref{p:proofs-global-2-analytic}, we have $\int \widehat L_{r,22} du = O(Th (h + (Th)^{-1})^2) = O(Th^3 + (Th)^{-1}+h)$. This is of the order $o(h^{1/2})$ since, in particular, $Th^{5/2}\to 0$ by assumption.  For $(j,k)=(3,3)$, we have similarly $\int \widehat L_{r,33} du = a_T^2 \int \mbox{tr}\{D'S_3D^2S_3D\}du$ where $S_3$ is defined in (\ref{e:proofs-local-MS-list}). The term $S_3$ does not depend on $u$ and its rate is $O_p(T^{-1/2})$ as shown in the proof of Proposition \ref{p:M-loctest-vc-an}. This leads to $\int \widehat L_{r,33} du = O_p(Th (T^{-1/2})^2) = O_p(h)$, which is again of the order $o_p(h^{1/2})$ as desired. For $(j,k)=(4,4)$, $S_4$ does not depend on $u$ either and its rate is $O(T^{-1})$ as shown in the proof of Proposition \ref{p:M-loctest-vc-an}. This leads to $\int \widehat L_{r,44} du = O(Th (T^{-1})^2) = O(h/T)=o(h^{1/2})$.

For $(j,k)=(1,2)$, we need to consider $\int \widehat L_{r,12} du = a_T^2 \int \mbox{tr}\{D'S_1D^2S_2D\}du$. After matrix multiplication and taking the trace, a general term in $\int \widehat L_{r,12} du$ has the form 
$$
R_{1,2} = a_T^2 \int D_\ast'(S_1)_\ast (D^2)_\ast (S_2)_\ast D_\ast du,
$$
where $\ast$ refers to an index pair that can change from matrix to matrix. Furthermore, in view of (\ref{e:proofs-local-MS-list}),
$$
R_{1,2} = \frac{a_T^2}{T} \sum_{t=1}^T ((X_tX_t')_\ast - A^2(\frac{t}{T})_\ast) \int D_\ast D_\ast^2 D_\ast' (S_2)_\ast K_h(u-\frac{t}{T}) du.
$$
By using the facts that $S_2 = O(h + (Th)^{-1})$ uniformly in $u$ as above and $\int_\RR K_h(u-t/T) du=1$, it follows that $\EE R_{1,2}^2 = O(a_T^4 T^{-1} (h + (Th)^{-1})^2) = O(Th^4 + h^2 + 1/T)$ and hence that $R_{1,2} = o_p(h^{1/2})$ since, in particular, $Th^3 \to 0$. For $(j,k)=(1,3)$, a general term of interest is similarly,
$$
R_{1,3} = a_T^2 \int D_\ast'(S_1)_\ast (D^2)_\ast (S_3)_\ast D_\ast' du =  \frac{a_T^2  (S_3)_\ast }{T} \sum_{t=1}^T ((X_tX_t')_\ast - A^2(\frac{t}{T})_\ast) \int D_\ast D_\ast^2 D_\ast' K_h(u-\frac{t}{T}) du
$$
and hence $\EE R_{1,3}^2 = a_T^4 T^{-2} = h^4$. This leads to $R_{1,3} = o_p(h^{1/2})$. The case $(j,k)=(1,4)$ can be dealt with similarly by using the fact that $S_4 = O(T^{-1})$. \quad \quad $\Box$

\medskip

Finally, we prove the last statement of Proposition \ref{p:global-an} concerning the behavior of the test statistic under the alternative. For this, note that by Proposition \ref{p:proofs-global-2-AM}, 
$\sup_{u\in {\cal H}} | a_T^{-2} \xi_r(u) - \sum_{i=1}^r \gamma_{2,i}(u)  | \to_p  0$.
Furthermore, under the considered alternative $H_1$, and by using smoothness of $\gamma_{2,i}(u)$, we have $\sup_{u \in {\cal H}} \sum_{i=1}^r \gamma_{2,i}(u) >0$. It follows that $\widehat{\xi}_r = (a_T^2 \int_{\cal H}a_T^{-2} \xi_r(u) du - C_1 )/(C_2 h^{1/2})$ with constants $C_1$, $C_2$ behaves asymptotically as $C a_T^2/h^{1/2} = C' Th^{1/2}  \rightarrow \infty$ (for example, since Lemma \ref{l:proofs-global-G11} assumes $Th^{3/2} \rightarrow \infty$). This concludes the proof of Proposition \ref{p:global-an}.

\subsection{Auxiliary technical results for proof of Proposition \ref{p:global-an}}
\label{s:proofs-global-2}

The following auxiliary results were used in the proof of Proposition \ref{p:global-an}. The first result is analogous to the results in e.g.\ \cite{newey:1994kernel}, Lemma B.1. A separate proof though is needed since averaging and smoothing in the estimators here are with respect to time $t$, rather than the values of a random variable as in the aforementioned results.

\begin{proposition}\label{p:proofs-global-2-AM}
Under Assumptions (Y1), (Y2), (K), (A1), (A2g), and 
$$
T\to\infty,\ h\to 0,\ T^{\epsilon/(4+\epsilon)}h/\ln T\to \infty,\ Th^3\to 0,
$$
we have,
\begin{equation}\label{l:proofs-global-2-AM-1}
	\sup|\widehat A^{2k} - A^{2k}| = O_p\Big( (Th/\ln T)^{-1/2} \Big),\quad k=-1,1,
\end{equation}	
and 
\begin{equation}\label{l:proofs-global-2-AM-2}
	\sup|\widehat M - M| = O_p\Big( (Th/\ln T)^{-1/2} \Big).
\end{equation}	
\end{proposition}

\noindent {\sc Proof:} We outline the proof in the case $p=1$. Write 
\begin{equation}\label{e:proofs-global-2-AA}
	\widehat A^2(u) - A^2(u) = A^2(u) \Big(R_1(u) + R_2(u)\Big) + R_3(u) + R_4(u),
\end{equation}
where 
\begin{eqnarray*}
	R_1(u) &=& \frac{1}{T} \sum_{t=1}^T (Y_t^2 - 1)K_h(u-\frac{t}{T}),\\
	R_2(u) & = & \frac{1}{T} \sum_{t=1}^T K_h(u-\frac{t}{T}) - 1,\\
	R_3(u) & = &  \frac{1}{T} \sum_{t=1}^T \Big(A^2(\frac{t}{T}) - A^2(u)\Big)(Y_t^2 - 1)K_h(u-\frac{t}{T}),\\
	R_4(u) & = & \frac{1}{T} \sum_{t=1}^T \Big(A^2(\frac{t}{T}) - A^2(u)\Big) K_h(u-\frac{t}{T}).
\end{eqnarray*}
Set $\delta_T = (Th/\ln T)^{-1/2}$ or $\delta^{-1} = (Th/\ln T)^{1/2}$. We will show that $\delta_T^{-1}\sup_{u\in {\cal H}} |R_i(u)| = O_p(1)$, $i=1,2,3,4$. This will prove the convergence (\ref{l:proofs-global-2-AM-1}) for $k=1$, and the convergence for $k=-1$ will also follow since $A^2(u)$ is smooth for $u\in {\cal H}$.

By Lemma \ref{l:sum-int-kernel}, $(ii)$, we have $\sup_{u\in {\cal H}}|R_2(u)| = O((Th)^{-1})$ and hence also $\delta_T^{-1}\sup_{u\in {\cal H}} |R_2(u)| \to 0$ since $Th\to\infty$. Similarly, by the smoothness of $A^2(u)$ and Lemma \ref{l:sum-int-kernel}, $(ii)$, applied to $G=K_1$,
$$
\sup_{u\in {\cal H}} |R_4(u)| \leq \sup_{u\in {\cal H}} \frac{Ch}{T} \sum_{t=1}^T \Big| \frac{u-t/T}{h} \Big| K_h(u-\frac{t}{T}) = O(h)
$$
and hence $\delta_T^{-1}\sup_{u\in {\cal H}} |R_4(u)| \to 0$  since $Th^3\to 0$. The terms $R_1(u)$ and $R_3(u)$ are a more delicate to deal with. In their analysis, we will follow the proof of Lemma B.1 in \cite{newey:1994kernel} which, unsurprisingly, involves Bernstein's inequality.

We will consider the term $R_1(u)$ only, since the term $R_3(u)$ can be dealt with similarly. For the term $R_1(u)$, note that
$$
|R_1(u) - R_1(v)| \leq \frac{1}{T} \sum_{t=1}^T |Y_t^2 - 1| \Big| K_h(u-\frac{t}{T}) -  K_h(v-\frac{t}{T})  \Big| \leq \frac{C|u-v|}{Th^2} \sum_{t=1}^T |Y_t^2 - 1|  \leq \frac{C'(\omega) |u-v|}{h^2},
$$
where we used the Lipschitz continuity of $K$ and the law of large numbers in the last step with a constant $C'(\omega)$ which does not depend on $T,h,u,v$. The difference $\delta_T^{-1}|R_1(u)-R_1(v)|$ can then be made small as long as, for example, $|u-v|\leq T^{-2}$. Now, cover ${\cal H}$ with balls of radius $T^{-2}$ and centered at $u_j$, where we can take $j=1,\ldots,cT^2$. Then, for large $T$, any arbitrarily small $\epsilon$ and any fixed $\eta>0$, 
$$
\PP \Big( \sup_{u\in {\cal H}}|R_1(u)| > 2 \delta_T \eta \Big) \leq \epsilon+ \PP \Big( \sup_{j=1,\ldots,cT^2}|R_1(u_j)| > \delta_T \eta \Big)
$$
\begin{equation}\label{e:proofs-global-2-for-Bernst}
	\leq \epsilon + \sum_{j=1}^{cT^2} \PP \Big( |R_1(u_j)| > \delta_T \eta \Big) =
	\epsilon + \sum_{j=1}^{cT^2} \PP \Big( \Big|\sum_{t=1}^T (Y_t^2 -1) K(\frac{u_j-t/T}{h})\Big| > (Th\ln T)^{1/2} \eta \Big).
\end{equation}

Assume first that $Y_t^2-1$, $t\in\ZZ$, are bounded by a constant, almost surely. Then, by Bernstein's inequality, the last sum in (\ref{e:proofs-global-2-for-Bernst}) is bounded by (up to a constant)
\begin{equation}\label{e:proofs-global-2-Bernst}
\sum_{j=1}^{cT^2} \exp\Big\{ - \frac{Th(\ln T)\eta^2}{\EE(Y_0^2-1)^2 \sum_{t=1}^T K^2(u_j-t/T)/h) + C (Th\ln T)^{1/2} \eta} \Big\}.
\end{equation}
By Lemma \ref{l:sum-int-kernel}, $(ii)$, $\frac{1}{Th}\sum_{t=1}^T K^2(u_j-t/T)/h)$ behaves as a positive constant uniformly over $u_j$. Since $(Th\ln T)^{1/2}/(Th)\to 0$ by assumption, the above bound becomes, for large enough $T$ and constants $C_1,C_2>0$,
$$
C_1 T^2 \exp\Big\{-C_2 \eta^2 \ln T \Big\} = C_1  \exp\Big\{- (C_2 \eta^2 - 2) \ln T \Big\}. 
$$ 
The latter converges to $0$ as long as $\eta$ is large enough (so that $C_2 \eta^2 - 2>0$). When the variables $\widetilde \epsilon_t: = Y_t^2-1$ are not bounded almost surely, a standard truncation argument is used. Let $\widetilde \epsilon_{t,T} = Y_t^2-1$ if $|Y_t^2-1|\leq C_0T^{1/p}$ and $\widetilde \epsilon_{t,T} = C_0T^{1/p}$ otherwise, where $p=2+\epsilon/2$ and $\epsilon$ appears in Assumption (Y2). Let also $\widetilde R_1(u) = \frac{1}{T} \sum_{t=1}^T \widetilde \epsilon_{t,T} K_h(u-t/T)$. Suppose for simplicity that $\EE \widetilde \epsilon_{t,T}=0$. Then, 
$$
\PP(R_1(u)\neq \widetilde R_1(u)\ \mbox{for some}\ u) \leq \PP(\widetilde \epsilon_t \neq \widetilde \epsilon_{t,T}\ \mbox{for some}\ t=1,\ldots,T)
$$
$$
\leq T \PP(\widetilde \epsilon_t \neq \widetilde \epsilon_{t,T}) = T \PP(|\widetilde \epsilon_t|> C_0T^{1/p}) \leq \frac{\EE|\widetilde \epsilon_t|^p}{C_0^p}, 
$$
which can be made arbitrarily small by taking large enough $C_0$. The above argument for $R_1(u)$ can now be applied to $\widetilde R_1(u)$ for this large fixed $C_0$, but keeping track of and dealing with the truncation $C_0T^{1/p}$ which depends on $T$. More precisely, the bound (\ref{e:proofs-global-2-Bernst}) becomes
\begin{equation}\label{e:proofs-global-2-Bernst-again}
\sum_{j=1}^{cT^2} \exp\Big\{ - \frac{Th(\ln T)\eta^2}{\EE\widetilde \epsilon_{0,T}^2 \sum_{t=1}^T K((u_j-t/T)/h)^2 + C T^{1/p}(Th\ln T)^{1/2} \eta} \Big\}
\end{equation}
and the same argument as above applies as long as $T^{1/p}(Th\ln T)^{1/2}/(Th)\to 0$. The latter follows from the assumption $T^{1-2/p}h/\ln T= T^{\epsilon/(4+\epsilon)}h/\ln T\to \infty$.

The convergence (\ref{l:proofs-global-2-AM-2}) follows from (\ref{l:proofs-global-2-AM-1}) and the fact that $\widehat{\overline{A}^2} = \frac{1}{T}\sum_{t=1}^TX_tX_t'$ converges to $\overline{A}^2 = \int_0^1 A^2(u)du$ at a rate faster than $\delta_T$. Indeed, by using the notation of the proof of Proposition \ref{p:M-loctest-vc-an}, this difference is $S_3+S_4$ which is of the order $1/T^{1/2}$ as shown in that proof. \quad \quad $\Box$

\medskip

The next result allows having smooth eigenvectors associated with the matrix $M$. This is the only place where the analyticity of the matrix $M$ needs to be assumed.

\begin{proposition}\label{p:proofs-global-2-analytic}
Under Assumption (A2g), the matrix $D_0$ in the proof of Proposition \ref{p:global-an} can be chosen analytic.
\end{proposition}

\noindent {\sc Proof:} By Assumption (A2g), the matrix $FMF$ is analytic. By using the analytic singular value decomposition (e.g.\ \cite{bunse:1991numerical}), there are $p\times p$ analytic matrices $U$ and $T$ such that $FMF=UTU'$, where $T=\mbox{diag}(t_1,\ldots,t_p)$ are the singular values and an orthogonal matrix $U$ consists of the eigenvectors of $FMF$. Now take $D_0 = FU$. Then, $D_0$ is analytic and satisfies (\ref{e:proofs-global-CF}). 
\quad \quad $\Box$

\medskip

We also used the following result on several occasions concerning the limiting behavior of a quadratic form 
\begin{equation} \label{e:proofs-global-2-Ustat}
Q_T = \sum_{1 \leq t_1<t_2 \leq T} b_{T,t_1,t_2} (W_{t_1},W_{t_2}),
\end{equation}
where $W_t$, $t=1, \dots,T$, are i.i.d.\ random vectors in $\RR^d$ and $b_{T,t_1,t_2}: \RR^d \times \RR^d \to \RR$.

% symmetric kernel (that is, $a_T(x,y)=a_T(y,x)$).

\begin{proposition}\label{p:proofs-global-2-Ustat-result}
Let $Q_T$ be a quadratic form defined by (\ref{e:proofs-global-2-Ustat}) and assume that $\EE b_{T,t_1,t_2}(W_{t_1},W_{t_2})^2<\infty$. Then,
\begin{eqnarray}
	Q_T & = & \sum_{1 \leq t_1<t_2 \leq T}  \EE b_{T,t_1,t_2}(W_{t_1},W_{t_2}) \nonumber \\
	& & + O_p \left(\sqrt{
\sum_{t_1=1}^T \Big( \sum_{t_2=t_1+1}^T  ( \EE(\EE(b_{T,t_1,t_2} (W_{t_1},W_{t_2})|W_{t_1})^2) )^{1/2} \Big)^2  }\right) \nonumber \\
	& & + O_p \left(\sqrt{
\sum_{t_2=2}^T \Big( \sum_{t_1=1}^{t_2-1}  ( \EE(\EE(b_{T,t_1,t_2} (W_{t_1},W_{t_2})|W_{t_2})^2) )^{1/2} \Big)^2 
}\right) \nonumber \\
	& & + O_p\left( \sqrt{\sum_{1 \leq t_1<t_2 \leq T}  \EE b_{T,t_1,t_2}(W_{t_1},W_{t_2})^2} \right). \label{e:proofs-global-2-Ustat-result}
\end{eqnarray}
\end{proposition}

When $b_{T,t_1,t_2}(\cdot) \equiv b_T(\cdot)$ and the quadratic form $Q_T$ becomes a $U$-statistic, the result restates Lemma C.1 in \cite{fortuna:2008local}, p.\ 181. As for that lemma, the proof below uses the arguments of the proof of Lemma 3.1 in \cite{powell:1989semiparametric}.

\medskip

\noindent {\sc Proof:} Let 
$$
\widehat Q_T = \sum_{1 \leq t_1<t_2 \leq T}  \EE b_{T,t_1,t_2}(W_{t_1},W_{t_2}) 
$$
$$
+  \sum_{1 \leq t_1<t_2 \leq T} \Big(  \EE (b_{T,t_1,t_2}(W_{t_1},W_{t_2})|W_{t_1}) + \EE (b_{T,t_1,t_2}(W_{t_1},W_{t_2})|W_{t_2}) - 2\EE b_{T,t_1,t_2}(W_{t_1},W_{t_2})  \Big). 
$$
Then,
$$
Q_T - \widehat Q_T =  \sum_{1 \leq t_1<t_2 \leq T}  c_{T,t_1,t_2}(W_{t_1},W_{t_2}),
$$
where 
\begin{eqnarray*}
	c_{T,t_1,t_2}(W_{t_1},W_{t_2}) & = & b_{T,t_1,t_2}(W_{t_1},W_{t_2}) - \EE (b_{T,t_1,t_2}(W_{t_1},W_{t_2})|W_{t_1}) \\
	& & - \EE (b_{T,t_1,t_2}(W_{t_1},W_{t_2})|W_{t_2}) + \EE b_{T,t_1,t_2}(W_{t_1},W_{t_2}).
\end{eqnarray*}
This yields, by using the independence of $W_t$'s,
$$
\EE(Q_T - \widehat Q_T)^2 =  \sum_{1 \leq t_1<t_2 \leq T}  \EE c_{T,t_1,t_2}(W_{t_1},W_{t_2})^2,
$$
since $\EE c_{T,t_1,t_2}(W_{t_1},W_{t_2})c_{T,t_1',t_2'}(W_{t_1'},W_{t_2'})=0$ whenever $(t_1,t_2)\neq (t_1',t_2')$ (the interested reader should check this, especially  when e.g.\ $t_1=t_1'$, $t_2\neq t_2'$). Since $\EE c_{T,t_1,t_2}(W_{t_1},W_{t_2})^2 = O(\EE b_{T,t_1,t_2}(W_{t_1},W_{t_2})^2)$, we obtain that
$$
Q_T - \widehat Q_T =  O_p\left( \sqrt{\sum_{1 \leq t_1<t_2 \leq T}  \EE b_{T,t_1,t_2}(W_{t_1},W_{t_2})^2} \right),
$$
which is the last term in the relation (\ref{e:proofs-global-2-Ustat-result}). The first term on the right-hand side of the relation (\ref{e:proofs-global-2-Ustat-result}) is the first sum in the definition of $\widehat Q_T$. Finally, to deal with the second sum in the definition of $\widehat Q_T$, note that, for example,
$$
\EE \Big( \sum_{1 \leq t_1<t_2 \leq T} \Big(  \EE (b_{T,t_1,t_2}(W_{t_1},W_{t_2})|W_{t_1}) - \EE b_{T,t_1,t_2}(W_{t_1},W_{t_2})  \Big)  \Big)^2
$$
$$
= \sum_{t_1=1}^T \EE \Big( \sum_{t_2=t_1+1}^T \Big(  \EE (b_{T,t_1,t_2}(W_{t_1},W_{t_2})|W_{t_1}) - \EE b_{T,t_1,t_2}(W_{t_1},W_{t_2})  \Big)  \Big)^2
$$
$$
\leq \sum_{t_1=1}^T \Big( \sum_{t_2=t_1+1}^T \Big(  \EE  (\EE (b_{T,t_1,t_2}(W_{t_1},W_{t_2})|W_{t_1}) - \EE b_{T,t_1,t_2}(W_{t_1},W_{t_2}))^2  \Big)^{1/2}  \Big)^2
$$
$$
\leq \sum_{t_1=1}^T \Big( \sum_{t_2=t_1+1}^T \Big(  \EE (\EE (b_{T,t_1,t_2}(W_{t_1},W_{t_2})|W_{t_1})^2)  \Big)^{1/2}  \Big)^2,
$$
where for the first inequality above we used the generalized Minkowski inequality. This yields the second term on the right-hand side of the relation (\ref{e:proofs-global-2-Ustat-result}). The third term results from a similar argument for the term in the definition of $\widehat Q_T$, where the conditioning is on $W_{t_2}$. This proves the result (\ref{e:proofs-global-2-Ustat-result}). \quad \quad $\Box$

\medskip

Finally, the following auxiliary result concerns the behavior of the sum of power integrals of kernel functions, and was used in the proof of Lemma \ref{l:proofs-global-G11}. We also note that the result extends easily to other powers than $2$ and $4$ appearing below, and a product of two different kernel functions. It is formulated only for what is needed in Lemma \ref{l:proofs-global-G11} for the shortness sake.

\begin{lemma}\label{l:proofs-global-2-kernel-result}
For two kernel functions $K,\widetilde K$ satisfying Assumption (K), let
$$
K_{T,h}(t_1,t_2) = \int_{\cal H} K_h(u - \frac{t_1}{T}) K_h(u-\frac{t_2}{T}) du,\quad \widetilde K_{T,h}(t_1,t_2) = \int_{\cal H} \widetilde K_h(u - \frac{t_1}{T}) K_h(u-\frac{t_2}{T}) du.
$$	
Then, as $T\to\infty$, $h\to 0$, $Th\to\infty$, we have
\begin{equation}\label{e:proofs-global-2-kernel-result-0}
	\frac{h}{T^2}\sum_{1\leq t_1<t_2\leq T} \Big( K_{T,h}(t_1,t_2) \Big)^2 \to \frac{|{\cal H}|\|\overline K\|_2^2}{2} 
\end{equation}
and 
\begin{equation}\label{e:proofs-global-2-kernel-result-1}
\sum_{1\leq t_1<t_2\leq T} \Big( \widetilde K_{T,h}(t_1,t_2)  \Big)^p = O\Big(\frac{T^2}{h^{p-1}}\Big),\quad p=2,4,
\end{equation}
\begin{equation}\label{e:proofs-global-2-kernel-result-2}
\sum_{1\leq t_1<t_2<t_3\leq T} \Big( K_{T,h}(t_1,t_2)  \Big)^2 \Big( K_{T,h}(t_1,t_3)  \Big)^2 = O\Big(\frac{T^2}{h^2} \Big),
\end{equation}
\begin{equation}\label{e:proofs-global-2-kernel-result-3}
\sum_{1\leq t_1<t_2<t_3<t_4\leq T} K_{T,h}(t_1,t_2) K_{T,h}(t_1,t_3) K_{T,h}(t_2,t_4)K_{T,h}(t_3,t_4) = O\Big(\frac{T^4}{h} \Big).
\end{equation}
\end{lemma}

\noindent {\sc Proof:} We first show (\ref{e:proofs-global-2-kernel-result-0}) which requires a more delicate treatment. Denote the left hand-side of (\ref{e:proofs-global-2-kernel-result-0}) by $I_0$ and the interval ${\cal H}$ by $[a,b]$. Suppose $\mbox{supp}\{K\} = [-S,S]$ as in Assumption (K). Then,
$$
I_0 = \frac{1}{T^2h^3}\sum_{1\leq t_1<t_2\leq T} \Big( \int_a^b K\Big(\frac{u - \frac{t_1}{T}}{h}\Big) K\Big(\frac{u - \frac{t_2}{T}}{h}\Big)  du \Big)^2
$$
$$
=  \frac{1}{T^2h^3}\sum_{1\leq t_1<t_2\leq T} \Big( \int_{(a,b)\cap (\frac{t_2}{T}-Sh,\frac{t_1}{T}+Sh)} K\Big(\frac{u - \frac{t_1}{T}}{h}\Big) K\Big(\frac{u - \frac{t_2}{T}}{h}\Big)  du \Big)^2 = I_{0,1} + I_{0,2} + I_{0,3} + I_{0,4},
$$
where 
$$
I_{0,1} = \frac{1}{T^2h^3}\sum_{a<\frac{t_2}{T}-Sh<\frac{t_1}{T}+Sh<b} \Big( \int_{(\frac{t_2}{T}-Sh ,\frac{t_1}{T}+Sh )} K\Big(\frac{u - \frac{t_1}{T}}{h}\Big) K\Big(\frac{u - \frac{t_2}{T}}{h}\Big)  du \Big)^2,
$$
$$
I_{0,2} = \frac{1}{T^2h^3}\sum_{\frac{t_2}{T}-Sh<a<\frac{t_1}{T}+Sh < b } \Big( \int_{(a,\frac{t_1}{T}+Sh)} K\Big(\frac{u - \frac{t_1}{T}}{h}\Big) K\Big(\frac{u - \frac{t_2}{T}}{h}\Big)  du \Big)^2,
$$
$$
I_{0,3} = \frac{1}{T^2h^3}\sum_{ a < \frac{t_2}{T}-Sh<b<\frac{t_1}{T}+Sh} \Big( \int_{(\frac{t_2}{T} - Sh,b)} K\Big(\frac{u - \frac{t_1}{T}}{h}\Big) K\Big(\frac{u - \frac{t_2}{T}}{h}\Big)  du \Big)^2,
$$
$$
I_{0,4} = \frac{1}{T^2h^3}\sum_{\frac{t_2}{T}-Sh<a<b<\frac{t_1}{T}+Sh} \Big( \int_{(a,b)} K\Big(\frac{u - \frac{t_1}{T}}{h}\Big) K\Big(\frac{u - \frac{t_2}{T}}{h}\Big)  du \Big)^2
$$
and otherwise all the sums for $t_1,t_2$ are over $1\leq t_1<t_2\leq T$. Note that $I_{0,4}=0$ for small enough $h$, since $Th\to\infty$. For $I_{0,1}$, we have
$$
I_{0,1} = \frac{1}{T^2h}\sum_{a<\frac{t_2}{T}-Sh<\frac{t_1}{T}+Sh<b} \overline K\Big(\frac{t_2-t_1}{Th}\Big)^2 
$$ 
$$
= \frac{1}{T^2h} \Big( \sum_{t_1=aT-STh+2}^{aT+STh} \sum_{t_2=aT+STh+1}^{t_1+(2STh-1)}  +  \sum_{t_1=aT+STh+1}^{bT-STh-1} \sum_{t_2=t_1+1}^{t_1+(2STh-1)}\Big ) \overline K\Big(\frac{t_2-t_1}{Th}\Big)^2 =: I_{0,1}^{(1)} + I_{0,1}^{(2)}.
$$
Note that
$$
I_{0,1}^{(1)} \leq \frac{C}{T^2h} \sum_{t_1=aT-STh+2}^{aT+STh}  (t_1-aT+STh-2) =  \frac{C}{T^2h} \sum_{t=0}^{2STh-2} t \leq  \frac{C'(Th)^2}{T^2h} \to 0 
$$ 
and, after changing the summation indices, 
$$
I_{0,1}^{(2)} = \frac{1}{T^2h} \sum_{s_1=1}^{(b-a)T-2STh-1} \sum_{s_2=s_1+1}^{s_1+aT+3STh-1}  \overline K\Big(\frac{s_2-s_1}{Th}\Big)^2 =  \frac{1}{T^2h} ((b-a)T -2STh-1) \sum_{t=1}^{aT+3STh-1} \overline K\Big(\frac{t}{Th}\Big)^2
$$
$$
= \frac{1}{T} ((b-a)T -2STh-1) \frac{1}{Th}\sum_{t=1}^{2STh} \overline K\Big(\frac{t}{Th}\Big)^2\to \frac{(b-a)\|\overline K\|_2^2}{2}.
$$
Thus, $I_{0,1}\to (b-a) \|\overline K\|_2^2/2$. For $I_{0,2}$, we have 
$$
I_{0,2} \leq \frac{C}{T^2h^3}\sum_{\frac{t_2}{T}-Sh<a<\frac{t_1}{T}+Sh} (\frac{t_1}{T}+Sh - a)^2 = \frac{C}{T^4h^3} \sum_{t_1=aT-STh+1}^{aT+STh-2} \sum_{t_2=t_1+1}^{aT+STh-1}  (t_1 +STh - aT)^2
$$
$$
= \frac{C}{T^4h^3} \sum_{t_1=aT-STh+1}^{aT+STh-2}   (t_1 +STh - aT)^2 (aT+STh - t_1-1) = \frac{C}{T^4h^3} \sum_{t=1}^{2STh-3}    t^2 (2STh - t-1) 
$$
$$
=  Ch  \sum_{t=1}^{2STh-3}    \Big(\frac{t}{Th}\Big)^2 \Big(2S - \frac{t}{Th}\Big) \frac{1}{Th} = O(h)\to 0. 
$$
One can show similarly that $I_{0,3} \to 0$.

The relations (\ref{e:proofs-global-2-kernel-result-1})--(\ref{e:proofs-global-2-kernel-result-3}) are easier to obtain because ${\cal H}$ can be replaced by $\RR$ and the involved integrals simplified to convolutions. We will show how this can be argued for (\ref{e:proofs-global-2-kernel-result-2}). Denote the sum on the left-hand side of (\ref{e:proofs-global-2-kernel-result-2}) by $I_2$. Then,
$$
I_2 \leq \sum_{1\leq t_1<t_2<t_3\leq T} \Big( \int_\RR K_h(u - \frac{t_1}{T}) K_h(u-\frac{t_2}{T}) du \Big)^2 \Big( \int_\RR K_h(u - \frac{t_1}{T}) K_h(u-\frac{t_3}{T}) du \Big)^2 
$$
$$
= \frac{1}{h^4} \sum_{1\leq t_1<t_2<t_3\leq T} \overline K \Big( \frac{t_2-t_1}{Th} \Big)^2 \overline K \Big( \frac{t_3-t_1}{Th} \Big)^2 \leq 
 \frac{1}{h^4} \sum_{t =1}^T \sum_{s_1=1}^T \sum_{s_2=1}^T\overline K \Big( \frac{s_1}{Th} \Big)^2 \overline K \Big( \frac{s_1-s_2}{Th} \Big)^2 
$$
$$
= \frac{T^4}{h^2}  \sum_{s_1=1}^{STh} \overline K \Big( \frac{s_1}{Th} \Big)^2 \frac{1}{Th} \sum_{s=-STh}^{STh} \overline K \Big( \frac{s}{Th} \Big)^2 \frac{1}{Th} = O\Big(\frac{T^4}{h^2} \Big).
$$
The relation  (\ref{e:proofs-global-2-kernel-result-1}) and (\ref{e:proofs-global-2-kernel-result-3}) can be argued for similarly.
\quad \quad $\Box$

\bibliographystyle{Chicago}

\bibliography{ssa}

\begin{thebibliography}{}

\bibitem[\protect\citeauthoryear{Blythe, von Bunau, Meinecke, and
  Muller}{Blythe et~al.}{2012}]{SSAchangepoint}
Blythe, D. A.~J., P.~von Bunau, F.~C. Meinecke, and K.~R. Muller (2012, April).
\newblock Feature extraction for change-point detection using stationary
  subspace analysis.
\newblock {\em IEEE Transactions on Neural Networks and Learning
  Systems\/}~{\em 23\/}(4), 631--643.

\bibitem[\protect\citeauthoryear{Bunse-Gerstner, Byers, Mehrmann, and
  Nichols}{Bunse-Gerstner et~al.}{1991}]{bunse:1991numerical}
Bunse-Gerstner, A., R.~Byers, V.~Mehrmann, and N.~K. Nichols (1991).
\newblock Numerical computation of an analytic singular value decomposition of
  a matrix valued function.
\newblock {\em Numerische Mathematik\/}~{\em 60\/}(1), 1--39.

\bibitem[\protect\citeauthoryear{Dahlhaus}{Dahlhaus}{1997}]{dahlhaus97}
Dahlhaus, R. (1997).
\newblock Fitting time series models to nonstationary processes.
\newblock {\em The Annals of Statistics\/}~{\em 25\/}(1), 1--37.

\bibitem[\protect\citeauthoryear{Dahlhaus}{Dahlhaus}{2012}]{dahlhaus12}
Dahlhaus, R. (2012).
\newblock 13 - locally stationary processes.
\newblock In S.~S.~R. Tata Subba~Rao and C.~Rao (Eds.), {\em Time Series
  Analysis: Methods and Applications}, Volume~30 of {\em Handbook of
  Statistics}, pp.\  351 -- 413. Elsevier.

\bibitem[\protect\citeauthoryear{de~Jong}{de~Jong}{1987}]{dejong:1987central}
de~Jong, P. (1987).
\newblock A central limit theorem for generalized quadratic forms.
\newblock {\em Probability Theory and Related Fields\/}~{\em 75\/}(2),
  261--277.

\bibitem[\protect\citeauthoryear{Donald, Fortuna, and Pipiras}{Donald
  et~al.}{2007}]{donald:2007rank}
Donald, S.~G., N.~Fortuna, and V.~Pipiras (2007).
\newblock On rank estimation in symmetric matrices: the case of indefinite
  matrix estimators.
\newblock {\em Econometric Theory\/}~{\em 23\/}(6), 1217--1232.

\bibitem[\protect\citeauthoryear{Donald, Fortuna, and Pipiras}{Donald
  et~al.}{2011}]{donald:2011local}
Donald, S.~G., N.~Fortuna, and V.~Pipiras (2011).
\newblock Local and global rank tests for multivariate varying-coefficient
  models.
\newblock {\em Journal of Business \& Economic Statistics\/}~{\em 29\/}(2),
  295--306.
\newblock Supplementary technical appendix available at
  http://pipiras.web.unc.edu/papers/.

\bibitem[\protect\citeauthoryear{D\"uker, Pipiras, and Sundararajan}{D\"uker
  et~al.}{2019}]{ssa-means}
D\"uker, M., V.~Pipiras, and R.~Sundararajan (2019).
\newblock Cotrending: testing for common deterministic trends in varying means
  model.
\newblock {\em Preprint\/}.

\bibitem[\protect\citeauthoryear{Fortuna}{Fortuna}{2008}]{fortuna:2008local}
Fortuna, N. (2008).
\newblock Local rank tests in a multivariate nonparametric relationship.
\newblock {\em Journal of Econometrics\/}~{\em 142\/}(1), 162--182.

\bibitem[\protect\citeauthoryear{Golub and Van~Loan}{Golub and
  Van~Loan}{2012}]{golub:2012matrix}
Golub, G.~H. and C.~F. Van~Loan (2012).
\newblock {\em Matrix Computations}.
\newblock JHU Press.

\bibitem[\protect\citeauthoryear{Kato}{Kato}{1976}]{kato:1976}
Kato, T. (1976).
\newblock {\em Perturbation Theory for Linear Operators\/} (Second ed.).
\newblock Berlin: Springer-Verlag.
\newblock Grundlehren der Mathematischen Wissenschaften, Band 132.

\bibitem[\protect\citeauthoryear{Lotte and Guan}{Lotte and
  Guan}{2011}]{lotte11}
Lotte, F. and C.~T. Guan (2011).
\newblock Regularizing common spatial patterns to improve bci designs: Unified
  theory and new algorithms.
\newblock {\em IEEE Transactions on Biomedical Engineering\/}~{\em 58\/}(2),
  355--362.

\bibitem[\protect\citeauthoryear{Magnus and Neudecker}{Magnus and
  Neudecker}{1999}]{magnus:neudecker:1999}
Magnus, J. and H.~Neudecker (1999).
\newblock {\em Matrix Differential Calculus with Applications in Statistics and
  Econometrics\/} (Second ed.).
\newblock Wiley Series in Probability and Statistics. Wiley.

\bibitem[\protect\citeauthoryear{Naeem, Brunner, Leeb, Graimann, and
  Pfurtscheller}{Naeem et~al.}{2006}]{bci4data}
Naeem, M., C.~Brunner, R.~Leeb, B.~Graimann, and G.~Pfurtscheller (2006).
\newblock Seperability of four-class motor imagery data using independent
  components analysis.
\newblock {\em Journal of Neural Engineering\/}~{\em 3\/}(3), 208.

\bibitem[\protect\citeauthoryear{Newey}{Newey}{1994}]{newey:1994kernel}
Newey, W.~K. (1994).
\newblock Kernel estimation of partial means and a general variance estimator.
\newblock {\em Econometric Theory\/}~{\em 10\/}(2), 1--21.

\bibitem[\protect\citeauthoryear{Ombao, von Sachs, and Guo}{Ombao
  et~al.}{2005}]{slex_ombao}
Ombao, H., R.~von Sachs, and W.~Guo (2005).
\newblock {SLEX} analysis of multivariate nonstationary time series.
\newblock {\em Journal of the American Statistical Association\/}~{\em
  100\/}(470), 519--531.

\bibitem[\protect\citeauthoryear{Pons and Latapy}{Pons and
  Latapy}{2005}]{pons_walktrap}
Pons, P. and M.~Latapy (2005).
\newblock Computing communities in large networks using random walks.
\newblock In p.~Yolum, T.~G{\"u}ng{\"o}r, F.~G{\"u}rgen, and C.~{\"O}zturan
  (Eds.), {\em Computer and Information Sciences - ISCIS 2005}, Berlin,
  Heidelberg, pp.\  284--293. Springer Berlin Heidelberg.

\bibitem[\protect\citeauthoryear{Powell, Stock, and Stoker}{Powell
  et~al.}{1989}]{powell:1989semiparametric}
Powell, J.~L., J.~H. Stock, and T.~M. Stoker (1989).
\newblock Semiparametric estimation of index coefficients.
\newblock {\em Econometrica\/}, 1403--1430.

\bibitem[\protect\citeauthoryear{Robin and Smith}{Robin and
  Smith}{2000}]{robin:2000tests}
Robin, J.-M. and R.~J. Smith (2000).
\newblock Tests of rank.
\newblock {\em Econometric Theory\/}~{\em 16\/}(2), 151--175.

\bibitem[\protect\citeauthoryear{Sundararajan, Pipiras, and
  Pourahmadi}{Sundararajan et~al.}{2019}]{vc-appendix}
Sundararajan, R., V.~Pipiras, and M.~Pourahmadi (2019).
\newblock Supplementary {M}aterial: ``{S}tationary subspace analysis of
  nonstationary covariance processes: eigenstructure description and testing''.
\newblock {A}vailable at http://pipiras.web.unc.edu/papers/.

\bibitem[\protect\citeauthoryear{Sundararajan, Palma, and
  Pourahmadi}{Sundararajan et~al.}{2017}]{sundararajan:neuro}
Sundararajan, R.~R., M.~A. Palma, and M.~Pourahmadi (2017).
\newblock Reducing brain signal noise in the prediction of economic choices: A
  case study in neuroeconomics.
\newblock {\em Frontiers in Neuroscience\/}~{\em 11}, 704.

\bibitem[\protect\citeauthoryear{Sundararajan and Pourahmadi}{Sundararajan and
  Pourahmadi}{2018}]{sundararajan:2017}
Sundararajan, R.~R. and M.~Pourahmadi (2018).
\newblock Stationary subspace analysis of nonstationary processes.
\newblock {\em Journal of Time Series Analysis\/}~{\em 39\/}(3), 338--355.

\bibitem[\protect\citeauthoryear{von Bahr and Esseen}{von Bahr and
  Esseen}{1965}]{von:1965inequalities}
von Bahr, B. and C.-G. Esseen (1965, 02).
\newblock Inequalities for the $r^{th}$ absolute moment of a sum of random
  variables, $1 \leq r \leq 2$.
\newblock {\em The Annals of Mathematical Statististics\/}~{\em 36\/}(1),
  299--303.

\bibitem[\protect\citeauthoryear{von B\"unau, Meinecke, Kir\'aly, and
  M\"uller}{von B\"unau et~al.}{2009}]{ssa09}
von B\"unau, P., F.~C. Meinecke, F.~C. Kir\'aly, and K.-R. M\"uller (2009,
  Nov).
\newblock Finding stationary subspaces in multivariate time series.
\newblock {\em Physical {R}eview {L}etters\/}~{\em 103}, 214101.

\end{thebibliography}

\end{document}